%% file: iclr2027_conference.tex
\documentclass{article}
\usepackage{iclr2027_conference,times}
\usepackage{style/macros}
\usepackage{amsthm}
\usepackage{graphicx}
\usepackage{placeins}
\usepackage{xcolor}
\usepackage{comment}
\usepackage{titletoc}
\usepackage{hyperref}
\usepackage{url}

\newtheorem{theorem}{Theorem}
\newtheorem{result}[theorem]{Result}

\newtheorem{proposition}[theorem]{Proposition}

\newtheorem{assumption}[theorem]{Assumption}

\newif\ifpreprint
\preprinttrue % Set to \preprintfalse for the anonymous ICLR submission.
\ifpreprint
  \iclrfinalcopy
\fi

\title{Asymptotic Analysis of Empirical Risk Minimization on Entry-wise i.i.d.\ Heavy-Tailed Data}
\author{Kaito Takanami$^{1}$\thanks{Corresponding author: \texttt{takanami255@g.ecc.u-tokyo.ac.jp}}\hspace{0.5em}, Takashi Takahashi$^{2}$, Yoshiyuki Kabashima$^{2,3}$\\
{\normalfont $^{1}$\,Graduate School of Science, The University of Tokyo}\\
{\normalfont $^{2}$\, Institute for Physics of Intelligence, 
The University of Tokyo
}\\
{\normalfont $^{3}$\,Trans-Scale Quantum Science Institute, The University of Tokyo}
}
\begin{document}

\maketitle
\ifpreprint
  \fancyhead[L]{Preprint}
\fi

\begin{abstract}
    Many real-world datasets exhibit unusually large values far more frequently than predicted by Gaussian models. 
    Heavy-tailed distributions capture this behavior, yet evaluating learning performance under them remains challenging because rare, large feature entries retain non-vanishing effects even in high dimensions.
    Even in the canonical setting of empirical risk minimization for linear regression with entry-wise i.i.d.\@ symmetric $\alpha$-stable data, a precise asymptotic characterization of prediction has been lacking. 
    In this work, we introduce a functional order parameter that describes the random effective problem associated with each coefficient.
    Using the replica method, we fully characterize the generalization error in the proportional high-dimensional limit where the sample size and feature dimension diverge at a fixed ratio. 
    Additionally, this analysis establishes a heavy-tail universality law, scaling laws relating typical errors to prediction reliability, and the Bayes-optimal prediction error.
    In addition to characterizing the effects of extreme entries on the learning process, our method applies broadly to other systems with persistent local heterogeneity.
\end{abstract}

\input{section/introduction}

\input{section/related_work}

\input{section/model}
\input{section/main_result}
\input{section/consequences}

\input{section/discussion}

\subsection*{AI use statement}

Generative AI tools assisted with manuscript drafting, editing, translation, paper organization, numerical code, and the exploration and checking of mathematical derivations.
All material was reviewed and verified by the
authors, who take responsibility for the final manuscript.

\subsection*{Reproducibility statement}

Detailed derivations are provided in the appendices.
\ifpreprint
Code for reproducing the numerical experiments is available at \url{https://github.com/taka255/iid-heavy-tail-reproduction}.
\else
Code for reproducing the numerical experiments is included in the supplementary material.
\fi

\ifpreprint
\subsection*{Acknowledgments}
K.T. was supported by JST BOOST NAIS Grant Number JPMJBS2418.
T.T. and Y.K. were supported by JSPS KAKENHI Grant Number 26K02981.
T.T. was supported by JST ACT-X Grant Number JPMJAX24CG and JSPS KAKENHI Grant Number 23K16960.
Y.K. was supported by JSPS KAKENHI Grant Numbers 22H05117 and 26K02981, and JST, ARiSE, Japan Grant Number JPMJAR2633.
\fi

\bibliographystyle{iclr2027_conference}
\bibliography{ref}

\clearpage
\appendix
\startcontents[appendix]
\begin{center}
  {\LARGE\bfseries Appendix\par}
\end{center}
\vspace{1em}
\noindent\textbf{Contents}\par
\smallskip
\printcontents[appendix]{}{1}{\setcounter{tocdepth}{2}}
\bigskip
\input{section/additional_related_work}
\input{section/self_consistent_equations}
\input{section/replica_derivation}
\input{section/universality_proof}
\input{section/data_rich_scaling_proof}
\input{section/bayes_benchmark_proof}

\end{document}

%% file: section/introduction.tex
\section{Introduction}

Real-world data in finance~\citep{d5c57a75-fb19-3afe-9a07-833e74744f04}, communications~\citep{678475, 4091042}, natural images~\citep{NIPS1999_6a5dfac4}, and the life sciences~\citep{PhysRevLett.90.088102} often contain outliers and extreme observations much more frequently than predicted by a Gaussian distribution.
When even the feature variance does not exist, the Gaussian universality suggested by the usual central limit theorem may no longer hold, and tail behavior can have a non-vanishing effect on learning even in the high-dimensional limit.
Heavy-tailedness therefore not only poses a robustness issue, but can fundamentally change the statistical theory of high-dimensional learning.

The theory of learning from heavy-tailed data has mainly developed in two directions.
The first is non-asymptotic statistics, which provides error guarantees under weak moment assumptions~\citep{Brownlees_2015, pmlr-v35-mendelson14, mendelson2018learning, lecué2017regularizationsmallballmethodi, lecué2017robustmachinelearningmedianofmeans}.
These results apply to a broad class of data distributions, including models with independent heavy-tailed coordinates.
However, they mainly provide worst-case error bounds and do not precisely characterize typical high-dimensional errors or their dependence on the loss, regularization strength, and tail index.

The second direction studies typical high-dimensional behavior in the proportional limit, where the dimension \(D\) and sample size \(N\) grow at the same rate.
For heavy-tailed features, existing precise results rely almost exclusively on elliptical models $\boldsymbol{x}_\mu=\sigma_\mu \boldsymbol{z}_\mu$ with $\boldsymbol{z}_\mu\sim\mathcal N(\boldsymbol 0, \boldsymbol{I}/D)$, where a single random variable $\sigma_\mu$ scales all coordinates of sample $\mu$~\citep{NEURIPS2023_88be0230, Adomaityte_2024}.
However, this reliance on elliptical symmetry has two major limitations:
%First, from a theoretical standpoint, the elliptical models considered in prior work can be viewed as Gaussian scale mixtures; conditioned on the sample scale $\sigma_\mu$, the features remain Gaussian and affect all coordinates uniformly. 
First, from a theoretical standpoint, prior elliptical models introduce heavy-tailedness through a single random scale $\sigma_\mu$ per sample. Conditioning on this scale restores a Gaussian feature vector, making these models a more restricted class of high-dimensional heavy-tailed distributions.
In contrast, when coordinates are independently heavy-tailed, extreme values appear in isolated entries, so existing analytical techniques inapplicable.
Second, in high-dimensional sensing and measurement systems, including communication, signal detection, and imaging applications, localized interference and physical events can generate extreme observations concentrated in a small number of coordinates~\citep{karakus2020impulsive,chopra2012joint,shen2016cosmic}.
For such applications, an entry-wise heavy-tailed model is far more realistic than an elliptical one that scales all features through a shared random factor.

To bridge this gap, we first study empirical risk minimization (ERM) with coordinate-wise i.i.d.\@ symmetric $\alpha$-stable features ($0<\alpha<2$) in the proportional high-dimensional limit, where $N,D\to\infty$ and $N/D\to\kappa\in(0,\infty)$.
The corresponding Gaussian-feature problem is a canonical setting in high-dimensional statistics, and its typical performance has been studied for many years~\citep{PhysRevLett.65.1683, PhysRevA.45.6056,RevModPhys.65.499, pnas1307842110, donoho2016high, 8365826}.
Yet, after replacing Gaussian features with independent heavy-tailed features, no general asymptotic theory is available for this canonical problem.
We therefore ask how the typical performance of the learned estimator can be characterized in this setting.

We use the replica method~\citep{mezard1987spin, charbonneau2023spin} from statistical physics to reduce the problem to an effective one-dimensional optimization.
For Gaussian features, as in many standard replica calculations, this effective problem is described by a finite set of scalar order parameters.
%For coordinate-wise independent heavy-tailed features, the effective objective remains random from one coordinate to another, so the joint distribution $\hat{\mathcal P}$ of its random linear and nonlinear terms is required.
By contrast, for coordinate-wise independent heavy-tailed features, a single extreme entry can preserve the full nonlinear response of the loss in the effective problem.
The samples producing such extremes differ across coordinates, and this heterogeneity persists in high dimensions.
The system therefore requires a \textit{functional} order parameter $\hat{\mathcal P}$, defined as the distribution of these random response functions.
We derive a self-consistent equation for this functional order parameter and use it to compute the learned coefficient distribution and typical estimation error.

Our main results are as follows.
\begin{itemize}
  \item Section~\ref{sec:precise-asymptotics} characterizes the high-dimensional ERM through a functional order parameter, reducing it to a random one-dimensional effective optimization problem and determining its asymptotic generalization error.
  \item Section~\ref{sec:universality} establishes heavy-tail universality: the learned behavior depends on the feature distribution only through its tail index and total tail constant.
  \item Section~\ref{sec:data-rich-scaling} shows that typical prediction errors decrease as $\kappa^{-1/\alpha}$, faster than the Gaussian $\kappa^{-1/2}$ rate, while the probability of a fixed large error decreases only as $\kappa^{-1}$ rather than exponentially.  
  \item Section~\ref{sec:bayes-benchmark} characterizes Bayes-optimal prediction performance and shows that entry-wise heavy-tailed features outperform their elliptical counterparts at finite $\kappa$, although both exhibit the same leading-order scaling for $\kappa \gg 1$.
\end{itemize}

%% file: section/related_work.tex
\section{Related Work}
\label{sec:related-work}

\paragraph{Analysis with functional order parameters.}
Heterogeneous systems with sparse or heavy-tailed interactions remain locally nonuniform even in the high-dimensional limit, and their fluctuations cannot generally be described by finitely many scalar order parameters
\citep{mezardParisi2001Bethe,Franz_2001,neriMetzBolle2010Levy,jung2018LevyKhintchine}.
Discrete-variable models can often be described by distributions of local effective fields, whereas the unbounded continuous variables and general nonlinear convex losses considered here make the local effective problems themselves random functions.
Moreover, heavy-tailed features produce a few strong contributions together with infinitely many small ones, making the order parameters probability distributions on function spaces.
We derive the corresponding saddle-point equations, and the same approach may extend to other general heterogeneous continuous-variable systems.
Further details are provided in Appendix~\ref{sec:additional-related-work}.

\paragraph{Learning from heavy-tailed data.}
Learning from entry-wise i.i.d.\@ heavy-tailed data with squared-loss ridge regression reduces to a resolvent calculation for a heavy-tailed sample covariance matrix.
The spectral distributions and resolvents of such matrices have been studied~\citep{belinschiDemboGuionnet2009Heavy,jung2018LevyKhintchine}.
In this quadratic setting, the local effective problems can be represented by finitely many coefficients, whereas general nonlinear convex losses require an order parameter that specifies the distribution of the entire random function.
Our analysis extends this quadratic description to general nonlinear convex losses and characterizes the limiting estimator distribution and prediction error.

%% file: section/model.tex
\section{Notation and Problem Setup}
\label{sec:model}

We study high-dimensional linear regression with coordinate-wise independent symmetric $\alpha$-stable features using regularized ERM.

\subsection{Notation}

We denote by $S\alpha S(\sigma)$ the symmetric $\alpha$-stable distribution with scale parameter $\sigma$, whose characteristic function is
$\E \exp(itX)=\exp\ab(-\sigma^\alpha |t|^\alpha)$ where $X\sim S\alpha S(\sigma)$.
For a random vector $\vb{x}=(x_1,\ldots,x_D)^\top$, we write $\vb{x}\sim P^{\otimes D}$ when its coordinates are independently distributed according to $P$.
We use $\xrightarrow{\mathrm{p}}$ and $\xrightarrow{\mathrm{d}}$ to denote convergence in probability and convergence in distribution, respectively.
At a nondifferentiable point, $f'(x)$ denotes a subgradient of $f$ at $x$, and $\partial f(x)$ denotes the subdifferential of $f$ at $x$.

\subsection{Problem setup}

We consider linear regression with $N$ training examples in $D$ dimensions.
For $0<\alpha<2$, the training inputs are generated as
\begin{align}
  \vb{x}_1,\ldots,\vb{x}_N
  &\overset{\mathrm{i.i.d.}}{\sim}
  \ab(S\alpha S(1)\ab)^{\otimes D}.
  \label{eq:stable-input-distribution}
\end{align}
Independently of the inputs, the teacher and label noise satisfy $\vb{w}^*\sim p_w^{\otimes D}$ and $\epsilon_\mu\overset{\mathrm{i.i.d.}}{\sim}p_\epsilon$,
where both distributions are centered, nondegenerate, and sub-Gaussian,
and $p_\epsilon$ additionally has a bounded density.
The labels are
\begin{align}
  y_\mu
  &= \frac{1}{D^{1/\alpha}}\vb{x}_\mu^\top\vb{w}^*+\epsilon_\mu,
  \label{eq:label-model}
\end{align}
where $D^{-1/\alpha}$ keeps the linear predictor of order one.
Given a convex loss $\ell$, a strongly convex penalty function $\gamma$, and $\lambda>0$, we define the regularized ERM estimator by
\begin{align}
  \hat{\vb{w}}
  &:= \argmin_{\vb{w}\in\R^D}
  \ab\{
    \sum_{\mu=1}^N
    \ell\ab(
      \frac{1}{D^{1/\alpha}}\vb{x}_\mu^\top\vb{w}-y_\mu
    )
    +\lambda\sum_{i=1}^D\gamma\ab(w_i)
  \ab\}.
  \label{eq:training-objective}
\end{align}
To evaluate the learned estimator, we draw an independent test example $(\vb{x}_0,y_0)$ from the same model. Its label and prediction are
\begin{align}
  y_0
  =
  \frac{1}{D^{1/\alpha}}\vb{x}_0^\top\vb{w}^*
  +\epsilon_0,
  \qquad
  \hat y_0
  =
  \frac{1}{D^{1/\alpha}}\vb{x}_0^\top\hat{\vb w}.
\end{align}
For a test loss $\ell_{\mathrm{test}}$, we define the noise-free test prediction error and the generalization error by
\begin{align}
  \Delta_{N,D}^{(\alpha)}
  &:=
  \hat y_0-(y_0-\epsilon_0)
  =
  \frac{1}{D^{1/\alpha}}
  \vb{x}_0^\top(\hat{\vb w}-\vb w^*)
  \label{eq:noise-free-pred^error}
  \\
  \mathcal E_{N,D}^{(\alpha)}
  &:=
  \E\left[
    \ell_{\mathrm{test}}(\hat y_0-y_0)
  \right],
  \label{eq:generalization-error}
\end{align}
where the expectation is over the test example.
The test noise is excluded from $\Delta_{N,D}^{(\alpha)}$ but included in $\mathcal E_{N,D}^{(\alpha)}$.
We study these quantities in the proportional high-dimensional limit:
\begin{align}
  N,D\to\infty,
  \qquad
  \frac{N}{D} \to\kappa\in(0,\infty).
  \label{eq:proportional_limit}
\end{align}

\subsection{Regularity conditions}

We impose the following assumptions to guarantee a unique ERM minimizer and to ensure that the population loss has a well-defined minimum at zero residual.

\begin{assumption}[Training loss]
\label{asm:loss-regularity}
The training loss $\ell:\R\to\R$ is convex and is either globally Lipschitz or the squared loss $\ell(r)=r^2/2$.
Assume
\begin{align}
  \operatorname*{argmin}_{r\in\R}
  \E_\epsilon\ell\ab(r-\epsilon)
  &=
  \{0\},
  \quad
  0
  <
  \left.
  \frac{\mathrm d^2}{\mathrm d r^2}
  \E_\epsilon\ell\ab(r-\epsilon)
  \right|_{r=0}
  <\infty.
  \label{eq:noise-averaged-training-loss}
\end{align}
\end{assumption}

This assumption imposes the local strict convexity required for stable parameter estimation.
Consequently, this assumption accommodates a wide range of standard objectives, including the squared loss, Huber and pseudo-Huber losses, and $L_p$ losses ($1 \le p \le 2$).

%% file: section/main_result.tex
\section{Precise Asymptotic Characterization}
\label{sec:precise-asymptotics}

Our goal is to characterize the generalization error of ERM estimator in the proportional high-dimensional limit.
The key observation is that prediction on an independent test example depends on the learned coefficient errors only through a single scalar.
Determining this scalar from the training problem, however, requires a function-valued order parameter.

Define the empirical $\alpha$-error of the learned coefficients by
\begin{align}
  q_{\alpha,D}
  &:=%
  \frac{1}{D}
  \sum_{i=1}^D
  \abs{\hat w_i-w_i^*}^{\alpha}.
  \label{eq:empirical-alpha-error-main}
\end{align}
%In the high-dimensional limit, $q_{\alpha,D}$ becomes self-averaging, i.e.,  $q_{\alpha, D} \xrightarrow{\mathrm{p}} {\mathbb E}_{D_{\rm tr},\bm{x}_0, \epsilon_0}[q_{\alpha, D}]=:q_{\alpha}$. 
In the proportional high-dimensional limit, there exists a deterministic $q_\alpha$ such that $q_{\alpha,D}\xrightarrow{\mathrm p}q_\alpha$.
In this limit, by the stability property of symmetric \(\alpha\)-stable distributions, the noise-free test prediction error in
Equation~\eqref{eq:noise-free-pred^error} satisfies
\begin{align}
  \Delta_{N,D}^{(\alpha)}
    \xrightarrow{\mathrm d}
  \Delta^{(\alpha)}\ab(\kappa)
  &:=
  q_\alpha^{1/\alpha}S_\alpha,
  \label{eq:limiting-noise-free-prediction-error}
\end{align}
where $S_\alpha \sim S\alpha S(1)$.
This yields the following expression for the asymptotic generalization error.

\begin{proposition}[Generalization error of ERM estimator]
  \label{prop:test-prediction-generalization}
  Let $0<\alpha<2$ and $S_\alpha\sim S\alpha S(1)$.
  Then, the asymptotic generalization error \eqref{eq:generalization-error} is given by
  \begin{align}
    \mathcal E_{N,D}^{(\alpha)}
    \xrightarrow{\mathrm p}
    \mathcal E^{(\alpha)}(\kappa)
    :=
    \E_{S_\alpha,\epsilon_0}
    \ab[
      \ell_{\mathrm{test}}
      \ab(
        q_\alpha^{1/\alpha}S_\alpha-\epsilon_0
      )
    \ab].
    \label{eq:limiting-stable-generalization-error}
  \end{align}
\end{proposition}

Proposition~\ref{prop:test-prediction-generalization} shows that $q_\alpha$ is the only scalar needed to evaluate the asymptotic generalization error.
To compute this scalar, we need the distribution of the coefficient errors $\hat w_i-w_i^*$.

For Gaussian features, the contribution of each data point to an individual coefficient is infinitesimally small, so the collective effect of the data becomes statistically uniform across all coefficients.
As is well known, their combined effect, and consequently the coefficient distribution and $q_\alpha$, is fully determined by a finite set of scalar equations.

However, for heavy-tailed features, a small number of exceptionally large entries retain a finite influence.
Which entries are extreme and how strongly they affect the loss differ across coefficients, so their effects do not reduce to common scalar values.
The coefficient-specific effective problem must therefore remain random.
It consists of a random scale $\hat Q  \in[0,\infty)$ and a standardized Gaussian fluctuation $\zeta$, which together determine the linear \textit{data fluctuation}, and a random convex \textit{data response} function $\hat{\mathcal R}:\R\to\R$ that quantifies the nonlinear loss change induced by a perturbation of the coefficient.
We denote their joint distribution by the functional order parameter $\hat{\mathcal P}$.
Under the model in Section~\ref{sec:model} and Assumption~\ref{asm:loss-regularity}, with $0<\alpha<2$, draw
\begin{align}
  (\hat Q,\zeta,\hat{\mathcal R})
  \sim\hat{\mathcal P},
  \quad
  w^*
  \sim p_w.
  \label{eq:input-functional-order-parameter-sampling}
\end{align}
Here $w^*$ is independent of the first triple, $\hat Q$ is a nonnegative $\alpha/2$-stable random variable, $\zeta\sim\mathcal N(0,1)$ is independent of $\hat Q$, and $\hat{\mathcal R}$ is convex with $\hat{\mathcal R}(0)=0$ and $0\in\partial\hat{\mathcal R}(0)$.
Although $\zeta$ is independent of $\hat Q$, it is generally dependent on $\hat{\mathcal R}$.
Its full definition is given in Appendix~\ref{app:self-consistent-equations}.

For each sample of the effective surrogate randomness, define $\hat u$ as the
minimizer of an effective objective consisting of the change in regularization, a random linear fluctuation, and the nonlinear response
of the training loss:
\begin{align}
  \hat u
  &:=
  \argmin_{u\in\R}
  \ab\{
    \underbrace{
      \lambda\ab[\gamma\ab(w^*+u)-\gamma\ab(w^*)]
    }_{\text{regularization}}
    +
    \underbrace{
      \sqrt{\hat Q}\,\zeta u
    }_{\text{data fluctuation}}
    +
    \underbrace{
      \hat{\mathcal R}\ab(u)
    }_{\text{data response}}
  \ab\}.
  \label{eq:heavy-tail-parameter-statistics}
\end{align}
The $\hat u$ constructed in this way has the same limiting distribution as the learned coefficient error $\hat w_i-w_i^*$.
The corresponding statement is the following result.

\begin{result}[Statistics of the learned parameters]
  \label{res:learned-parameter-statistics}
  With $\E\abs{\hat u}^{\alpha}<\infty$, the empirical coefficient distribution and its $\alpha$-moment satisfy
  \begin{align}
    &\hat w_i-w_i^*
    \xrightarrow{\mathrm d}
    \hat u \label{eq:learned-empirical-law-main}
    \\
    &q_{\alpha,D}
    \xrightarrow{\mathrm p}
    q_\alpha
    :=
    \E\abs{\hat u}^{\alpha},
    \label{eq:learned-empirical-law-main_q}
  \end{align}
  where the expectation is over the joint distribution of $\zeta$, $\hat{\mathcal R}$, $\hat Q$, and $w^*$.
\end{result}

The detailed derivation is provided in Appendix~\ref{app:replica-derivation}.
Figure~\ref{fig:rs-erm-coefficient-comparison} confirms that both the predicted $\alpha$-moment and the full coefficient-error distribution closely match finite-dimensional ERM simulations.
Thus, $\hat{\mathcal P}$ determines both the full limiting coefficient-error distribution and its scalar summary $q_\alpha$.
Combined with Equation~\eqref{eq:limiting-stable-generalization-error}, this result provides a precise characterization of the generalization error on test data.

The decomposition in Equation~\eqref{eq:heavy-tail-parameter-statistics} further clarifies how heavy-tailed features affect each learned coefficient.
The regularization term penalizes coefficient values according to $\gamma$, while the linear data-fluctuation term generates a random gradient that can cause the coefficient $\hat w_i$ to deviate from its true value $w_i^\ast$, 
thereby increasing the estimation error $\lvert \hat w_i-w_i^\ast\rvert$.
The nonlinear data-response term counteracts this deviation and tends to keep the coefficient error near zero.
Because the latter two terms are generated from the same heavy-tailed samples, they are statistically dependent, and their magnitudes vary across coordinates.
Consequently, an extreme data fluctuation can produce a large linear term.
If $\hat{\mathcal R}(u)$ increases sufficiently away from $u=0$, the minimizer $\hat u$ remains close to zero; otherwise, the coefficient error can be large and increase the generalization error.

\begin{figure*}[t]
  \centering
  \includegraphics[width=\textwidth]{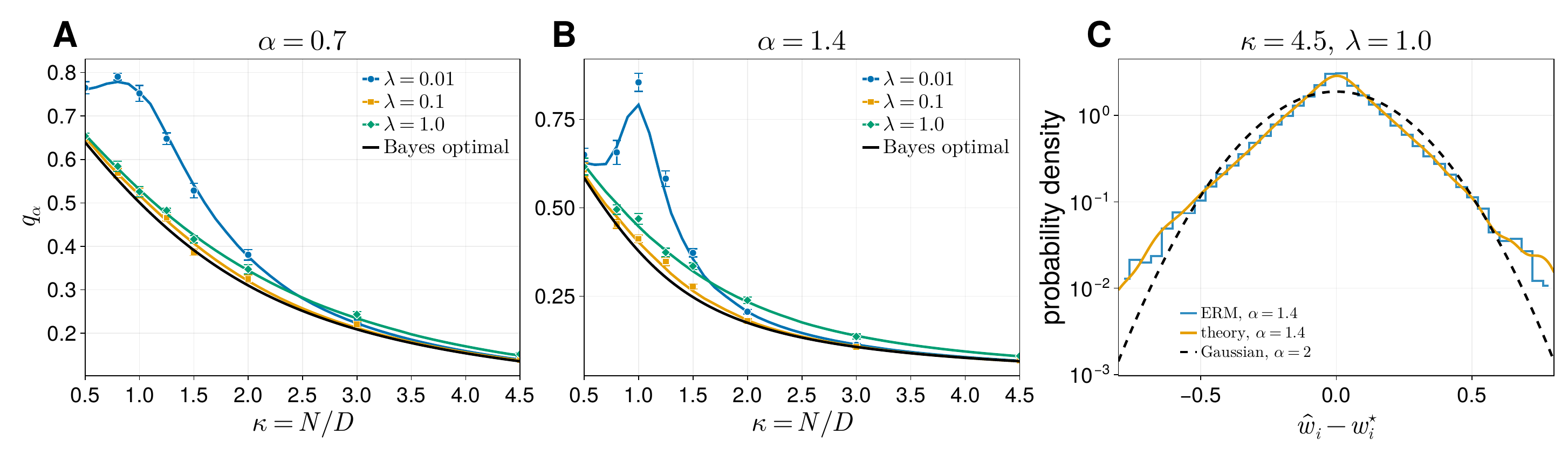}
  \caption{
    \textbf{Empirical validation of the theoretical coefficient-error predictions.}
    (A, B) Solid colored lines show the theoretical predictions of $q_\alpha$ (Equation~\ref{eq:learned-empirical-law-main_q}), while markers represent numerical
    results obtained via finite-dimensional ERM. 
    The solid black line shows the theoretical prediction of $q_\alpha$ under the Bayes--optimal setting in Section~\ref{sec:bayes-benchmark}. 
    (C) Comparison between the theoretical distribution (Equation~\ref{eq:learned-empirical-law-main}) and the empirical
    coefficient-error distribution obtained by finite-dimensional ERM.
    The dashed black curve shows the theoretical coefficient-error distribution at $\alpha=2$, which is Gaussian.
    (A-C) Huber loss with $\delta=1$, $p_w=\mathcal N(0,1)$, and $p_\epsilon=\mathcal N(0,0.5^2)$;
    (A,B) $D=256$; (C) $D=4096$, $\lambda=1$, and the empirical distribution is averaged over 8 independent trials.
    (A,B) Error bars indicate the standard error over 8 independent trials for each point. 
  }
  \label{fig:rs-erm-coefficient-comparison}
\end{figure*}

\paragraph{Comparison with Gaussian features.}
When $\alpha=2$, the Gaussian counterpart of the heavy-tailed coefficient problem in Equation~\eqref{eq:heavy-tail-parameter-statistics} is
  \begin{align}
  \hat u_{\mathrm G}
  &:=
  \argmin_{u\in\R}
  \ab\{
    \underbrace{
      \lambda\ab[\gamma\ab(w^*+u)-\gamma\ab(w^*)]
    }_{\text{regularization}}
    +
    \underbrace{
      \sqrt{\hat q}\,\xi u
    }_{\text{data fluctuation}}
    +
    \underbrace{
      \frac{\hat\chi}{2}u^2
      }_{\text{data response}}
  \ab\},
  \label{eq:gaussian-parameter-statistics-comparison}
  \end{align}
where $\xi\sim\mathcal N(0,1)$ and $\hat{q}$ and $\hat{\chi}$ are some positive constants.
Gaussian features average these effects into a common fluctuation scale $\sqrt{\hat q}$ and a common quadratic data response $\hat\chi u^2/2$.
For heavy-tailed features, by contrast, the relative magnitudes of the linear and nonlinear terms vary randomly across coordinates, so some coefficients have small errors whereas others are poorly estimated within the same model.

%% file: section/consequences.tex
\section{Broader Implications of the Asymptotic Theory}
\label{sec:fixed-point-consequences}

Result~\ref{res:learned-parameter-statistics} in Section~\ref{sec:precise-asymptotics} characterized the asymptotic coefficient distribution and generalization error of regularized ERM with entry-wise i.i.d.\ $\alpha$-stable features.
Using this characterization, we next investigate three questions that have been extensively studied for Gaussian features: universality, the scaling of the typical and extreme error, and the performance attainable by Bayes-optimal inference.

\input{section/universality}
\input{section/data_rich_scaling}

\input{section/bayes_benchmark}

%% file: section/universality.tex
\subsection{Universality law}
\label{sec:universality}

A central question in high-dimensional learning is which properties of the feature distribution determine the asymptotic performance of trained models.
For finite-variance features, Gaussian universality states that, under suitable regularity conditions, replacing the feature distribution with any other distribution matching the same mean and covariance leaves the asymptotic generalization error unchanged.
The following proposition formalizes this classical result.

\begin{proposition}[Gaussian universality for regularized ERM]
  \label{prop:gaussian-universality}
  Let the feature coordinates be i.i.d. from a sub-Gaussian distribution $P_X$ satisfying
  $\E_{P_X}X=0$ and $\E_{P_X}X^2=1$, 
  and suppose that Assumption~\ref{asm:loss-regularity} holds with $\alpha=2$.
  Let $q_2^{P_X}$ and $q_2^{\mathrm G}$ denote the limiting coefficient-error
  moments under $P_X$ and standard Gaussian features, respectively.
  Then
  \begin{align}
    q_2^{P_X}
    &=q_2^{\mathrm G}.
    \label{eq:gaussian-coefficient-error-universality}
  \end{align}
\end{proposition}

Thus, within the sub-Gaussian class, the limiting coefficient error, and therefore the generalization error, depends on the feature distribution only through its mean and variance.
While Proposition~\ref{prop:gaussian-universality} was established rigorously by~\citet{montanari2022universality}, an equivalent moment universality was already classically known via non-rigorous replica calculations for finite-variance patterns~\citep{PhysRevA.45.6056,RevModPhys.65.499}.

However, when the feature variance is infinite, this Gaussian universality law no longer applies.
Nevertheless, we claim that an analogous universality principle holds through replica calculations for coordinate-wise i.i.d.\@ heavy-tailed feature distributions.

\begin{result}[Heavy-tail universality for regularized ERM]
  \label{res:heavy-tail-universality}
  Let $0<\alpha<2$, and let the feature coordinates be i.i.d. from a symmetric distribution $P_X$ satisfying
  \begin{align}
    \Pr\ab(\abs{X}>x)
    &=C_\alpha^{\mathrm S}x^{-\alpha}+o\ab(x^{-\alpha})
    \qquad (x\to\infty),
    \label{eq:stable-domain-universality-tail}
  \end{align}
  where $C_\alpha^{\mathrm S} := 2\Gamma(\alpha)\sin(\pi\alpha/2) / \pi$ is the total tail constant of $S\alpha S(1)$
  \footnote{If the original distribution has a total tail constant $C>0$, rescaling $X$ by $(C_\alpha^{\mathrm S}/C)^{1/\alpha}$ recovers this normalization; equivalently, the original distribution is compared against $S\alpha S((C/C_\alpha^{\mathrm S})^{1/\alpha})$.},
  and suppose that Assumption~\ref{asm:loss-regularity} holds.
  Let $q_\alpha^{P_X}$ and $q_\alpha^{\mathrm S}$ denote the limiting
  coefficient-error moments under $P_X$ and $S\alpha S(1)$ features,
  respectively.
  Then
  \begin{align}
    q_\alpha^{P_X}
    &=q_\alpha^{\mathrm S}.
    \label{eq:stable-coefficient-error-universality}
  \end{align}
\end{result}

The detailed derivation is provided in Appendix~\ref{app:universality-proof}.
Result~\ref{res:heavy-tail-universality} demonstrates that, for the broad class of losses under Assumption~\ref{asm:loss-regularity}, the limiting coefficient error and generalization error depend on the feature distribution solely through its tail index $\alpha$ and total tail constant.
This universality principle is numerically confirmed in Figure~\ref{fig:stable-domain-universality}.

\begin{figure*}[t]
  \centering
  \includegraphics[width=0.9\textwidth]{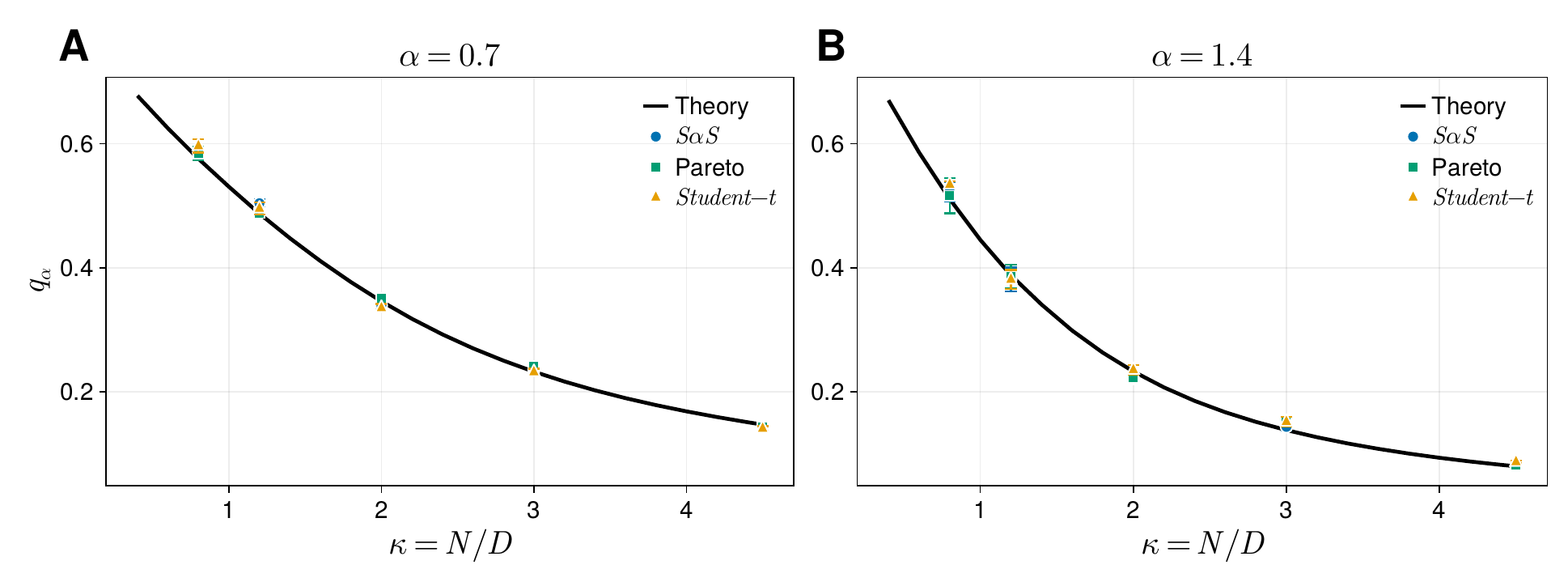}
  \caption{
    \textbf{Empirical validation of universality across heavy-tailed feature distributions.}
    The solid line shows the theoretical prediction of $q_\alpha$ for $S\alpha S(1)$ features, while markers represent numerical results obtained via Huber--ridge ERM with
    symmetric $\alpha$-stable, Pareto, and Student-$t$ coordinates. 
    The Pareto and
    Student-$t$ distributions are rescaled to satisfy
    $\Pr(|X|>x)\sim C_\alpha^{\mathrm S}x^{-\alpha}$.
    Parameters: (A, B) $D=256$, $\delta=1$, $\lambda=1$,
    $p_w=\mathcal N(0,1)$, and $p_\epsilon = \mathcal N(0,0.5^2)$.
    Error bars indicate the standard error over 6 independent trials for each point.
  }
  \label{fig:stable-domain-universality}
\end{figure*}

Together, Proposition~\ref{prop:gaussian-universality} and Result~\ref{res:heavy-tail-universality} establish two distinct universality classes:
the finite-variance regime, represented by the Gaussian distribution and determined by the first two moments; and the infinite-variance regime, represented by the $\alpha$-stable distribution and specified by the tail index and total tail constant.
In both cases, these parameters completely determine the dependence of the limiting coefficient error on the feature distribution.

%% file: section/data_rich_scaling.tex
\subsection{Scaling Laws and Prediction Reliability}
\label{sec:data-rich-scaling}

How does prediction performance improve as the training sample size increases?
In this subsection, we evaluate this improvement as the sampling ratio $\kappa = N/D$ grows, using two key metrics: the typical test prediction error and the probability of an error exceeding a fixed threshold.
According to Proposition~\ref{prop:test-prediction-generalization}, the prediction error induced by the estimator $\hat{\vb{w}}$ converges in distribution to $q_\alpha(\kappa)^{1/\alpha}S_\alpha$.
Using this convergence, we derive the asymptotic expansion of $q_\alpha(\kappa)$ in the regime $\kappa \gg 1$ and obtain the following main results on typical performance and tail error probabilities in comparison with the finite-variance (Gaussian) baseline.

\begin{result}[Typical prediction error and tail reliability with many training examples]
  \label{res:data-rich-tail-reliability}
  Suppose that Assumption~\ref{asm:loss-regularity} holds.
  Let $\Delta^{(2)}(\kappa)$ denote the Gaussian counterpart of
  $\Delta^{(\alpha)}(\kappa)$ in
  Equation~\eqref{eq:noise-free-pred^error}.
  Then, for every $0<\alpha<2$, there exists a constant
  $c_\alpha\in(0,\infty)$ such that
  \begin{align}
    q_\alpha(\kappa)
    =
    \frac{c_\alpha}{\kappa}
    +o(\kappa^{-1}),
    \qquad
    \kappa\to\infty.
    \label{eq:data-rich-coefficient-error}
  \end{align}
  Consequently, for every fixed $t>0$,
  \begin{align}
    \operatorname{median}\ab(
      \abs{\Delta^{(\alpha)}(\kappa)}
    )
    =
    \Theta\ab(\kappa^{-1/\alpha}),
    \quad
    \Pr\ab(
      \abs{\Delta^{(\alpha)}(\kappa)}>t
    )
    =
    \Theta\ab(\kappa^{-1}).
    \label{eq:data-rich-fixed-threshold-failure}
  \end{align}
  For Gaussian features,
  \begin{align}
    \operatorname{median}\ab(
      \abs{\Delta^{(2)}(\kappa)}
    )
    =
    \Theta\ab(\kappa^{-1/2}),
    \quad
    \Pr\ab(
      \abs{\Delta^{(2)}(\kappa)}>t
    )
    =
    \exp[-\Theta(\kappa)],
    \label{eq:data-rich-fixed-threshold-failure-gauss}
  \end{align}
  where $f=\Theta(g)$ means that $f$ and $g$ have the same asymptotic
  order up to constant factors.
\end{result}
The detailed derivation is provided in Appendix~\ref{app:data-rich-scaling-proof}.
Figure~\ref{fig:data-rich-tail-reliability} shows that the predictions of both scaling laws agree remarkably well with finite-dimensional ERM results.

\FloatBarrier
\begin{figure*}[t]
  \centering
  \includegraphics[width=0.9\textwidth]{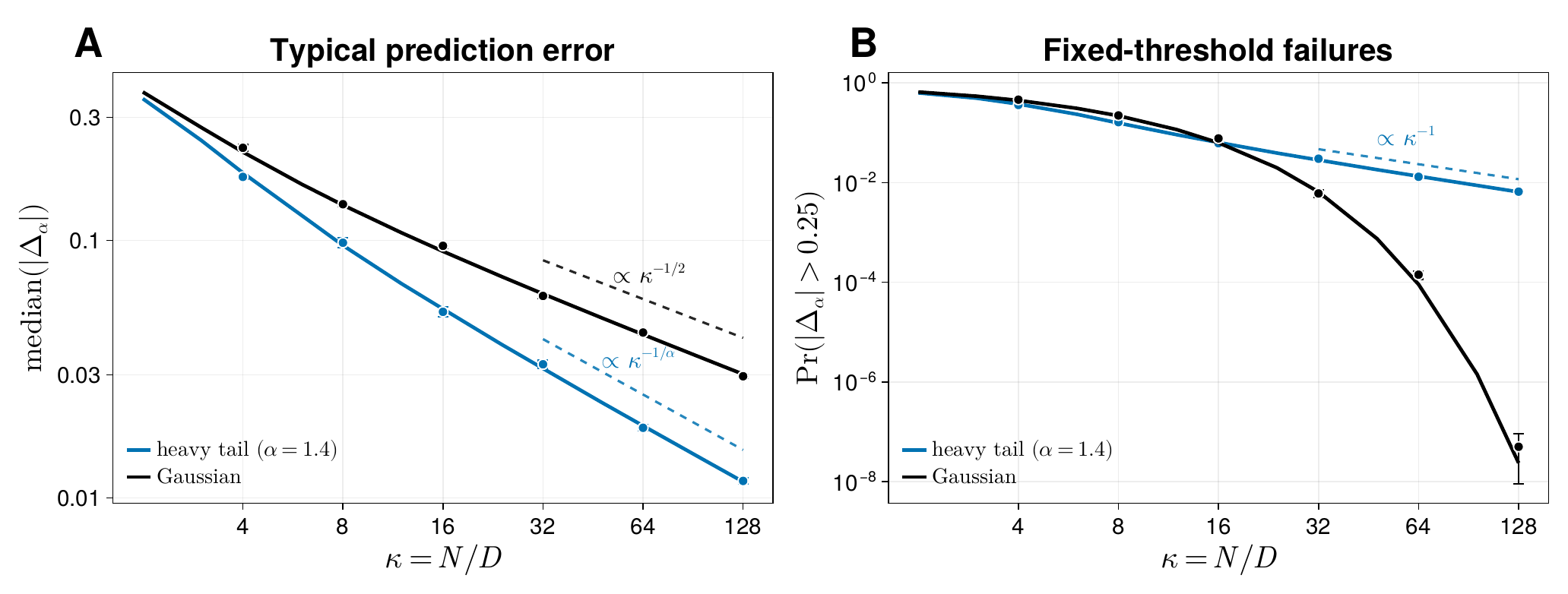}
  \caption{
    \textbf{Empirical validation of prediction-error scaling and reliability.}
    (A) The solid colored and black lines show the theoretical predictions for
    $\operatorname{median}(|\Delta^{(\alpha)}(\kappa)|)$ and
    $\operatorname{median}(|\Delta^{(2)}(\kappa)|)$ under heavy-tailed and Gaussian
    features, respectively, while markers represent the finite-dimensional ERM results.
    (B) The solid colored and black lines show the theoretical predictions for
    $\Pr(|\Delta^{(\alpha)}(\kappa)|>t)$ and
    $\Pr(|\Delta^{(2)}(\kappa)|>t)$ under heavy-tailed and Gaussian features,
    respectively, while markers represent the finite-dimensional ERM results.
    Parameters: (A, B) Huber loss with $\delta=1$, $\lambda=1$,
    $p_w=\mathcal N(0,1)$, $p_\epsilon=\mathcal N(0,0.5^2)$, and $D=128$.
    Error bars indicate the standard error over 8 independent trials for each point.
  }
  \label{fig:data-rich-tail-reliability}
\end{figure*}

This result shows that the typical prediction error decreases faster under heavy-tailed features than under Gaussian features.
Intuitively, rare training observations with large feature magnitudes contain substantial information for estimating the coefficients.
In contrast, the probability of exceeding a fixed error threshold is $\Theta(\kappa^{-1})$ under heavy-tailed features, whereas it is exponentially small under Gaussian features.
This is because the model can exploit the extreme coordinates observed during training, whereas an extreme value in a new test example may occur in a different coordinate and amplify the residual error in the corresponding coefficient.

%% file: section/bayes_benchmark.tex
\subsection{Bayes-Optimal Setting}
\label{sec:bayes-benchmark}

In this subsection, we study prediction under Bayes-optimal inference, where the prior and noise distributions used by the learner match the data-generating model. 
This setting enables a more tractable analytical characterization. 
Moreover, for any specified test loss, the corresponding Bayes risk lower-bounds the risk of any predictor based on the same training data, including more flexible approaches such as kernel methods, ensembles, and deep neural networks.
It thereby helps distinguish limitations due to heavy-tailed features from those due to the choice of inference method.

For simplicity, we assume a Gaussian teacher prior and Gaussian label noise:
\begin{align}
  p_{w}
  = \mathcal N\ab(0,\tau^2),
  \quad
  p_\epsilon
  = \mathcal N\ab(0,\sigma^2).
  \label{eq:bayes-matched-model}
\end{align}
Combining the Gaussian likelihood with this prior gives
\begin{align}
  p\ab(\vb w\mid X,\vb y)
  &\propto
  \exp\ab\{
    -\frac{1}{2\sigma^2}
    \left\lVert
      \vb y-D^{-1/\alpha}X\vb w
    \right\rVert_2^2
    -\frac{1}{2\tau^2}\lVert\vb w\rVert_2^2
  \ab\},
  \label{eq:bayes-posterior-density}
\end{align}
where $X\in\R^{N\times D}$ has rows $\vb x_1^\top,\ldots,\vb x_N^\top$, and $\vb y=(y_1,\ldots,y_N)^\top$ is the label vector.
Completing the square shows that the posterior is Gaussian:
\begin{align}
  \vb w^*\mid X,\vb y
  \sim
  \mathcal N\ab(\hat{\vb w}_{\mathrm B},\Sigma_{\mathrm B}),
  \label{eq:bayes-posterior-law}
\end{align}
where the posterior mean and covariance are
\begin{align}
  \hat{\vb w}_{\mathrm B}
  &:=%
  \E\ab[\vb w^*\mid X,\vb y]
  =
  \ab(
    D^{-2/\alpha}X^\top X
    +\frac{\sigma^2}{\tau^2}I_D
  \ab)^{-1}
  D^{-1/\alpha}X^\top\vb y\\
  \Sigma_{\mathrm B}
  &:=%
  \sigma^2
  \ab(
    D^{-2/\alpha}X^\top X
    +\frac{\sigma^2}{\tau^2}I_D
  \ab)^{-1}.
  \label{eq:bayes-ridge-estimator}
\end{align}
Thus $\hat{\vb w}_{\mathrm B}$ is the squared-loss ridge estimator with
$\lambda_{\mathrm B}=\sigma^2/\tau^2$.
We introduce its proportional-limit coefficient error as
\begin{align}
  q_{\alpha,\mathrm B}\ab(\kappa)
  &:={}
  \E\ab[
    \frac1D
    \sum_{i=1}^D
    \abs{\hat w_{\mathrm B,i}-w_i^*}^{\alpha}
  \ab].
  \label{eq:bayes-coefficient-error-definition}
\end{align}
When $\alpha=2$, the coefficient error is determined by the average posterior variance and can be analyzed through the usual Stieltjes transform of the empirical covariance matrix. 
For $0<\alpha<2$, however, the error also depends on how the posterior variances vary across coordinates. 
Hence the normalized trace alone is insufficient, and we must characterize the distribution of the local diagonal resolvent entries.
The required local-resolvent analysis nevertheless reduces to a scalar fixed-point equation for $q_{\alpha,\mathrm B}$. 
The following result states this equation together with its data-rich asymptotics.

\begin{result}[Coefficient error of the Bayes ridge estimator]
  \label{res:bayes-optimal-error}
  The coefficient error satisfies
  \begin{align}
    q_{\alpha,\mathrm B}&\ab(\kappa)
    =
    \frac{
      2^{\alpha/2}\Gamma\ab((\alpha+1)/2)\sigma^\alpha
    }{
      \sqrt\pi\,\Gamma\ab(1+\alpha/2)
    }\notag \\
    &\quad \times 
    \int_0^\infty
    \exp\Bigg\{
      -\frac{\sigma^2}{\tau^2}u^{\frac{2}{\alpha}}
      -\frac{
        \kappa 2^\alpha\Gamma\ab(\frac{\alpha+1}{2})
      }{
        \sqrt\pi\,\Gamma\ab(1+\alpha/2)
      }
      u
      \int_0^\infty
      \exp\ab(
        -v^{\frac{2}{\alpha}}
        -\frac{
          2^{\frac{\alpha}{2}}q_{\alpha,\mathrm B}\ab(\kappa)
        }{
          \sigma^\alpha
        }v
      \ab)
      \,\mathrm dv
    \Bigg\}
    \,\mathrm du.
    \label{eq:bayes-optimal-alpha-error-fixed-point}
  \end{align}
  In particular, expanding
  Equation~\eqref{eq:bayes-optimal-alpha-error-fixed-point} for large
  $\kappa$ gives
  \begin{align}
    q_{\alpha,\mathrm B}\ab(\kappa)
    &=
    \frac{2^{-\alpha/2}}{\Gamma\ab(1+\alpha/2)}
    \frac{\sigma^\alpha}{\kappa}
    +o\ab(\kappa^{-1}).
    \label{eq:bayes-optimal-alpha-error}
  \end{align}
\end{result}

The detailed derivation is provided in Appendix~\ref{app:bayes-benchmark}.
The black curves in Figures~\ref{fig:rs-erm-coefficient-comparison}A and B show the theoretical Bayes-optimal error, which indeed provides a tight lower bound for the errors achieved by the Huber-loss estimator.

\paragraph{Comparison with elliptical heavy-tailed features.}
Following \citet{Adomaityte_2024}, we consider elliptical features written in
our scaling convention as
\begin{align}
  \vb x_\mu^{\mathrm{ell}}
  &=
  D^{1/\alpha}R_\mu\vb z_\mu,
  \qquad
  \vb z_\mu\sim\mathcal N\ab(\vb 0,I_D/D),
  \label{eq:elliptical-feature-model}
\end{align}
The labels are generated from Equation~\eqref{eq:label-model}, with
$\vb x_\mu$ replaced by $\vb x_\mu^{\mathrm{ell}}$.
To isolate the effect of feature dependence, we choose $R_\mu^2$, independently of $\vb z_\mu$, to have the limiting
distribution of the normalized squared row norm in the coordinate-wise
independent model. Then
\begin{align}
  D^{-2/\alpha}\lVert\vb x_\mu\rVert_2^2
  \xrightarrow[D\to\infty]{\mathrm d}R_\mu^2, 
  \quad
  D^{-2/\alpha}\lVert\vb x_\mu^{\mathrm{ell}}\rVert_2^2
  \xrightarrow[D\to\infty]{\mathrm d}R_\mu^2.
  \label{eq:elliptical-radial-normalization}
\end{align}
Both \(\vb x_\mu\) and \(\vb x_\mu^{\mathrm{ell}}\) are multivariate symmetric \(\alpha\)-stable.

Let $q_{\alpha,\mathrm B}^{\mathrm{ell}}(\kappa)$ denote the corresponding coefficient $\alpha$-error.
Specializing the Bayes-optimal equations of
\citet{Adomaityte_2024} to the Gaussian prior and noise in
Equation~\eqref{eq:bayes-matched-model} gives
\begin{align}
  1
  &=
  \frac{1}{\tau^2}
  \left[
    \frac{
      \sqrt\pi\,
      q_{\alpha,\mathrm B}^{\mathrm{ell}}\ab(\kappa)
    }{
      2^{\alpha/2}\Gamma\ab((\alpha+1)/2)
    }
  \right]^{\frac{2}{\alpha}}
  +
  \kappa
  \left[
    1-\sigma^2
    \int_0^\infty
    \exp\ab(
      -\sigma^2u
      -2^{\alpha/2}
      q_{\alpha,\mathrm B}^{\mathrm{ell}}\ab(\kappa)
      u^{\alpha/2}
    )
    \,\mathrm du
  \right].
  \label{eq:elliptical-bayes-alpha-error}
\end{align}
Comparison with
Equation~\eqref{eq:bayes-optimal-alpha-error-fixed-point} yields the following
result.

\begin{result}[Coordinate-wise versus elliptical heavy tails]
  \label{res:iid-elliptical-bayes-comparison}
  For $0<\alpha<2$ and $\kappa>0$,
  \begin{align}
    q_{\alpha,\mathrm B}\ab(\kappa)
    &<
    q_{\alpha,\mathrm B}^{\mathrm{ell}}\ab(\kappa).
    \label{eq:iid-elliptical-bayes-ordering}
  \end{align}
  Nevertheless, their data-rich asymptotics agree including the leading
  coefficient:
  \begin{align}
    q_{\alpha,\mathrm B}\ab(\kappa)
    \sim
    q_{\alpha,\mathrm B}^{\mathrm{ell}}\ab(\kappa)
    \sim
    \frac{2^{-\alpha/2}}{\Gamma\ab(1+\alpha/2)}
    \frac{\sigma^\alpha}{\kappa}.
    \label{eq:iid-elliptical-bayes-leading-order}
  \end{align}
\end{result}

A derivation of Equation~\eqref{eq:elliptical-bayes-alpha-error} and
Result~\ref{res:iid-elliptical-bayes-comparison} is provided in
Appendix~\ref{app:bayes-elliptical-comparison}.
This ordering has a simple interpretation.
In the entry-wise independent model, large feature magnitudes occur at different coordinates of different samples and are therefore distributed among parameter directions.
In the elliptical model, by contrast, the common radial variable $R_\mu$ scales all coordinates of a sample simultaneously.
At finite $\kappa$, large feature magnitudes are therefore concentrated in fewer samples and parameter directions, leading to a larger coefficient error in the elliptical model.
As $\kappa\to\infty$, the number of independently observed directions is sufficient for this difference not to affect the leading-order error, so the two models have the same $1/\kappa$ asymptotics.

%% file: section/discussion.tex
\section{Discussion}
\label{sec:discussion}

We developed an asymptotic theory for regularized ERM in linear regression with independent symmetric $\alpha$-stable feature coordinates. 
Using the replica method with a functional order parameter, we characterized the coefficient-error distribution and generalization error in the proportional high-dimensional limit. 
This characterization reveals a heavy-tail universality law: the asymptotic prediction behavior is determined only by the tail index and scale of the feature distribution.
The typical prediction errors exhibit the favorable scaling $\kappa^{-1/\alpha}$, where $\kappa$ denotes the number of data per dimension, whereas the probability of a fixed-magnitude error decays only as $\kappa^{-1}$.
The theory also characterizes Bayes-optimal performance and clarifies how extreme observations distinguish entry-wise and elliptical heavy-tailed models.

A promising direction for future work is to develop the present results into inference algorithms for heavy-tailed data. 
% In particular, it would be interesting to derive an approximate message passing (AMP) algorithm~\citep{YoshiyukiKabashima_2003, donoho2009message} whose fixed points reproduce the mean-field equations obtained in this work. 
In particular, it would be interesting to derive an approximate message passing algorithm~\citep{YoshiyukiKabashima_2003,donoho2009message}. Such an algorithm would provide a computationally efficient iterative inference procedure, whose macroscopic properties at fixed-point are expected to be described by the replica saddle-point equations derived in this work.
%The present results will also provide guidance for designing more reliable inference methods in the presence of extreme observations.
More broadly, the functional order parameter provides a framework for high-dimensional systems in which local environments remain random.
Retaining the distribution of these environments may enable precise asymptotic analyses of other spin glass, learning, and optimization problems with persistent local heterogeneity.

%% file: section/additional_related_work.tex
\section{Additional Related Work}
\label{sec:additional-related-work}

The functional order parameters used in this work were first developed for heterogeneous systems whose coordinate-specific effective quantities remain nonuniform even in the high-dimensional limit.
In spin glasses, optimization problems, and random matrices with sparse or heavy-tailed interactions, this heterogeneity cannot generally be summarized by finitely many scalar order parameters; instead, the distribution of the coordinate-specific effective quantities must be specified
\citep{mezardParisi2001Bethe,Franz_2001,neriMetzBolle2010Levy,jung2018LevyKhintchine}.

In discrete-variable models such as diluted spin glasses, local effects can often be parameterized by finitely many effective fields, so that the order parameter is a probability distribution over these fields
\citep{mezardParisi2001Bethe,Franz_2001}.
For continuous-variable models, when the local effective problem cannot be represented by a finite-dimensional family, it must instead be represented as a function of the continuous variable, and the order parameter becomes a probability distribution on a function space.
Previous studies have primarily considered cases in which the effective functions can be represented by finitely many coefficients because of a quadratic structure, their domains are compact and hence amenable to numerical discretization, or each variable has only finitely many interactions in a strictly sparse model
\citep{rogersEtAl2008Cavity,lupoRicciTersenghi2017XY,delBonoEtAl2024Heisenberg,wongSaadGao2005Task,wongSaad2006Equilibration,wongSaad2007RealEdges}.

In contrast, our model combines general convex losses and unbounded continuous variables with heavy-tailed interactions, for which there are a few contributions of large magnitude and infinitely many of small magnitude.
We derive probability distributions on function spaces and their saddle-point equations, whose solutions characterize the limiting coefficient distribution and prediction error.
This construction provides a basis for analyzing continuous-variable systems with infinitely many heterogeneous interactions.

%% file: section/self_consistent_equations.tex
\section{Self-Consistent Equations for Result~\ref{res:learned-parameter-statistics}}
\label{app:self-consistent-equations}

\subsection{Gaussian reference equations}

For comparison, we state only the resulting scalar saddle-point equations for
Gaussian features, with the feature variance included in the definitions of
the order parameters.
Let
$z,\hat z\sim\mathcal N(0,1)$ be independent of each other and of
$w^*\sim p_w$ and $\epsilon\sim p_\epsilon$.  Introduce an external field $h$ and define
\begin{align}
  u_*\ab(h)
  &:=
  \argmin_{u\in\R}
  \ab\{
    \lambda\ab[\gamma\ab(w^*+u)-\gamma\ab(w^*)]
    +\frac{\hat\chi}{2}u^2
    +\sqrt{\hat q}\,\hat z u
    -hu
  \ab\}.
  \label{eq:appendix-gaussian-coefficient-problem}
\end{align}
The coefficient-side saddle equations are
\begin{align}
  q
  &=
  \E \ab[u_*\ab(0)^2] \\
  \chi
  &= \displaystyle
  \E\ab[
    \left.
    \frac{\partial u_*\ab(h)}{\partial h}
    \right|_{h=0}
  \ab].
  \label{eq:appendix-gaussian-coefficient-saddle}
\end{align}
Using $h$ again as a local perturbation variable, the sample-side problem is
\begin{align}
  v_*\ab(h)
  &:=
  \argmin_{v\in\R}
  \ab\{
    \ell\ab(v)
    +\frac{
      \ab(v-\sqrt q\,z+\epsilon-h\ab)^2
    }{2\chi}
  \ab\}.
  \label{eq:appendix-gaussian-sample-problem}
\end{align}
Its saddle equations are
\begin{align}
  \hat q
  &=
  \kappa\E\ab[
    \ell'\ab(v_*\ab(0))^2
  \ab]
  \\
  \hat\chi
  &=
  \kappa
  \E\ab[
    \left.
    \frac{\partial}{\partial h}
    \ell'\ab(v_*\ab(h))
    \right|_{h=0}
  \ab].
  \label{eq:appendix-gaussian-sample-saddle}
\end{align}
Thus the two scalar problems define the coupled updates
\begin{align}
  \ab(\hat q,\hat\chi)
  &\xrightarrow{\ u_*\ab(h)\ }
  \ab(q,\chi)
  \\
  \ab(q,\chi)
  &\xrightarrow{\ v_*\ab(h)\ }
  \ab(\hat q,\hat\chi).
  \label{eq:appendix-gaussian-coupled-updates}
\end{align}
The heavy-tailed equations below also consist of two problems, but use random
nonlinear response functions instead of scalar quadratic responses.

\subsection{Poisson point processes}
Define the convex conjugate of $\ell$ by
\begin{align}
  \ell^*\ab(p)
  &:=
  \sup_{r\in\R}\ab\{pr-\ell\ab(r)\ab\},
\end{align}
and the effective domain of $\ell^*$ by
\begin{align}
  K_\ell
  &:=
  \ab\{p\in\R:\ell^*\ab(p)<\infty\ab\}
  =
  \operatorname{dom}\ell^*.
\end{align}
Set
\begin{align}
  d_\alpha
  &:=
  \frac{
    \Gamma\ab(\alpha+1)\sin\ab(\pi\alpha/2)
  }{
    2^{\alpha/2}\sqrt\pi\,\Gamma\ab((\alpha+1)/2)
  }.
\end{align}
We use two independent Poisson point processes $\{s_j\}_{j\geq1}$ and $\{\hat s_l\}_{l\geq1}$ on $(0,\infty)$.
Their points, ordered by decreasing magnitude, satisfy
\begin{align}
  \#\ab\{j:s_j>t\}
  &\overset{\mathrm d}{=}
  \operatorname{Poisson}
  \ab(
    \frac{2d_\alpha}{\alpha}t^{-\alpha/2}
  \ab),
  \\
  \#\ab\{l:\hat s_l>t\}
  &\overset{\mathrm d}{=}
  \operatorname{Poisson}
  \ab(
    \frac{2\kappa d_\alpha}{\alpha}t^{-\alpha/2}
  \ab)
  \label{eq:appendix-positive-poisson-sequence}
\end{align}
for $t>0$.
Independently, let
\begin{align}
  g_j,\hat g_l
  &\overset{\mathrm{i.i.d.}}{\sim}\mathcal N\ab(0,1).
\end{align}
These are independent Gaussian variables attached to the Poisson points
\(s_j\) and \(\hat s_l\), respectively; such attached variables are called
\emph{marks}.
\subsection{Fixed-point equation}

Fix $0<\alpha<2$ and suppose that Assumption~\ref{asm:loss-regularity} holds.
We seek a pair of probability distributions
$\mathcal P$ for $(Q,\xi,\mathcal R)$ and $\hat{\mathcal P}$ for $(\hat Q,\zeta,\hat{\mathcal R})$,
such that each distribution is generated from the other by the following two updates.

\subsubsection{From \texorpdfstring{$\hat{\mathcal P}$}{hat P} to
\texorpdfstring{$\mathcal P$}{P}}

Here $\hat Q\geq0$, and $\hat{\mathcal R}$ is a random convex function satisfying
$\hat{\mathcal R}\ab(0)=0$ and $0\in\partial\hat{\mathcal R}\ab(0)$.
Draw independent copies
\begin{align}
  \ab(\hat Q_j,\zeta_j,\hat{\mathcal R}_j)
  &\overset{\mathrm{i.i.d.}}{\sim}
  \hat{\mathcal P}
  \\
  w_j^*&\overset{\mathrm{i.i.d.}}{\sim}p_w
\end{align}
independently of $\{(s_j,g_j)\}_{j\geq1}$, with the $w_j^*$ independent of the triples.
As in the Gaussian coefficient problem, introduce an external field $h$ and
define
\begin{align}
  u_{j,*}\ab(h)
  &:=
  \argmin_{u\in\R}
  \ab\{
    \lambda\ab[
      \gamma\ab(w_j^*+u)-\gamma\ab(w_j^*)
    ]
    +\sqrt{\hat Q_j}\,\zeta_j u
    +\hat{\mathcal R}_j\ab(u)
    -hu
  \ab\}.
  \label{eq:appendix-input-optimizer-copy}
\end{align}
The coefficient-side saddle-point equations are
\begin{align}
  Q
  &=
  \sum_{j\geq1}s_j u_{j,*}\ab(0)^2
  \label{eq:appendix-output-random-scale}
  \\
  \xi
  &=
  \frac{
    1
  }{
    \sqrt{Q}
  }
  \sum_{j\geq1}\sqrt{s_j}\,g_j u_{j,*}\ab(0)
  \label{eq:appendix-output-standardized-field}
  \\
  \mathcal R'\ab(h)
  &=
  \sum_{j\geq1}\sqrt{s_j}\,g_j
  \ab[
    u_{j,*}\ab(-\sqrt{s_j}\,g_jh)-u_{j,*}\ab(0)
  \ab],
  \label{eq:appendix-output-remainder}
\end{align}
where $\mathcal R\ab(0)=0$.

The first update requires the resulting joint distribution
to satisfy
\begin{align}
  \ab(Q,\xi,\mathcal R)
  &\sim\mathcal P.
  \label{eq:appendix-primal-law}
\end{align}

\subsubsection{From \texorpdfstring{$\mathcal P$}{P} to
\texorpdfstring{$\hat{\mathcal P}$}{hat P}}

Draw independent copies
\begin{align}
  \ab(Q_l,\xi_l,\mathcal R_l)
  &\overset{\mathrm{i.i.d.}}{\sim}\mathcal P \\
  \epsilon_l
  &\overset{\mathrm{i.i.d.}}{\sim}p_\epsilon,
\end{align}
independently of $\{(\hat s_l,\hat g_l)\}_{l\geq1}$.  As in the Gaussian
sample problem, introduce an external field $h$ and define
\begin{align}
  v_{l,*}\ab(h)
  &\in
  \argmin_{v\in\R}
  \ab\{
    \ell\ab(v)
    +\ab(-\mathcal R_l\ab)^*
    \ab(\sqrt{Q_l}\,\xi_l-\epsilon_l+h-v)
  \ab\}.
  \label{eq:appendix-output-message}
\end{align}
The sample-side saddle-point equations are
\begin{align}
  \hat Q
  &=
  \sum_{l\geq1}\hat s_l
  \ell'\ab(v_{l,*}\ab(0))^2
  \label{eq:appendix-input-random-scale}
  \\
  \zeta
  &=
  \frac{
    1
  }{
    \sqrt{\hat Q}
  }
  \sum_{l\geq1}\sqrt{\hat s_l}\,\hat g_l
    \ell'\ab(v_{l,*}\ab(0))
  \label{eq:appendix-input-linear-field}
  \\
  \hat{\mathcal R}'\ab(h)
  &=
  \sum_{l\geq1}\sqrt{\hat s_l}\,\hat g_l
  \ab[
    \ell'\ab(v_{l,*}\ab(\sqrt{\hat s_l}\,\hat g_lh))
    -\ell'\ab(v_{l,*}\ab(0))
  \ab],
  \label{eq:appendix-input-random-function}
\end{align}
where $\hat{\mathcal R}\ab(0)=0$.

The second update requires the resulting joint distribution to satisfy
\begin{align}
  \ab(\hat Q,\zeta,\hat{\mathcal R})
  &\sim\hat{\mathcal P}.
  \label{eq:appendix-conjugate-law}
\end{align}
Equations~\eqref{eq:appendix-primal-law} and
\eqref{eq:appendix-conjugate-law} are the two coupled saddle-point equations.
Thus the two random scalar problems define the coupled updates
\begin{align}
  \hat{\mathcal P}
  &\xrightarrow{\ u_{j,*}\ab(h)\ }
  \mathcal P
  \\
  \mathcal P
  &\xrightarrow{\ v_{l,*}\ab(h)\ }
  \hat{\mathcal P}.
  \label{eq:appendix-heavy-tail-coupled-updates}
\end{align}
These equations also consist of two problems, but use the joint distributions
of a random scale, a standard Gaussian field, and a random nonlinear response
function instead of scalar order parameters.

\subsection{Limiting coefficient and preactivation distributions}
At the fixed point, draw fresh effective variables
\begin{align}
  \ab(\hat Q,\zeta,\hat{\mathcal R})
  &\sim\hat{\mathcal P}
  \\
  w^*&\sim p_w,
\end{align}
independently.
The coefficient-side one-body problem is
\begin{align}
  u_*\ab(h)
  &:=
  \argmin_{u\in\R}
  \ab\{
    \lambda\ab[
      \gamma\ab(w^*+u)-\gamma\ab(w^*)
    ]
    +\sqrt{\hat Q}\,\zeta u
    +\hat{\mathcal R}\ab(u)
    -hu
  \ab\}.
  \label{eq:appendix-limiting-coefficient-problem}
\end{align}
Its unperturbed solution gives the limiting coefficient displacement:
\begin{align}
  \hat w_i-w_i^*
  &\xrightarrow{\mathrm d}
  u_*\ab(0).
  \label{eq:appendix-limiting-coefficient-law}
\end{align}

For the sample side, draw fresh effective variables
\begin{align}
  \ab(Q,\xi,\mathcal R)
  &\sim\mathcal P \\
  \epsilon
  &\sim p_\epsilon,
\end{align}
independently, and define
\begin{align}
  v_*\ab(h)
  &\in
  \argmin_{v\in\R}
  \ab\{
    \ell\ab(v)
    +\ab(-\mathcal R\ab)^*
    \ab(\sqrt{Q}\,\xi-\epsilon+h-v)
  \ab\}.
  \label{eq:appendix-limiting-preactivation-problem}
\end{align}
Its unperturbed solution gives the limiting training preactivation, or
equivalently the fitted residual entering the loss:
\begin{align}
  \frac{1}{D^{1/\alpha}}\vb{x}_\mu^\top\hat{\vb w}-y_\mu
  &=
  \frac{1}{D^{1/\alpha}}\vb{x}_\mu^\top
  \ab(\hat{\vb w}-\vb w^*\ab)-\epsilon_\mu
  \xrightarrow{\mathrm d}
  v_*\ab(0).
  \label{eq:appendix-limiting-preactivation-law}
\end{align}

%% file: section/replica_derivation.tex
\section{Derivation of Result~\ref{res:learned-parameter-statistics}}
\label{app:replica-derivation}

This appendix derives the self-consistent equations in Appendix~\ref{app:self-consistent-equations}.

\subsection{Replicated partition function}

Set $\vb u=\vb w-\vb w^*$ and introduce the partition function
\begin{align}
  Z_\beta
  &:=
  \int_{\R^D}\mathrm d\vb u\,
  \exp\ab\{
    -\beta\sum_{\mu=1}^N
    \ell\ab(
      D^{-1/\alpha}\vb x_\mu^\top\vb u-\epsilon_\mu
    )
    -\beta\lambda\sum_{i=1}^D\gamma\ab(w_i^*+u_i)
  \ab\}.
  \label{eq:replica-partition-function}
\end{align}
Its zero-temperature free energy is the optimized training objective:
\begin{align}
  -\lim_{\beta\to\infty}\frac{1}{\beta D}\E\log Z_\beta
  &=
  \frac1D\E\min_{\vb u\in\R^D}
  \ab\{
    \sum_{\mu=1}^N
    \ell\ab(
      D^{-1/\alpha}\vb x_\mu^\top\vb u-\epsilon_\mu
    )
    +\lambda\sum_{i=1}^D\gamma\ab(w_i^*+u_i)
  \ab\}.
  \label{eq:replica-zero-temperature-objective}
\end{align}
The replica identity replaces the logarithm by integer moments:
\begin{align}
  \E\log Z_\beta
  &=
  \lim_{n\to0}
  \frac{\log\E Z_\beta^n}{n}.
  \label{eq:replica-identity}
\end{align}
The right-hand side is first evaluated for integer $n\geq1$ and then analytically continued to a neighborhood of $n=0$.
From this point, our goal is to evaluate the data average of the replicated partition function $Z_\beta^n$ for integer $n$.

Set
\begin{align}
  g_\beta\ab(r)
  &:=
  e^{-\beta\ell(r)}.
\end{align}
For integer $n$, let $u_i^a$ denote coordinate $i$ ($i \in \{1,\ldots,D\}$) in replica $a$ ($a \in \{1,\ldots,n\}$).
Direct multiplication of the $n$ copies of Equation~\eqref{eq:replica-partition-function} gives
\begin{align}
  Z_\beta^n
  &=
  \int
  \prod_{a=1}^n\prod_{i=1}^D\mathrm du_i^a\,
  \exp\ab(
    -\beta\lambda
    \sum_{a=1}^n\sum_{i=1}^D
    \gamma\ab(w_i^*+u_i^a)
  )
  \prod_{\mu=1}^N\prod_{a=1}^n
  g_\beta\ab(
    D^{-1/\alpha}
    \sum_{i=1}^D x_{\mu i}u_i^a-\epsilon_\mu
  ).
  \label{eq:replica-expanded-partition-function}
\end{align}
Write $\vb u_i=(u_i^1,\ldots,u_i^n)^\top$ and introduce the empirical distribution of the replicated coordinate vectors,
\begin{align}
  \rho_D\ab(\vb u)
  &:=
  \frac1D\sum_{i=1}^D\delta\ab(\vb u-\vb u_i).
  \label{eq:replica-empirical-measure}
\end{align}

For one factor in the product over $\mu$, let $\vb k=(k_1,\ldots,k_n)^\top\in\R^n$.
Independence over the feature coordinates and the characteristic function of $S\alpha S(1)$ give
\begin{align}
  \E_{\vb x_\mu}
  \exp\ab(
    iD^{-1/\alpha}
    \sum_{a=1}^n k_a
    \sum_{i=1}^D x_{\mu i}u_i^a
  )
  &=
  \prod_{i=1}^D
  \E_{x_{\mu i}}
  \exp\ab(
    iD^{-1/\alpha}x_{\mu i}
    \sum_{a=1}^n k_au_i^a
  )
  \\
  &=
  \prod_{i=1}^D
  \exp\ab(
    -\frac1D
    \abs{\sum_{a=1}^nk_au_i^a}^{\alpha}
  )
  \\
  &=
  \exp\ab(
    -\frac1D\sum_{i=1}^D
    \abs{\vb k^\top\vb u_i}^{\alpha}
  )
  \\
  &=
  \exp\ab(
    -\int_{\R^n}\mathrm d\vb v\,\rho_D\ab(\vb v)
    \abs{\vb k^\top\vb v}^{\alpha}
  ).
  \label{eq:replica-stable-feature-average}
\end{align}

Use the Fourier convention
\begin{align}
  \widetilde g_\beta\ab(k)
  &:=
  \int_\R\mathrm dr\,g_\beta\ab(r)e^{-ikr},
  \\
  g_\beta\ab(r)
  &=
  \int_\R\frac{\mathrm dk}{2\pi}\,
  \widetilde g_\beta\ab(k)e^{ikr}.
  \label{eq:replica-loss-fourier-transform}
\end{align}
This step assumes that Fourier inversion is valid; otherwise the following formula is interpreted distributionally.
Applying it to all $n$ loss factors in Equation~\eqref{eq:replica-expanded-partition-function} and then using Equation~\eqref{eq:replica-stable-feature-average} yields
\begin{align}
  &\E_{\vb x,\epsilon}
  \prod_{a=1}^n
  g_\beta\ab(
    D^{-1/\alpha}
    \sum_{i=1}^Dx_i u_i^a-\epsilon
  )
  \\
  &=
  \int_{\R^n}
  \prod_{a=1}^n
  \frac{\mathrm dk_a}{2\pi}\,
  \widetilde g_\beta\ab(k_a)
  \E_\epsilon
  \exp\ab(
    -i\epsilon\sum_{a=1}^nk_a
  )
  \exp\ab(
    -\int_{\R^n}\mathrm d\vb v\,
    \rho_D\ab(\vb v)
    \abs{\vb k^\top\vb v}^{\alpha}
  )
  \\
  &=:
  \mathcal I_n[\rho_D].
  \label{eq:replica-single-factor-fourier}
\end{align}
The integral in the second line, with $\rho_D$ replaced by an arbitrary probability density $\rho$ on $\R^n$, defines $\mathcal I_n[\rho]$.

The training examples are independent.
Therefore, conditioning on $\vb w^*$ and using the same empirical measure for every factor,
\begin{align}
  &\E_{\{\vb x_\mu,\epsilon_\mu\}_{\mu=1}^N}
  \ab[Z_\beta^n\mid\vb w^*]
  \\
  &=
  \int\prod_{a=1}^n\prod_{i=1}^D\mathrm du_i^a\,
  \exp\ab(
    -\beta\lambda\sum_{a=1}^n\sum_{i=1}^D
    \gamma\ab(w_i^*+u_i^a)
  )
  \prod_{\mu=1}^N\mathcal I_n[\rho_D]
  \\
  &=
  \int\prod_{a=1}^n\prod_{i=1}^D\mathrm du_i^a\,
  \exp\ab(
    -\beta\lambda\sum_{a=1}^n\sum_{i=1}^D
    \gamma\ab(w_i^*+u_i^a)
  )
  \mathcal I_n[\rho_D]^N.
  \label{eq:replica-averaged-moment}
\end{align}

We enforce Equation~\eqref{eq:replica-empirical-measure} before integrating over the replicated coordinates.
The Fourier representation of the functional delta gives
\begin{align}
  1
  &=
  \int\mathcal D\rho\,
  \delta\ab[
    D\rho(\mathord\cdot)-\sum_{i=1}^D\delta(\mathord\cdot-\vb u_i)
  \ab]\propto
  \int\mathcal D\rho\,\mathcal D\hat\rho\,
  \exp\ab\{
    -D\int_{\R^n}\mathrm d\vb u\,\rho\ab(\vb u)\hat\rho\ab(\vb u)
    +\sum_{i=1}^D\hat\rho\ab(\vb u_i)
  \ab\}.
  \label{eq:replica-functional-delta}
\end{align}
The second representation initially integrates the conjugate field along an imaginary contour.
The omitted proportionality factor is independent of the replicated coordinates and is subexponential in $D$.
Define the one-coordinate integral
\begin{align}
  \mathcal Z_\gamma\ab(w;\hat\rho)
  &:=
  \int_{\R^n}\mathrm d\vb u\,
  \exp\ab\{
    \hat\rho\ab(\vb u)
    -\beta\lambda\sum_{a=1}^n\gamma\ab(w+u^a)
  \ab\}.
  \label{eq:replica-one-coordinate-partition-function}
\end{align}
Substituting Equation~\eqref{eq:replica-functional-delta} into Equation~\eqref{eq:replica-averaged-moment} and averaging the independent teacher coordinates gives
\begin{align}
  \E Z_\beta^n
  &\doteq
  \int\mathcal D\rho\,\mathcal D\hat\rho\,
  \exp\ab\{
    -D\int\mathrm d\vb u\,\rho\ab(\vb u)\hat\rho\ab(\vb u)
    +N\log\mathcal I_n[\rho]
  \ab\}
  \prod_{i=1}^D
  \E_{w_i^*}\mathcal Z_\gamma\ab(w_i^*;\hat\rho)
  \notag\\
  &=
  \int\mathcal D\rho\,\mathcal D\hat\rho\,
  \exp\ab\{
    -D\int\mathrm d\vb u\,\rho\ab(\vb u)\hat\rho\ab(\vb u)
    +N\log\mathcal I_n[\rho]
    +D\log\E_{w^*}\mathcal Z_\gamma\ab(w^*;\hat\rho)
  \ab\}
  \notag\\
  &=
  \int\mathcal D\rho\,\mathcal D\hat\rho\,
  \exp\ab(D\mathcal S_n[\rho,\hat\rho]+o(D)),
  \label{eq:replica-functional-integral}
\end{align}
where $\doteq$ denotes equality on the exponential scale and $N/D\to\kappa$.
Here we define the action as
\begin{align}
  &\mathcal S_n[\rho,\hat\rho] \\
  =&
  -\int\mathrm d\vb u\,\rho\ab(\vb u)\hat\rho\ab(\vb u)
  +\kappa\log\mathcal I_n[\rho]
  +\log\E_{w^*}
  \mathcal Z_\gamma\ab(w^*;\hat\rho)
  \\
  =&
  -\int\mathrm d\vb u\,\rho\ab(\vb u)\hat\rho\ab(\vb u)
  +\kappa\log\mathcal I_n[\rho]
  +
  \log\E_{w^*}
  \int_{\R^n}\mathrm d\vb u\,
  \exp\ab\{
    \hat\rho\ab(\vb u)
    -\beta\lambda\sum_{a=1}^n\gamma\ab(w^*+u^a)
  \ab\}.
  \label{eq:replica-action}
\end{align}
Next, we take the $D\to\infty$ limit of the action.
The saddle point assumption gives
\begin{align}
  \lim_{D\to\infty}\frac1D\log\E Z_\beta^n
  &=
  \underset{\rho,\hat\rho}{\operatorname{extr}}\,
  \mathcal S_n[\rho,\hat\rho].
  \label{eq:replica-saddle-evaluation}
\end{align}
Differentiating the action with respect to $\hat\rho$ gives
\begin{align}
  \rho\ab(\vb u)
  &=
  \frac{
    \E_{w^*}\exp\ab\{
      \hat\rho\ab(\vb u)
      -\beta\lambda\sum_{a=1}^n\gamma\ab(w^*+u^a)
    \ab\}
  }{
    \E_{w^*}\int_{\R^n}\mathrm d\vb v\,
    \exp\ab\{
      \hat\rho\ab(\vb v)
      -\beta\lambda\sum_{a=1}^n\gamma\ab(w^*+v^a)
    \ab\}
  }
  \label{eq:replica-rho-stationarity}
\end{align}
and the variation with respect to $\rho$ gives
\begin{align}
  \hat\rho\ab(\vb u)
  &=
  \kappa
  \frac{\delta\log\mathcal I_n[\rho]}
  {\delta\rho\ab(\vb u)}
  =
  \frac{\kappa}{\mathcal I_n[\rho]}
  \frac{\delta\mathcal I_n[\rho]}
  {\delta\rho\ab(\vb u)}.
  \label{eq:replica-rhohat-stationarity}
\end{align}
For variations of $\rho$ with zero total mass, the functional derivative is defined up to an additive constant.
A representative obtained directly from Equation~\eqref{eq:replica-single-factor-fourier} is
\begin{align}
  \frac{\delta\mathcal I_n[\rho]}
  {\delta\rho\ab(\vb u)}
  &=
  -\int_{\R^n}
  \prod_{a=1}^n\frac{\mathrm dk_a}{2\pi}\,
  \widetilde g_\beta\ab(k_a)
  \E_\epsilon
  \exp\ab(
    -i\epsilon\sum_{a=1}^nk_a
  )
  \abs{\vb k^\top\vb u}^{\alpha}
  \exp\ab(
    -\int_{\R^n}\mathrm d\vb v\,
    \rho\ab(\vb v)
    \abs{\vb k^\top\vb v}^{\alpha}
  ),
  \label{eq:replica-single-factor-functional-derivative}
\end{align}
and hence
\begin{align}
  \hat\rho\ab(\vb u)
  &=
  -\frac{\kappa}{\mathcal I_n[\rho]}
  \int_{\R^n}
  \prod_{a=1}^n\frac{\mathrm dk_a}{2\pi}\,
  \widetilde g_\beta\ab(k_a)
  \E_\epsilon
  \exp\ab(
    -i\epsilon\sum_{a=1}^nk_a
  )
  \abs{\vb k^\top\vb u}^{\alpha}
  \exp\ab(
    -\int_{\R^n}\mathrm d\vb v\,
    \rho\ab(\vb v)
    \abs{\vb k^\top\vb v}^{\alpha}
  ).
  \label{eq:replica-rhohat-fourier-form}
\end{align}
Equation~\eqref{eq:replica-rho-stationarity} gives $\int\mathrm d\vb u\,\rho\ab(\vb u)=1$.
Consequently, for every constant $c$,
\begin{align}
  \mathcal S_n[\rho,\hat\rho+c]
  &=
  \mathcal S_n[\rho,\hat\rho]
  -c\int\mathrm d\vb u\,\rho\ab(\vb u)+c
  =
  \mathcal S_n[\rho,\hat\rho].
\end{align}
Thus an additive constant in $\hat\rho$ is immaterial.

\begingroup
\setlength{\fboxsep}{8pt}
\noindent\colorbox{blue!5}{%
\begin{minipage}{\dimexpr\linewidth-2\fboxsep\relax}
\small
\textbf{\color{blue!45!black}Summary: Saddle-point equations for \(\rho\) and \(\hat\rho\) on \(\R^n\)}
\begin{align}
  \rho\ab(\vb u)
  &=
  \frac{
    \E_{w^*}\exp\ab\{
      \hat\rho\ab(\vb u)
      -\beta\lambda\sum_{a=1}^n\gamma\ab(w^*+u^a)
    \ab\}
  }{
    \E_{w^*}\int_{\R^n}\mathrm d\vb v\,
    \exp\ab\{
      \hat\rho\ab(\vb v)
      -\beta\lambda\sum_{a=1}^n\gamma\ab(w^*+v^a)
    \ab\}
  }
  \notag
  \\
  \hat\rho\ab(\vb u)
  &=
  -\frac{\kappa}{\mathcal I_n[\rho]}
  \int_{\R^n}
  \prod_{a=1}^n\frac{\mathrm dk_a}{2\pi}\,
  \widetilde g_\beta\ab(k_a)
  \E_\epsilon
  \exp\ab(
    -i\epsilon\sum_{a=1}^nk_a
  )
  \abs{\vb k^\top\vb u}^{\alpha}
  \exp\ab(
    -\int_{\R^n}\mathrm d\vb v\,
    \rho\ab(\vb v)
    \abs{\vb k^\top\vb v}^{\alpha}
  ).
  \notag
\end{align}
\end{minipage}%
}
\endgroup

\subsection{Replica symmetry}

Equations~\eqref{eq:replica-rho-stationarity} and \eqref{eq:replica-rhohat-fourier-form} are saddle equations for functions on \(\R^n\).
To analyze these equations in the \(n\to0\) limit, we assume replica symmetry (RS).
Although RS remains an assumption rather than a theorem for the present model, to our knowledge there is no known convex problem in the proportional high-dimensional limit for which it gives incorrect predictions.
This expectation is also consistent with rigorous analyses of related convex estimators with Gaussian features~\citep{8365826,donoho2016high} and with the numerical agreement in Figure~\ref{fig:rs-erm-coefficient-comparison} for the present model.
Under this assumption, the replicated coordinates are conditionally independent Gibbs draws from a common random one-dimensional objective.
We write this assumption as
\begin{align}
  \rho\ab(\vb u)
  &:=
  \int\Pi_\beta\ab(\mathrm d\mathcal V)
  \prod_{a=1}^n\mu_{\mathcal V}\ab(u^a),
  \label{eq:replica-rs-ansatz}
  \\
  \mu_{\mathcal V}\ab(u)
  &:=
  \frac{
    e^{-\beta\mathcal V(u)}
  }{
    \int_\R e^{-\beta\mathcal V(v)}\,\mathrm dv
  }.
  \label{eq:replica-single-site-gibbs-measure}
\end{align}
Here \(\Pi_\beta\) is a probability distribution on random one-dimensional coefficient objectives \(\mathcal V:\R\to\R\).
Because
\begin{align}
  \mu_{\mathcal V+c}(u)
  &=
  \frac{e^{-\beta[\mathcal V(u)+c]}}
       {\int_\R e^{-\beta[\mathcal V(v)+c]}\,\mathrm dv}
  =
  \frac{e^{-\beta\mathcal V(u)}}
       {\int_\R e^{-\beta\mathcal V(v)}\,\mathrm dv}
  =
  \mu_{\mathcal V}(u),
  \label{eq:replica-gibbs-additive-invariance}
\end{align}
the Gibbs measure determines \(\mathcal V\) only up to an additive constant.
We select a unique representative by imposing \(\mathcal V(0)=0\).

Similarly, let \(\hat\Pi_\beta\) be the distribution of a random function
\(\hat{\mathcal V}_\beta:\R\to\R\) with \(\hat{\mathcal V}_\beta(0)=0\),
and write the conjugate RS ansatz as
\begin{align}
  e^{\hat\rho(\vb u)}
  &:=
  \int\hat\Pi_\beta
  \ab(\mathrm d\hat{\mathcal V})
  \prod_{a=1}^n
  e^{-\beta\hat{\mathcal V}(u^a)}.
  \label{eq:replica-conjugate-rs-ansatz}
\end{align}
Thus, under the RS ansatz, the saddle-point problem for the order parameters \(\rho\) and \(\hat\rho\) becomes a problem of determining the probability distributions \(\Pi_\beta\) and \(\hat\Pi_\beta\) of the random functions \(\mathcal V_\beta\) and \(\hat{\mathcal V}_\beta\).

\begingroup
\setlength{\fboxsep}{8pt}
\noindent\colorbox{blue!5}{%
\begin{minipage}{\dimexpr\linewidth-2\fboxsep\relax}
\small
\textbf{\color{blue!45!black}Summary: RS parametrization by \(\Pi_\beta\) and \(\hat\Pi_\beta\)}
\begin{align}
  \rho(\vb u)
  &=
  \E_{\mathcal V_\beta\sim\Pi_\beta}
  \prod_{a=1}^n\mu_{\mathcal V_\beta}(u^a)
  =
  \E_{\mathcal V_\beta\sim\Pi_\beta}
  \prod_{a=1}^n
  \frac{e^{-\beta\mathcal V_\beta(u^a)}}
       {\int_\R e^{-\beta\mathcal V_\beta(v)}\,\mathrm dv},
  \notag
  \\
  e^{\hat\rho(\vb u)}
  &=
  \E_{\hat{\mathcal V}_\beta\sim\hat\Pi_\beta}
  \prod_{a=1}^n e^{-\beta\hat{\mathcal V}_\beta(u^a)}.
  \notag
\end{align}
\end{minipage}%
}
\endgroup

\subsection{RS saddle-point equations of the
\texorpdfstring{$\hat\Pi_\beta\to\Pi_\beta$}{hat Pi beta to Pi beta} update}
We first substitute the conjugate RS ansatz \eqref{eq:replica-conjugate-rs-ansatz} into the \(\rho\)-stationarity condition \eqref{eq:replica-rho-stationarity}.
For independent \(w^*\sim p_w\) and \(\hat{\mathcal V}_\beta\sim\hat\Pi_\beta\), this gives
\begin{align}
  \rho(\vb u)
  &=
  \frac{
    \E\ab[
      \prod_{a=1}^n
      e^{-\beta[
        \lambda\gamma(w^*+u^a)
        +
        \hat{\mathcal V}_\beta(u^a)
      ]}
    ]
  }{
    \E\ab[
      \prod_{a=1}^n
      \int_\R
      e^{-\beta[
        \lambda\gamma(w^*+v^a)
        +
        \hat{\mathcal V}_\beta(v^a)
      ]}
      \,\mathrm dv^a
    ]
  }.
  \label{eq:replica-rho-after-conjugate-rs-substitution}
\end{align}
The exponent in \eqref{eq:replica-rho-after-conjugate-rs-substitution} can be rewritten as
\begin{align}
  \lambda\gamma(w^*+u)
  +
  \hat{\mathcal V}_\beta(u)
  &=
  \lambda\gamma(w^*)
  +
  \lambda\ab[
    \gamma(w^*+u)-\gamma(w^*)
  \ab]
  +
  \hat{\mathcal V}_\beta(u).
  \label{eq:replica-variable-potential-emergence}
\end{align}
The \(u\)-dependent part on the right-hand side of \eqref{eq:replica-variable-potential-emergence} vanishes at \(u=0\).
Define the random coefficient objective
\begin{align}
  \mathcal V_\beta(u)
  &:=
  \lambda\ab[
    \gamma(w^*+u)-\gamma(w^*)
  \ab]
  +
  \hat{\mathcal V}_\beta(u).
  \label{eq:replica-finite-temperature-variable-update}
\end{align}
Using \eqref{eq:replica-finite-temperature-variable-update} in \eqref{eq:replica-rho-after-conjugate-rs-substitution} gives
\begin{align}
  \rho(\vb u)
  &=
  \frac{
    \E\ab[
      e^{-n\beta\lambda\gamma(w^*)}
      \prod_{a=1}^n
      e^{-\beta\mathcal V_\beta(u^a)}
    ]
  }{
    \E\ab[
      e^{-n\beta\lambda\gamma(w^*)}
      \prod_{a=1}^n
      \int_\R
      e^{-\beta\mathcal V_\beta(v^a)}
      \,\mathrm dv^a
    ]
  }
  \\
  &=
  \frac{
    \E\ab[
      e^{-n\beta\lambda\gamma(w^*)}
      \ab(
        \int_\R e^{-\beta\mathcal V_\beta(v)}
        \,\mathrm dv
      )^n
      \prod_{a=1}^n
      \mu_{\mathcal V_\beta}(u^a)
    ]
  }{
    \E\ab[
      e^{-n\beta\lambda\gamma(w^*)}
      \ab(
        \int_\R e^{-\beta\mathcal V_\beta(v)}
        \,\mathrm dv
      )^n
    ]
  }
  \\
  &\xrightarrow{\ n\to0\ }
  \E
  \prod_{a=1}^n
  \mu_{\mathcal V_\beta}(u^a)
  \\
  &=
  \int
  \operatorname{Law}(\mathcal V_\beta)
  \ab(\mathrm d\mathcal V)
  \prod_{a=1}^n
  \mu_{\mathcal V}(u^a).
  \label{eq:replica-rs-law-closure}
\end{align}
Comparing \eqref{eq:replica-rs-law-closure} with \eqref{eq:replica-rs-ansatz} yields the first distributional relation
\begin{align}
  \hat{\mathcal V}_\beta
  &\sim\hat\Pi_\beta,
  \quad
  w^*\sim p_w
  \quad\Longrightarrow\quad
  \mathcal V_\beta
  =
  \lambda\ab[
    \gamma(w^*+\cdot)-\gamma(w^*)
  \ab]
  +
  \hat{\mathcal V}_\beta
  \sim\Pi_\beta.
  \label{eq:replica-conjugate-to-primal-law}
\end{align}

\begingroup
\setlength{\fboxsep}{8pt}
\noindent\colorbox{blue!5}{%
\begin{minipage}{\dimexpr\linewidth-2\fboxsep\relax}
\small
\textbf{\color{blue!45!black}Summary: Distributional RS saddle-point equation for \(\Pi_\beta\)}
\begin{align}
  \Pi_\beta
  &=
  \operatorname{Law}\ab(
    u\mapsto
    \lambda\ab[\gamma(w^*+u)-\gamma(w^*)\ab]
    +\hat{\mathcal V}_\beta(u)
  \ab),\notag
\end{align}
where $w^*\sim p_w$ and $\hat{\mathcal V}_\beta\sim\hat\Pi_\beta$ are independently distributed.
\end{minipage}%
}
\endgroup

\subsection{RS saddle-point equations of the
\texorpdfstring{$\Pi_\beta\to\hat\Pi_\beta$}{Pi beta to hat Pi beta} update}
We now substitute the primal RS ansatz \eqref{eq:replica-rs-ansatz} into the \(\hat\rho\)-stationarity condition \eqref{eq:replica-rhohat-fourier-form}.
We first rewrite its stable exponent in a form that makes this substitution explicit.
For \(\rho\) with finite \(\alpha\)-moment, the stable exponent in Equation~\eqref{eq:replica-rhohat-fourier-form} has the L\'evy--Khintchine representation
\begin{align}
  \int_{\R^n}\mathrm d\vb v\,\rho\ab(\vb v)
  \abs{\vb k^\top\vb v}^{\alpha}
  &=
  d_\alpha\int_0^\infty
  \E_{g,\vb v\sim\rho}\ab[
    1-\cos\ab(\sqrt{s}\,g\,\vb k^\top\vb v)
  \ab]
  s^{-1-\alpha/2}\,\mathrm ds.
  \label{eq:replica-levy-khintchine}
\end{align}
Here \(d_\alpha\) is defined in Appendix~\ref{app:self-consistent-equations} and \(g\sim\mathcal N(0,1)\).
The identity follows from \(\E_g\cos(\sqrt{s}\,gz)=e^{-sz^2/2}\).
The following formula identifies the exponential of the negative stable exponent with the characteristic function of a Poisson sum.

\begin{proposition}[Campbell formula]
  \label{prop:replica-stable-poisson-transformation}
Assume that \(\rho\) has finite \(\alpha\)-moment.
Let \(\{s_j\}_{j\geq1}\) be the Poisson process in \eqref{eq:appendix-positive-poisson-sequence}, ordered by decreasing size.
Independently draw \(g_j\overset{\mathrm{i.i.d.}}{\sim}\mathcal N(0,1)\) and \(\vb v_j\overset{\mathrm{i.i.d.}}{\sim}\rho\).
The series \(\sum_j\sqrt{s_j}g_j\vb v_j\) converges almost surely and
  \begin{align}
    \E\exp\ab(i\sum_j\sqrt{s_j}g_j\vb k^\top\vb v_j)
    &=\exp\ab\{-d_\alpha\int_0^\infty
      \E_{g,\vb v}\ab[1-\cos(\sqrt{s}g\vb k^\top\vb v)]
      s^{-1-\alpha/2}\,\mathrm ds\ab\}
    \\
    &=\exp\ab\{-\int\mathrm d\vb v\,\rho(\vb v)
      \abs{\vb k^\top\vb v}^{\alpha}\ab\}.
    \label{eq:replica-stable-poisson-transform}
  \end{align}
\end{proposition}

\begin{proof}
First recall the Poisson product identity, including its short proof.
For disjoint sets \(B_r\) of finite intensity \(\nu(B_r)\), their counts \(N_r\) are independent Poisson variables.
Thus a function equal to \(f_r\) on \(B_r\) and to one elsewhere satisfies
\begin{align}
  \E\prod_j f(X_j)
  &=\prod_r\E f_r^{N_r}
  =\prod_r\sum_{k=0}^\infty
    e^{-\nu(B_r)}\frac{\nu(B_r)^k f_r^k}{k!}
  =\exp\ab\{\int(f-1)\,\mathrm d\nu\ab\}.
  \label{eq:replica-poisson-product-identity}
\end{align}
Approximation by simple functions gives this identity for \(f=e^{-h}\), \(h\geq0\), by monotone convergence, and for other integrands whenever both sides and the approximation are integrable.
Independent marks are included by replacing \(\nu\) with intensity times mark distribution.
The same simple-function calculation gives \(\E\sum_j h(X_j)=\int h\,\mathrm d\nu\) for \(h\geq0\).

For any independent scalar or vector mark \(c\), substitution \(r=s\lVert c\rVert^2\) gives
\begin{align}
  d_\alpha\int_0^\infty
    \E[1\wedge s\lVert c\rVert^2]s^{-1-\alpha/2}\,\mathrm ds
  &=d_\alpha\E\lVert c\rVert^\alpha
    \ab[\int_0^1r^{-\alpha/2}\,\mathrm dr
       +\int_1^\infty r^{-1-\alpha/2}\,\mathrm dr]
  \\
  &=\frac{4d_\alpha}{\alpha(2-\alpha)}
    \E\lVert c\rVert^\alpha.
  \label{eq:replica-marked-summability}
\end{align}
If this quantity is finite, there are almost surely only finitely many terms with \(s_j\lVert c_j\rVert^2>1\), and the sum of the remaining terms has finite expectation.
Consequently, \(\sum_j s_j\lVert c_j\rVert^2<\infty\) almost surely.
Taking \(c_j=\vb v_j\), the partial Gaussian sums, conditional on \(\{s_j,\vb v_j\}\), form an \(L^2\)-bounded martingale and converge almost surely and in conditional \(L^2\).
Their characteristic function is
\begin{align}
  \E\exp\ab(i\sum_j\sqrt{s_j}g_j\vb k^\top\vb v_j)
  &=\E\exp\ab(-\frac12\sum_j s_j(\vb k^\top\vb v_j)^2)
  \\
  &=\exp\ab\{d_\alpha\int_0^\infty
    \E_{\vb v}\ab[e^{-s(\vb k^\top\vb v)^2/2}-1]
    s^{-1-\alpha/2}\,\mathrm ds\ab\}.
  \label{eq:replica-conditional-gaussian-series}
\end{align}
Here the second equality is \eqref{eq:replica-poisson-product-identity}.
Using \(\E_g e^{i\sqrt{s}gz}=e^{-sz^2/2}\) and then \eqref{eq:replica-levy-khintchine} proves \eqref{eq:replica-stable-poisson-transform}.
\end{proof}
Draw the Poisson points and Gaussian marks in \eqref{eq:appendix-positive-poisson-sequence}, independent objectives \(\mathcal V_{j,\beta}\sim\Pi_\beta\), and \(\epsilon\sim p_\epsilon\).
For fixed values of these objects, draw independently over $j$ and $a$
\begin{align}
  U_j^a\mid\mathcal V_{j,\beta}
  &\sim\mu_{\mathcal V_{j,\beta}},
\end{align}
where \(\mu_{\mathcal V}\) is the Gibbs measure defined in \eqref{eq:replica-single-site-gibbs-measure}.
Below, \(\E_U\) averages only these Gibbs draws, while the outer expectation averages the objectives, Poisson points, Gaussian marks, and noise.
Substituting \eqref{eq:replica-stable-poisson-transform} into \eqref{eq:replica-single-factor-fourier} and using \eqref{eq:replica-loss-fourier-transform} gives
\begin{align}
  \mathcal I_n[\rho]
  &=\E\E_U\prod_{a=1}^n
    \ab[\int_\R\frac{\mathrm dk_a}{2\pi}\,
      \widetilde g_\beta(k_a)
      e^{ik_a(\sum_j\sqrt{s_j}g_jU_j^a-\epsilon)}]
  \\
  &=\E\prod_{a=1}^n
    \E_{U^a}e^{-\beta\ell(\sum_j\sqrt{s_j}g_jU_j^a-\epsilon)}
  \\
  &=\E\ab[\ab(\E_Ue^{-\beta\ell(\sum_j\sqrt{s_j}g_jU_j-\epsilon)})^n].
  \label{eq:replica-factor-replica-product}
\end{align}
Here the product separates because the replicas are independent conditional on the same objectives and marks.

We next transform the prefactor in \eqref{eq:replica-rhohat-fourier-form}.
Symmetry of \(g\) and \eqref{eq:replica-levy-khintchine} give
\begin{align}
  -\abs{\vb k^\top\vb u}^\alpha
  &=d_\alpha\int_0^\infty
    \E_g[\cos(\sqrt{s}g\vb k^\top\vb u)-1]
    s^{-1-\alpha/2}\,\mathrm ds
  \\
  &=d_\alpha\int_0^\infty
    \E_g[e^{i\sqrt{s}g\sum_a k_au^a}-1]
    s^{-1-\alpha/2}\,\mathrm ds.
  \label{eq:replica-distinguished-coupling}
\end{align}
The additional exponential shifts each Fourier integral explicitly as
\begin{align}
  \int_\R\frac{\mathrm dk_a}{2\pi}\,
    \widetilde g_\beta(k_a)
    e^{ik_a(\sum_j\sqrt{s_j}g_jU_j^a-\epsilon)}
    e^{ik_a\sqrt{s}gu^a}
  &=g_\beta\ab(\sqrt{s}gu^a+\sum_j\sqrt{s_j}g_jU_j^a-\epsilon)
  \\
  &=e^{-\beta\ell(\sqrt{s}gu^a+\sum_j\sqrt{s_j}g_jU_j^a-\epsilon)}.
  \label{eq:replica-shifted-fourier-inversion}
\end{align}
Thus \eqref{eq:replica-rhohat-fourier-form}, \eqref{eq:replica-factor-replica-product}, and \eqref{eq:replica-shifted-fourier-inversion} yield
\begin{align}
  \hat\rho(\vb u)
  &=\frac{\kappa d_\alpha}{\mathcal I_n[\rho]}
    \int_0^\infty s^{-1-\alpha/2}\,\mathrm ds\,
    \E\Bigg[\prod_{a=1}^n
      \E_{U^a}e^{-\beta\ell(\sqrt{s}gu^a+\sum_j\sqrt{s_j}g_jU_j^a-\epsilon)}
    \\[-1mm]
  &\hspace{52mm}
      -\ab(\E_Ue^{-\beta\ell(\sum_j\sqrt{s_j}g_jU_j-\epsilon)})^n\Bigg]
  \\
  &=\frac{\kappa d_\alpha}{\mathcal I_n[\rho]}
    \int_0^\infty s^{-1-\alpha/2}\,\mathrm ds\,
    \E\Bigg[
      \ab(\E_Ue^{-\beta\ell(\sum_j\sqrt{s_j}g_jU_j-\epsilon)})^n
      \\[-1mm]
  &\hspace{31mm}\times
      \left\{
      \prod_{a=1}^n
      \frac{
        \E_{U^a}e^{-\beta\ell(
          \sqrt{s}gu^a+\sum_j\sqrt{s_j}g_jU_j^a-\epsilon)}
      }{
        \E_{U^a}e^{-\beta\ell(
          \sum_j\sqrt{s_j}g_jU_j^a-\epsilon)}
      }
      -1
      \right\}
    \Bigg].
  \label{eq:replica-factor-ratio-emergence}
\end{align}
The \(\hat\rho\)-stationarity equation therefore contains this ratio directly.
It compares the contribution of one training example when the outer coefficient is fixed at \(u\) with its contribution at \(u=0\).
Define its normalized negative logarithm by
\begin{align}
  \hat{\mathcal W}_{s,g,\beta}(u)
  &:=
  -\frac1\beta
  \log
  \frac{
    \E_Ue^{-\beta\ell(
      \sqrt{s}gu+\sum_j\sqrt{s_j}g_jU_j-\epsilon)}
  }{
    \E_Ue^{-\beta\ell(
      \sum_j\sqrt{s_j}g_jU_j-\epsilon)}
  }
  \\
  &=
  -\frac1\beta
  \log
  \E_Ue^{-\beta\ell(
    \sqrt{s}gu+\sum_j\sqrt{s_j}g_jU_j-\epsilon)}
  +
  \frac1\beta
  \log
  \E_Ue^{-\beta\ell(
    \sum_j\sqrt{s_j}g_jU_j-\epsilon)}.
  \label{eq:replica-finite-temperature-sample-contribution}
\end{align}
Substituting \eqref{eq:replica-finite-temperature-sample-contribution} into \eqref{eq:replica-factor-ratio-emergence}, and using \eqref{eq:replica-factor-replica-product} for \(\mathcal I_n[\rho]\), gives
\begin{align}
  \hat\rho(\vb u)
  &=\kappa d_\alpha\int_0^\infty
    \frac{\E\ab[
      \ab(\E_Ue^{-\beta\ell(\sum_j\sqrt{s_j}g_jU_j-\epsilon)})^n
      \ab(e^{-\beta\sum_a\hat{\mathcal W}_{s,g,\beta}(u^a)}-1)]}
      {\E\ab[\ab(\E_Ue^{-\beta\ell(\sum_j\sqrt{s_j}g_jU_j-\epsilon)})^n]}
    s^{-1-\alpha/2}\,\mathrm ds.
  \label{eq:replica-rhohat-replica-weight}
\end{align}
The replica weight tends to one:
\begin{align}
  \ab(\E_Ue^{-\beta\ell(\sum_j\sqrt{s_j}g_jU_j-\epsilon)})^n
   \xrightarrow{\ n\to0\ }1.
  \label{eq:replica-environment-unweighting}
\end{align}
After this limit, \eqref{eq:replica-rhohat-replica-weight} becomes
\begin{align}
  \hat\rho(\vb u)
  &=\kappa d_\alpha\int_0^\infty
    \E\ab[e^{-\beta\sum_a\hat{\mathcal W}_{s,g,\beta}(u^a)}-1]
    s^{-1-\alpha/2}\,\mathrm ds.
  \label{eq:replica-rhohat-n-zero-form}
\end{align}
Exponentiating \eqref{eq:replica-rhohat-n-zero-form} and applying \eqref{eq:replica-poisson-product-identity} with intensity \(\kappa d_\alpha s^{-1-\alpha/2}\,\mathrm ds\) gives
\begin{align}
  e^{\hat\rho(\vb u)}
  &=\exp\ab\{\kappa d_\alpha\int_0^\infty
    \E\ab[e^{-\beta\sum_a\hat{\mathcal W}_{s,g,\beta}(u^a)}-1]
    s^{-1-\alpha/2}\,\mathrm ds\ab\}
  \\
  &=\E\prod_l e^{-\beta\sum_a
       \hat{\mathcal W}^{(l)}_{\hat s_l,\hat g_l,\beta}(u^a)}
  \\
  &=\E e^{-\beta\sum_a\sum_l
       \hat{\mathcal W}^{(l)}_{\hat s_l,\hat g_l,\beta}(u^a)}.
  \label{eq:replica-campbell-conjugate}
\end{align}

\begingroup
\setlength{\fboxsep}{8pt}
\noindent\colorbox{teal!7}{%
\begin{minipage}{\dimexpr\linewidth-2\fboxsep\relax}
\small
\textbf{\color{teal!55!black}Intuition.}
We have substituted the primal RS ansatz
\begin{align*}
  \rho(\vb u)=\E_{\mathcal V}\prod_a\mu_{\mathcal V}(u^a)
\end{align*}
into the equation for $\hat\rho$.  The L\'evy--Khintchine representation then gives \eqref{eq:replica-rhohat-n-zero-form}:
\begin{align*}
  \hat\rho(\vb u)
  &=\kappa d_\alpha\int_0^\infty
    \left[\E\prod_a e^{-\beta\hat{\mathcal W}_{s,g,\beta}(u^a)}-1\right]
    s^{-1-\alpha/2}\,\mathrm ds.
\end{align*}
Campbell's formula turns its exponential into an expectation over a shared marked Poisson field $\mathcal P$.  Writing
$W_l(u):=\hat{\mathcal W}^{(l)}_{\hat s_l,\hat g_l,\beta}(u)$ for its points,
\begin{align*}
  e^{\hat\rho(\vb u)}
  =\E_{\mathcal P}\prod_l\prod_a e^{-\beta W_l(u^a)}
   =\E_{\mathcal P}\prod_a e^{-\beta\hat{\mathcal V}_{\mathcal P}(u^a)},
  \quad \hat{\mathcal V}_{\mathcal P}(u):=\sum_l W_l(u).
\end{align*}
Finally, for a fixed $\mathcal P$, the factors for different replicas separate.

Intuitively, although these steps are involved, the final product is natural.  The primal
RS ansatz makes $\rho$ unchanged by any permutation $\pi$ of the replicas:
\begin{align*}
  \rho(u^{\pi(1)},\ldots,u^{\pi(n)})=\rho(\vb u).
\end{align*}
The equation for $\hat\rho$ treats the replicas in the same way, so we expect $e^{\hat\rho}$ to have this symmetry as well.
As in de Finetti's theorem, this symmetry suggests a common random field given which the replica factors separate.
The Poisson field $\mathcal P$ plays exactly that role above.
In this sense, the product form in
\eqref{eq:replica-campbell-conjugate} is a natural consequence of RS.
In the present heavy-tailed model, obtaining this form required the
L\'evy--Khintchine representation, Campbell's formula, and a marked Poisson
point process.  Other heterogeneous models may require very different and
equally involved calculations.  Even so, we expect their final expressions
to factor over $a$ after conditioning on an appropriate shared random field.
\end{minipage}%
}
\endgroup

\begingroup
\setlength{\fboxsep}{8pt}
\noindent\colorbox{teal!7}{%
\begin{minipage}{\dimexpr\linewidth-2\fboxsep\relax}
\small
\textbf{\color{teal!55!black}Gaussian comparison: Hubbard--Stratonovich decoupling.}
For Gaussian features, the counterpart of the steps from
\eqref{eq:replica-rhohat-n-zero-form} to
\eqref{eq:replica-campbell-conjugate} is
\begin{align*}
  e^{\hat\rho_{\mathrm G}(\vb u)}
  &=\exp\ab\{-\frac{\beta\hat\chi_{\mathrm G}}2
      \sum_a(u^a)^2
      +\frac{\beta^2\hat q_{\mathrm G}}2
      \ab(\sum_a u^a)^2\ab\}\\
  &=\E_{z\sim\mathcal N(0,1)}
    \prod_a\exp\ab\{-\beta\ab(
      \frac{\hat\chi_{\mathrm G}}2(u^a)^2
      +\sqrt{\hat q_{\mathrm G}}\,z u^a\ab)\ab\}\\
  &=\E_z\exp\ab\{-\beta\sum_a\ab(
      \frac{\hat\chi_{\mathrm G}}2(u^a)^2
      +\sqrt{\hat q_{\mathrm G}}\,z u^a\ab)\ab\}.
\end{align*}
The second line uses the Hubbard--Stratonovich identity.  It introduces the
shared Gaussian field $z$ and exposes $\prod_a$: for fixed $z$, each factor
depends on only one $u^a$.  In the present heavy-tailed model, the marked
Poisson point process plays the role of this shared field.  The
L\'evy--Khintchine representation and Campbell's formula together serve the
same purpose as the Hubbard--Stratonovich step: they turn
$e^{\hat\rho(\vb u)}$ into an expectation of a product over replicas, as in
\eqref{eq:replica-campbell-conjugate}.
\end{minipage}%
}
\endgroup

Comparing \eqref{eq:replica-campbell-conjugate} with the conjugate RS ansatz \eqref{eq:replica-conjugate-rs-ansatz}, we identify the random function and its distribution as
\begin{align}
  \hat{\mathcal V}_\beta(u)
  &:=
  \sum_l
  \hat{\mathcal W}^{(l)}_{\hat s_l,\hat g_l,\beta}(u),
  &
  \operatorname{Law}(\hat{\mathcal V}_\beta)
  &=\hat\Pi_\beta.
  \label{eq:replica-finite-temperature-conjugate-update}
\end{align}
The copies indexed by \(l\) are independent, and \(\hat{\mathcal W}^{(l)}_{\hat s_l,\hat g_l,\beta}(0)=0\), so \(\hat{\mathcal V}_\beta(0)=0\).
For each outer Poisson point \(\hat s_l\) let \(\{(s_{lj},g_{lj},\mathcal V_{lj,\beta},U_{lj})\}_{j\geq1},\epsilon_l\) be an independent copy of the sample environment used in \eqref{eq:replica-finite-temperature-sample-contribution}.
Thus \(\mathcal V_{lj,\beta}\) has distribution \(\Pi_\beta\), and, conditional on it, \(U_{lj}\) has Gibbs distribution \(\mu_{\mathcal V_{lj,\beta}}\).
With \(\E_{U_l}\) denoting expectation only over \(\{U_{lj}\}_{j\geq1}\), conditional on the objectives, Poisson points, Gaussian marks, and noise, the complete update is
\begin{equation}
\begin{aligned}
  \mathcal V_{lj,\beta}
  &\overset{\mathrm{i.i.d.}}{\sim}\Pi_\beta,
  \\
  U_{lj}\mid\mathcal V_{lj,\beta}
  &\sim\mu_{\mathcal V_{lj,\beta}},
  \\
  \hat{\mathcal V}_\beta(u)
  &:=-\frac1\beta\sum_l\log
  \frac{
    \E_{U_l}\exp\ab\{-\beta\ell\ab(
      \sqrt{\hat s_l}\hat g_l u
      +\sum_j\sqrt{s_{lj}}g_{lj}U_{lj}
      -\epsilon_l
    \ab)\ab\}
  }{
    \E_{U_l}\exp\ab\{-\beta\ell\ab(
      \sum_j\sqrt{s_{lj}}g_{lj}U_{lj}
      -\epsilon_l
    \ab)\ab\}
  },
  \\
  \operatorname{Law}\ab(\hat{\mathcal V}_\beta\ab)
  &=
  \hat\Pi_\beta.
\end{aligned}
  \label{eq:replica-primal-to-conjugate-law}
\end{equation}
Together, \eqref{eq:replica-conjugate-to-primal-law} and \eqref{eq:replica-primal-to-conjugate-law} are the two finite-temperature RS updates:
\begin{align}
  \hat\Pi_\beta
  &\xrightarrow{\
    \eqref{eq:replica-conjugate-to-primal-law}\ }
  \Pi_\beta,
  \\
  \Pi_\beta
  &\xrightarrow{\
    \eqref{eq:replica-primal-to-conjugate-law}\ }
  \hat\Pi_\beta.
  \label{eq:replica-finite-temperature-coupled-laws}
\end{align}

\begingroup
\setlength{\fboxsep}{8pt}
\noindent\colorbox{blue!5}{%
\begin{minipage}{\dimexpr\linewidth-2\fboxsep\relax}
\small
\textbf{\color{blue!45!black}Summary: Distributional RS saddle-point equation for \(\hat\Pi_\beta\)}
\begin{align}
  \hat\Pi_\beta
  &=
  \operatorname{Law}\ab(
    u\mapsto
    -\frac1\beta\sum_l\log
  \frac{
    \E_{U_l}\exp\ab\{-\beta\ell\ab(
      \sqrt{\hat s_l}\hat g_l u
      +\sum_j\sqrt{s_{lj}}g_{lj}U_{lj}
      -\epsilon_l
    \ab)\ab\}
  }{
    \E_{U_l}\exp\ab\{-\beta\ell\ab(
      \sum_j\sqrt{s_{lj}}g_{lj}U_{lj}
      -\epsilon_l
    \ab)\ab\}
  }
  \ab),\notag
\end{align}
where \(\mathcal V_{lj,\beta}\overset{\mathrm{i.i.d.}}{\sim}\Pi_\beta\), and the expectation over the conditionally independent variables \(U_{lj}\) is
\begin{align}
  \E_{U_l}f(\{U_{lj}\}_j)
  &:=
  \int
  \prod_j\ab[
    \mathrm du_{lj}\,
    \frac{e^{-\beta\mathcal V_{lj,\beta}(u_{lj})}}
         {\int_\R e^{-\beta\mathcal V_{lj,\beta}(v)}\,\mathrm dv}
  \ab]
  f(\{u_{lj}\}_j).
  \notag
\end{align}
\end{minipage}%
}
\endgroup

\subsection{RS saddle-point equations of the
\texorpdfstring{$\hat{\mathcal P}\to\mathcal P$}{hat P to P} update}

Up to this point, the RS order parameters have been the distributions
\(\Pi\) and \(\hat\Pi\) of the random functions \(\mathcal V\) and
\(\hat{\mathcal V}\).
At zero temperature, we will show that the randomness relevant to the
saddle-point equations can instead be represented by
\((Q,\xi,\mathcal R)\) and
\((\hat Q,\zeta,\hat{\mathcal R})\).
We therefore take their joint distributions
\begin{align}
  \mathcal P
  &:=
  \operatorname{Law}(Q,\xi,\mathcal R)\\
  \hat{\mathcal P}
  &:=
  \operatorname{Law}(\hat Q,\zeta,\hat{\mathcal R})
\end{align}
as the new zero-temperature order parameters.
This subsection derives the update
\(\hat{\mathcal P}\to\mathcal P\), and the next subsection derives the
reverse update \(\mathcal P\to\hat{\mathcal P}\).

Fix a candidate input distribution \(\hat{\mathcal P}\), and draw independently
\begin{align}
  (\hat Q,\zeta,\hat{\mathcal R})
  &\sim
  \hat{\mathcal P},
  &
  w^*
  &\sim
  p_w.
\end{align}
We parameterize the zero-temperature conjugate objective by
\begin{align}
  \hat{\mathcal V}(u)
  &:=
  \sqrt{\hat Q}\,\zeta u
  +\hat{\mathcal R}(u).
  \label{eq:replica-coefficient-conjugate-copy}
\end{align}
The next subsection derives this representation from the sample-side update.
The zero-temperature coefficient-side relation
\eqref{eq:replica-finite-temperature-variable-update} gives
\begin{align}
  \mathcal V(u)
  &:=
  \lambda[\gamma(w^*+u)-\gamma(w^*)]
  +\hat{\mathcal V}(u)
  \\
  &=
  \lambda[\gamma(w^*+u)-\gamma(w^*)]
  +\sqrt{\hat Q}\,\zeta u
  +\hat{\mathcal R}(u).
  \label{eq:replica-final-coefficient-objective}
\end{align}
Introducing an external field \(h\), define
\begin{align}
  u_*(h)
  &:=
  \argmin_u\{\mathcal V(u)-hu\}
  \\
  &=
  \argmin_u\ab\{
    \lambda[\gamma(w^*+u)-\gamma(w^*)]
    +\sqrt{\hat Q}\,\zeta u
    +\hat{\mathcal R}(u)
    -hu
  \ab\}.
  \label{eq:replica-final-coefficient-problem}
\end{align}
For the Poisson points, let
\((\mathcal V_j,u_{j,*})\) be independent copies of
\((\mathcal V,u_*)\).
We introduce the following quantities in anticipation of the scalar
reduction in the next subsection, which will show that the remaining
saddle-point equation depends on the random objectives
\(\{\mathcal V_j\}\) drawn from \(\Pi\) only through
\((Q,\xi,\mathcal R)\):
\begin{align}
  Q
  &:=\sum_j s_j u_{j,*}(0)^2,
  \\
  \xi
  &:=\frac1{\sqrt Q}\sum_j\sqrt{s_j}g_j u_{j,*}(0),
  \label{eq:replica-coefficient-gaussian-field}
  \\
  \mathcal R(p)
  &:=\sum_j\ab[
    \min_u\{\mathcal V_j(u)+\sqrt{s_j}g_jpu\}
    -\mathcal V_j(u_{j,*}(0))
    -\sqrt{s_j}g_jp u_{j,*}(0)].
  \label{eq:replica-coefficient-response-definition}
\end{align}
For fixed \(p\), Equation~\eqref{eq:replica-final-coefficient-problem} gives
\begin{align}
  \argmin_u\{\mathcal V_j(u)+\sqrt{s_j}g_jpu\}
  &=\argmin_u\{\mathcal V_j(u)-(-\sqrt{s_j}g_jp)u\}
  =u_{j,*}(-\sqrt{s_j}g_jp).
\end{align}
The envelope identity therefore gives
\begin{align}
  \partial_p\min_u\{\mathcal V_j(u)+\sqrt{s_j}g_jpu\}
  &=\sqrt{s_j}g_j u_{j,*}(-\sqrt{s_j}g_jp).
\end{align}
Differentiating the remaining linear term in
Equation~\eqref{eq:replica-coefficient-response-definition} and summing over
\(j\) yields
\begin{align}
  \mathcal R'(p)
  &=\sum_j\sqrt{s_j}g_j
    [u_{j,*}(-\sqrt{s_j}g_jp)-u_{j,*}(0)],
  &\mathcal R(0)&=\mathcal R'(0)=0.
  \label{eq:replica-coefficient-response}
\end{align}
The function \(\mathcal R\) is concave because each minimized term in its
definition is an infimum of affine functions of \(p\).
Consequently, the input distribution \(\hat{\mathcal P}\) determines the
coefficient-side output distribution through
\begin{align}
  \hat{\mathcal P}
  &\longmapsto
  \mathcal P
  :=
  \operatorname{Law}(Q,\xi,\mathcal R).
  \label{eq:replica-primal-law}
\end{align}

\begingroup
\setlength{\fboxsep}{8pt}
\noindent\colorbox{blue!5}{%
\begin{minipage}{\dimexpr\linewidth-2\fboxsep\relax}
\small
\textbf{\color{blue!45!black}Summary: Distributional RS saddle-point equation for \(\mathcal P\)}
\begin{align}
  \mathcal P
  &=
  \operatorname{Law}(Q,\xi,\mathcal R),
  \notag
\end{align}
where
\begin{align}
  Q
  &=
  \sum_j s_j u_{j,*}(0)^2,\notag\\
  \xi
  &=
  \frac1{\sqrt Q}\sum_j\sqrt{s_j}g_j u_{j,*}(0),
  \notag
  \\
  \mathcal R'(h)
  &=
  \sum_j\sqrt{s_j}g_j
  [u_{j,*}(-\sqrt{s_j}g_jh)-u_{j,*}(0)],
  \notag
\end{align}
with $\mathcal R(0)=0$.
Here, we introduced
\begin{align}
  u_{j,*}(h)
  &=
  \argmin_u\ab\{
    \lambda[\gamma(w_j^*+u)-\gamma(w_j^*)]
    +\sqrt{\hat Q_j}\,\zeta_j u
    +\hat{\mathcal R}_j(u)
    -hu
  \ab\}.
  \notag
\end{align}
where
\((\hat Q_j,\zeta_j,\hat{\mathcal R}_j)
\overset{\mathrm{i.i.d.}}{\sim}\hat{\mathcal P}\) and
\(w_j^*\overset{\mathrm{i.i.d.}}{\sim}p_w\).
\end{minipage}%
}
\endgroup

\subsection{RS saddle-point equations of the
\texorpdfstring{$\mathcal P\to\hat{\mathcal P}$}{P to hat P} update}

Fix the input distribution \(\mathcal P\).
Although the finite-temperature update is written in terms of
\(\Pi_\beta\), the reduction below will show that its zero-temperature
limit depends on the coefficient-side objectives only through
\((Q,\xi,\mathcal R)\sim\mathcal P\).
The finite-temperature update \eqref{eq:replica-primal-to-conjugate-law} draws every \(U_{lj}\) from the normalized Gibbs measure \(\mu_{\mathcal V_{lj,\beta}}\) defined in \eqref{eq:replica-single-site-gibbs-measure}.
We must therefore begin by taking the zero-temperature limit of this normalized density.
For one generic objective \(\mathcal V_{j,\beta}\), its negative logarithm divided by \(\beta\) is
\begin{align}
  -\frac1\beta\log
    \frac{e^{-\beta\mathcal V_{j,\beta}(u)}}
         {\int_\R e^{-\beta\mathcal V_{j,\beta}(v)}\,\mathrm dv}
  &=\mathcal V_{j,\beta}(u)
    +\frac1\beta\log\int_\R e^{-\beta\mathcal V_{j,\beta}(v)}\,\mathrm dv
  \\
  &\xrightarrow{\ \beta\to\infty\ }
    \mathcal V_j(u)-\min_v\mathcal V_j(v)
   =\mathcal V_j(u)-\mathcal V_j(u_{j,*}(0)).
  \label{eq:replica-normalized-gibbs-energy}
\end{align}
The first equality in \eqref{eq:replica-normalized-gibbs-energy} is the negative logarithm of \eqref{eq:replica-single-site-gibbs-measure}.
Its normalizing integral contributes \(-\min_v\mathcal V_j(v)\) in the limit, and \eqref{eq:replica-final-coefficient-problem} with \(h=0\) gives \(\min_v\mathcal V_j(v)=\mathcal V_j(u_{j,*}(0))\).
Thus the normalized Gibbs factor contributes precisely the centered objective appearing on the last line of \eqref{eq:replica-normalized-gibbs-energy}.
The two logarithms in \eqref{eq:replica-finite-temperature-sample-contribution} therefore give
\begin{align}
  \hat{\mathcal W}_{s,g}(u)
  &=\min_{\{u_j\}}
    \ab\{\ell\ab(\sqrt{s}gu+\sum_j\sqrt{s_j}g_ju_j-\epsilon)
      +\sum_j[\mathcal V_j(u_j)-\mathcal V_j(u_{j,*}(0))]\ab\}
  \notag\\
  &\quad-\min_{\{u_j\}}
    \ab\{\ell\ab(\sum_j\sqrt{s_j}g_ju_j-\epsilon)
      +\sum_j[\mathcal V_j(u_j)-\mathcal V_j(u_{j,*}(0))]\ab\}.
  \label{eq:replica-zero-temperature-sample-contribution}
\end{align}
Taking \(\beta\to\infty\) in \eqref{eq:replica-finite-temperature-conjugate-update} gives
\begin{align}
  \hat{\mathcal V}(u)
  &:=
  \sum_l
  \hat{\mathcal W}^{(l)}_{\hat s_l,\hat g_l}(u),
  &
  \operatorname{Law}(\hat{\mathcal V})
  &=
  \hat\Pi.
  \label{eq:replica-zero-temperature-conjugate-update}
\end{align}

We first simplify each sample contribution in Equation~\eqref{eq:replica-zero-temperature-sample-contribution}.
Passing to the dual problem reduces its coupled minimization over \(\{u_j\}\) to a minimization over the scalar residual \(v\).

\begin{proposition}[Scalar reduction of the sample problem]
  \label{prop:replica-scalar-reduction}
Define
\begin{align}
  F(h)
  &:=\min_{\{u_j\}}\ab\{
    \ell\ab(h+\sum_j\sqrt{s_j}g_ju_j-\epsilon)
    +\sum_j[\mathcal V_j(u_j)-\mathcal V_j(u_{j,*}(0))]\ab\}.
\end{align}
We have
\begin{align}
  F(h)
  &=\min_{v\in\R}\ab\{
    \ell(v)+(-\mathcal R)^*(\sqrt Q\xi-\epsilon+h-v)\ab\}.
  \label{eq:replica-factor-duality}
\end{align}
\end{proposition}

\begin{proof}
\begin{align}
  F(h)
  &=\min_{\{u_j\}}\sup_{p\in K_\ell}\ab\{
    (h-\epsilon)p-\ell^*(p)
    +\sum_j[\mathcal V_j(u_j)+\sqrt{s_j}g_jpu_j
             -\mathcal V_j(u_{j,*}(0))]\ab\}
  \\
  &=\sup_{p\in K_\ell}\ab\{
    (h-\epsilon)p-\ell^*(p)
    +\sum_j\ab[\min_u\{\mathcal V_j(u)+\sqrt{s_j}g_jpu\}
                   -\mathcal V_j(u_{j,*}(0))]\ab\}
  \\
  &=\sup_{p\in K_\ell}\ab\{
    (\sqrt Q\xi-\epsilon+h)p-\ell^*(p)+\mathcal R(p)\ab\}
  \\
  &=(\ell^*-\mathcal R)^*(\sqrt Q\xi-\epsilon+h)
  \\
  &=\min_{v\in\R}\ab\{
    \ell(v)+(-\mathcal R)^*(\sqrt Q\xi-\epsilon+h-v)\ab\}.
  \label{eq:replica-factor-dual-value}
\end{align}
Here the first equality uses \(\ell(r)=\sup_{p\in K_\ell}\{pr-\ell^*(p)\}\), the second formally interchanges the minimum and supremum and separates the minimizations over \(u_j\), the third uses the definitions of \(Q\), \(\xi\), and \(\mathcal R\), and the last two use the definition of the convex conjugate and the infimal-convolution identity.
\end{proof}

For each \(l\), independently draw
\((Q_l,\xi_l,\mathcal R_l)\sim\mathcal P\) and
\(\epsilon_l\sim p_\epsilon\), and let \(F_l\) be the corresponding copy
of \(F\).
Applying Proposition~\ref{prop:replica-scalar-reduction} immediately to the
two minimizations in
Equation~\eqref{eq:replica-zero-temperature-sample-contribution} gives
\begin{align}
  \hat{\mathcal W}^{(l)}_{\hat s_l,\hat g_l}(u)
  &=F_l(\sqrt{\hat s_l}\hat g_lu)-F_l(0).
\end{align}

Let \(v_*(h)\) denote a minimizer in Equation~\eqref{eq:replica-factor-duality}, and let \(p\) be the corresponding dual optimizer.
The optimality and envelope identities give
\begin{align}
  v_*(h)
  &=\sqrt Q\xi-\epsilon+h+\mathcal R'(p),
  \quad p\in\partial\ell(v_*(h)).
  \label{eq:replica-scalar-optimality}
  \\
  F'(h)
  &=p=\ell'(v_*(h)).
  \label{eq:replica-sample-envelope}
\end{align}
Let \(v_{l,*}\) be the corresponding copy of \(v_*\).
We use Equation~\eqref{eq:replica-sample-envelope} to separate the linear part of each contribution from its centered remainder:
\begin{align}
  \hat{\mathcal W}^{(l)}_{\hat s_l,\hat g_l}(u)
  &=\int_0^u\sqrt{\hat s_l}\hat g_l
    \ell'(v_{l,*}(\sqrt{\hat s_l}\hat g_lt))\,\mathrm dt
  \\
  &=\sqrt{\hat s_l}\hat g_l\ell'(v_{l,*}(0))u
    +\int_0^u\sqrt{\hat s_l}\hat g_l
    [\ell'(v_{l,*}(\sqrt{\hat s_l}\hat g_lt))
     -\ell'(v_{l,*}(0))]\,\mathrm dt.
  \label{eq:replica-sample-linear-response-split}
\end{align}
Thus the linear part and the centered remainder define
\begin{align}
  \hat Q
  &:=\sum_l\hat s_l\ell'(v_{l,*}(0))^2,
  \quad
  \zeta
  :=\frac1{\sqrt{\hat Q}}
    \sum_l\sqrt{\hat s_l}\hat g_l\ell'(v_{l,*}(0)),
  \label{eq:replica-sample-gaussian-field}
  \\
  \hat{\mathcal R}(u)
  &:=\sum_l\int_0^u\sqrt{\hat s_l}\hat g_l
    [\ell'(v_{l,*}(\sqrt{\hat s_l}\hat g_lt))
     -\ell'(v_{l,*}(0))]\,\mathrm dt,
  \\
  \hat{\mathcal R}'(u)
  &=\sum_l\sqrt{\hat s_l}\hat g_l
    [\ell'(v_{l,*}(\sqrt{\hat s_l}\hat g_lu))
     -\ell'(v_{l,*}(0))].
  \label{eq:replica-sample-response}
\end{align}
Substituting these definitions into Equation~\eqref{eq:replica-zero-temperature-conjugate-update} gives
\begin{align}
  \hat{\mathcal V}(u)
  =\sum_l\hat{\mathcal W}^{(l)}_{\hat s_l,\hat g_l}(u)
  =\sqrt{\hat Q}\,\zeta u+\hat{\mathcal R}(u).
  \label{eq:replica-final-conjugate-potential}
\end{align}

Conditional on \(\{\hat s_l\}\) and the sample problems, \(\zeta\) is
standard Gaussian and independent of \(\hat Q\); on \(\hat Q=0\), take an
independent standard Gaussian.
Consequently, the input distribution \(\mathcal P\) determines the sample-side
output distribution through
\begin{align}
  \mathcal P
  &\longmapsto
  \hat{\mathcal P}
  :=\operatorname{Law}(\hat Q,\zeta,\hat{\mathcal R}).
  \label{eq:replica-conjugate-law}
\end{align}
Together with \(\operatorname{Law}(\hat{\mathcal V})=\hat\Pi\),
Equation~\eqref{eq:replica-final-conjugate-potential} now identifies
\(\hat\Pi\) with the distribution induced by \(\hat{\mathcal P}\).
This proves the parametrization of \(\hat{\mathcal V}\) introduced in the
preceding subsection.
Together with the coefficient-side update
\eqref{eq:replica-primal-law}, this closes the two maps
\(\hat{\mathcal P}\to\mathcal P\) and
\(\mathcal P\to\hat{\mathcal P}\).
These equations reproduce the two distributional updates in
Appendix~\ref{app:self-consistent-equations}.

\begingroup
\setlength{\fboxsep}{8pt}
\noindent\colorbox{blue!5}{%
\begin{minipage}{\dimexpr\linewidth-2\fboxsep\relax}
\small
\textbf{\color{blue!45!black}Summary: Distributional RS saddle-point equation for \(\hat{\mathcal P}\)}
\begin{align}
  \hat{\mathcal P}
  &=
  \operatorname{Law}(\hat Q,\zeta,\hat{\mathcal R}),
  \notag
\end{align}
where
\begin{align}
  \hat Q
  &=
  \sum_l\hat s_l\ell'(v_{l,*}(0))^2,
  \notag
  \\
  \zeta
  &=
  \frac1{\sqrt{\hat Q}}
  \sum_l\sqrt{\hat s_l}\hat g_l\ell'(v_{l,*}(0)),
  \notag
  \\
  \hat{\mathcal R}'(h)
  &=
  \sum_l\sqrt{\hat s_l}\hat g_l
  [\ell'(v_{l,*}(\sqrt{\hat s_l}\hat g_lh))
   -\ell'(v_{l,*}(0))],
  \notag
\end{align}
with \(\hat{\mathcal R}(0)=0\).
Here, we introduced
\begin{align}
  v_{l,*}(h)
  &\in
  \argmin_v\ab(
    \ell(v)
    +(-\mathcal R_l)^*(\sqrt{Q_l}\xi_l-\epsilon_l+h-v)
  \ab).
  \notag
\end{align}
where
\((Q_l,\xi_l,\mathcal R_l)
\overset{\mathrm{i.i.d.}}{\sim}\mathcal P\) and
\(\epsilon_l\overset{\mathrm{i.i.d.}}{\sim}p_\epsilon\), independently of
the Poisson points and Gaussian variables.
\end{minipage}%
}
\endgroup

\begingroup
\setlength{\fboxsep}{8pt}
\noindent\colorbox{teal!7}{%
\begin{minipage}{\dimexpr\linewidth-2\fboxsep\relax}
\small
\textbf{\color{teal!55!black}Intuition: Is the reduction from a multivariable problem to a scalar problem inevitable?}
Initially, the components associated with data row \(l\) must be optimized jointly:
\begin{align*}
  \min_{\{u_{lj}\}}\ab\{
    \ell\ab(h+\sum_j\sqrt{s_{lj}}g_{lj}u_{lj}-\epsilon_l)
    +\sum_j[\mathcal V_{lj}(u_{lj})-\mathcal V_{lj}(u_{lj,*}(0))]
  \ab\}.
\end{align*}
Proposition~\ref{prop:replica-scalar-reduction} reduces this problem to an
optimization over the scalar residual \(v\):
\begin{align*}
  v_{l,*}(h)
  &\in\argmin_v\ab\{
    \ell(v)+(-\mathcal R_l)^*(\sqrt{Q_l}\xi_l-\epsilon_l+h-v)
  \ab\}.
\end{align*}
The heterogeneity across components remains in the random function
\(\mathcal R_l\).  Indeed,
\begin{align*}
  \mathcal R_l'(p)
  &=\sum_j\sqrt{s_{lj}}g_{lj}
    [u_{lj,*}(-\sqrt{s_{lj}}g_{lj}p)-u_{lj,*}(0)],
  \\
  v_{l,*}(h)
  &=\sqrt{Q_l}\xi_l-\epsilon_l+h+\mathcal R_l'(p),
  \quad p\in\partial\ell(v_{l,*}(h)).
\end{align*}
The components respond differently, but their responses are linked through
the same scalar \(p\), and hence through one residual \(v\).  This structure
explains why their joint optimization reduces to a scalar problem.

In other heterogeneous models, we can expect a similar reduction when the
effect of removing one data row is a sum of componentwise objectives and
adding that row couples the components only through one residual.  If the
remaining effective cost instead retains couplings between components and
cannot be written as \(\sum_j\mathcal V_j(u_j)\), as may happen with a
regularizer that keeps neighboring coefficients close, fixing \(p\) does
not separate the componentwise optimizations.  The same reduction to a
scalar problem may then no longer be possible.
\end{minipage}%
}
\endgroup

\subsection{From the RS saddle to observable marginals}

The preceding subsection derived the closed RS saddle equations and the two associated one-dimensional problems.
At that stage, however, \(u_*(0)\) and \(v_*(0)\) are only solutions of the effective coefficient and sample problems; their relation to the original high-dimensional estimator has not yet been established.
We now show that their distributions are, respectively, the limiting one-coordinate distribution of the coefficient error and the limiting empirical distribution of the training residual.
This identifies the one-dimensional saddle problems with observable distributions of the original estimator.

Recall from \eqref{eq:replica-final-coefficient-problem} and \eqref{eq:replica-factor-duality} that the two effective solutions at zero external field are
\begin{align}
  u_*(0)
  &=\argmin_u
  \ab\{
    \lambda[\gamma(w^*+u)-\gamma(w^*)]
    +\sqrt{\hat Q}\,\zeta u
    +\hat{\mathcal R}(u)
  \ab\},
  \\
  v_*(0)
  &\in\argmin_v
  \ab\{
    \ell(v)
    +(-\mathcal R)^*(\sqrt Q\xi-\epsilon-v)
  \ab\}.
  \label{eq:replica-recalled-effective-solutions}
\end{align}
We now identify the first random variable in \eqref{eq:replica-recalled-effective-solutions} with a coefficient error and the second with a training residual.

For a bounded continuous test function \(\phi\), the one-coordinate Gibbs marginal in \eqref{eq:replica-rs-law-closure} gives
\begin{align}
  \E\int_\R\mathrm du\,\phi(u)\mu_{\mathcal V_\beta}(u)
  &=\E\frac{\int_\R\phi(u)e^{-\beta\mathcal V_\beta(u)}\,\mathrm du}
           {\int_\R e^{-\beta\mathcal V_\beta(u)}\,\mathrm du}
  \\
  &\xrightarrow{\ \beta\to\infty\ }
    \E\phi\ab(\argmin_u\mathcal V(u))
   =\E\phi(u_*(0)).
  \label{eq:replica-coefficient-observable}
\end{align}
The minimizer is unique by \eqref{eq:replica-final-coefficient-objective}.
Therefore the coefficient-error conclusion is explicitly
\begin{align}
  \hat w_i-w_i^*
  &\xrightarrow{\mathrm d}
  u_*(0).
  \label{eq:replica-derived-coefficient-law}
\end{align}
This is the distribution stated in \eqref{eq:appendix-limiting-coefficient-law}.

For the training residual, insert the source \(J\phi\) into each loss factor.
At finite \(D\), this multiplies \eqref{eq:replica-partition-function} by the Gibbs expectation
\begin{align}
  Z_{\beta,J}
  &:=Z_\beta\left\langle
      e^{J\sum_\mu\phi(D^{-1/\alpha}\vb x_\mu^\top\vb u-\epsilon_\mu)}
    \right\rangle_\beta,
  \\
  \left.\frac1N\partial_J\log Z_{\beta,J}\right|_{J=0}
  &=\frac1N\sum_\mu
    \left\langle\phi(D^{-1/\alpha}\vb x_\mu^\top\vb u-\epsilon_\mu)
    \right\rangle_\beta
  \\
  &\xrightarrow{\ \beta\to\infty\ }
    \frac1N\sum_\mu\phi\ab(
      D^{-1/\alpha}\vb x_\mu^\top(\hat{\vb w}-\vb w^*)-\epsilon_\mu).
  \label{eq:replica-source-identity}
\end{align}
The same source changes \eqref{eq:replica-factor-replica-product} to
\begin{align}
  \mathcal I_{n,J}[\rho]
  &=\E\ab[\ab(\E_U
    e^{-\beta\ell(\sum_j\sqrt{s_j}g_jU_j-\epsilon)
       +J\phi(\sum_j\sqrt{s_j}g_jU_j-\epsilon)})^n].
  \label{eq:replica-source-factor}
\end{align}
At the saddle, the variations with respect to \(\rho\) and \(\hat\rho\) vanish, so only the explicit source dependence contributes when differentiating \eqref{eq:replica-action}:
\begin{align}
  \left.\frac1N\partial_J\E\log Z_{\beta,J}\right|_{J=0}
  &=\left.\lim_{n\to0}\frac1{\kappa n}
       \partial_J\mathcal S_n[\rho,\hat\rho;J]\right|_{J=0}
   =\left.\lim_{n\to0}\frac1n
       \partial_J\log\mathcal I_{n,J}[\rho]\right|_{J=0}
  \\
  &=\left.\partial_J\E\log\E_U
    e^{-\beta\ell(\sum_j\sqrt{s_j}g_jU_j-\epsilon)
       +J\phi(\sum_j\sqrt{s_j}g_jU_j-\epsilon)}\right|_{J=0}
  \\
  &=\E\frac{\E_U\ab[
      \phi(\sum_j\sqrt{s_j}g_jU_j-\epsilon)
      e^{-\beta\ell(\sum_j\sqrt{s_j}g_jU_j-\epsilon)}]}
      {\E_Ue^{-\beta\ell(\sum_j\sqrt{s_j}g_jU_j-\epsilon)}}
   \xrightarrow{\ \beta\to\infty\ }\E\phi(v_*(0)).
  \label{eq:replica-output-source-calculation}
\end{align}
The last limit uses the residual minimizer in \eqref{eq:replica-factor-duality} at \(h=0\).
Together with \eqref{eq:replica-source-identity}, the saddle prediction is
\begin{align}
  &\frac1N\sum_\mu
    \phi\ab(D^{-1/\alpha}\vb x_\mu^\top(\hat{\vb w}-\vb w^*)-\epsilon_\mu)
  \xrightarrow{\mathrm p}\E\phi(v_*(0)),
  \label{eq:replica-output-observable}\\
  &D^{-1/\alpha}\vb x_\mu^\top\hat{\vb w}-y_\mu
  =D^{-1/\alpha}\vb x_\mu^\top(\hat{\vb w}-\vb w^*)-\epsilon_\mu
   \xrightarrow{\mathrm d}v_*(0).
  \label{eq:replica-derived-training-residual-law}
\end{align}
This proves the training-residual distribution stated in \eqref{eq:appendix-limiting-preactivation-law}.
Combining \eqref{eq:replica-derived-coefficient-law} and \eqref{eq:replica-derived-training-residual-law}, the observable content of the RS saddle is
\begin{equation}
\begin{aligned}
  \hat w_i-w_i^*
  &\xrightarrow{\mathrm d}
  u_*(0),
  \\
  D^{-1/\alpha}\vb x_\mu^\top\hat{\vb w}-y_\mu
  &\xrightarrow{\mathrm d}
  v_*(0).
\end{aligned}
\label{eq:replica-final-observable-laws}
\end{equation}

%% file: section/universality_proof.tex
\section{Derivation of Result~\ref{res:heavy-tail-universality}}
\label{app:universality-proof}

In this section, we derive Result~\ref{res:heavy-tail-universality}.
The replica calculation uses the feature distribution when averaging one
training example in Equation~\eqref{eq:replica-stable-feature-average}.
We show that a general feature distribution satisfying
Equation~\eqref{eq:stable-domain-universality-tail} gives the same average up
to $o_D(1)$.  The remaining saddle-point calculation is then exactly the one
in Appendix~\ref{app:replica-derivation}, which gives the university law.

\begin{proof}[Derivation of Result~\ref{res:heavy-tail-universality}]
  Let the generative function of $X$ as $\phi_X(t)=\E e^{itX}$.
  Symmetry and the normalized tail in
  Equation~\eqref{eq:stable-domain-universality-tail} give
  \begin{align}
    \log\phi_X(t)
    &=
    -\abs t^\alpha+\abs t^\alpha\eta(t) ,
    \label{eq:universality-small-frequency-expansion}
  \end{align}
  where $\eta(t)\longrightarrow0$ as $t\to0$.

  For the replicated empirical measures considered in the saddle calculation,
  assume for each fixed $n$ and $\vb k\in\R^n$ that
  \begin{align}
    D^{-1/\alpha}
    \max_{1\leq i\leq D}\abs{\vb k^\top\vb u_i}
    &\longrightarrow0\\
    \sup_D\frac1D\sum_{i=1}^D
    \abs{\vb k^\top\vb u_i}^\alpha
    &<\infty.
    \label{eq:universality-infinitesimal-replica-array}
  \end{align}
  Independence across feature coordinates and
  Equation~\eqref{eq:universality-small-frequency-expansion} then give
  \begin{align}
    \E\exp\ab(
      iD^{-1/\alpha}\sum_{i=1}^D
      X_i\vb k^\top\vb u_i)
    &=
    \exp\ab(
      -\frac1D\sum_{i=1}^D
      \abs{\vb k^\top\vb u_i}^\alpha
      +o_D(1)
    )
    \\
    &=
    \exp\ab(
      -\int_{\R^n}\rho_D\ab(\mathrm d\vb v)
      \abs{\vb k^\top\vb v}^\alpha
      +o_D(1)
    ).
    \label{eq:universality-asymptotic-stable-feature-average}
  \end{align}
  Indeed, the absolute remainder in the exponent is bounded by
  \begin{align}
    &\max_{1\leq i\leq D}
    \abs{\eta\ab(D^{-1/\alpha}\vb k^\top\vb u_i)}
    \frac1D\sum_{i=1}^D
    \abs{\vb k^\top\vb u_i}^\alpha.
    \notag
  \end{align}
  which vanishes by Equation~\eqref{eq:universality-infinitesimal-replica-array}.
  Equation~\eqref{eq:universality-asymptotic-stable-feature-average} is
  Equation~\eqref{eq:replica-stable-feature-average} up to $o_D(1)$.
  Assuming this error is uniform under the Fourier integral in
  Equation~\eqref{eq:replica-single-factor-fourier}, the one-sample factor
  satisfies
  \begin{align}
    \mathcal I_{n,D}^{P_X}[\rho_D]
    &=
    \mathcal I_n^{S\alpha S(1)}[\rho_D]+o_D(1).
    \label{eq:universality-one-sample-factor-replacement}
  \end{align}
  The other terms in the replica action do not depend on the feature distribution.
  Since $N/D\to\kappa$, Equation~\eqref{eq:universality-one-sample-factor-replacement}
  therefore gives
  \begin{align}
    \mathcal S_n^{P_X}[\rho,\hat\rho]
    &=
    \mathcal S_n^{S\alpha S(1)}[\rho,\hat\rho]+o_D(1).
    \label{eq:universality-replica-action-replacement}
  \end{align}

  From Equation~\eqref{eq:universality-replica-action-replacement} onward,
  the saddle variation, replica-symmetric ansatz, replica limit, and
  zero-temperature limit are exactly those in
  Appendix~\ref{app:replica-derivation}.  The RS solution therefore
  gives Equation~\eqref{eq:stable-coefficient-error-universality}.
\end{proof}

%% file: section/data_rich_scaling_proof.tex
\section{Derivation of Result~\ref{res:data-rich-tail-reliability}}
\label{app:data-rich-scaling-proof}

In this section, we derive Results~\ref{res:data-rich-tail-reliability}.  Write
$C:=C_\alpha^{\mathrm S}$ for the tail constant in
Equation~\eqref{eq:stable-domain-universality-tail}.  The coefficient one-body equation and
the sample-side saddle equations first determine the scale of its solution and
hence the decay of $q_\alpha(\kappa)$.  The stable prediction distribution then gives
the probability of a large prediction error.

For quantities in Appendix~\ref{app:self-consistent-equations}, a subscript
\(\kappa\) below indicates the fixed point at sample ratio \(\kappa\), and
\(u_{\kappa,*}(h)\) denotes one copy of the coefficient solution
\eqref{eq:appendix-input-optimizer-copy} at that fixed point.  At zero external
field, this one-body problem is
\begin{align}
  u_{\kappa,*}\ab(0)
  &=
  \argmin_{u\in\R}
  \ab\{
    \lambda\ab[
      \gamma\ab(w^*+u)-\gamma\ab(w^*)
    \ab]
    +\sqrt{\hat Q_\kappa}\,\zeta_\kappa u
    +\hat{\mathcal R}_\kappa\ab(u)
  \ab\}.
  \label{eq:data-rich-coefficient-one-body}
\end{align}
Its optimality condition is
\begin{align}
  0
  &\in
  \lambda\partial\gamma\ab(w^*+u_{\kappa,*}\ab(0))
  +\sqrt{\hat Q_\kappa}\,\zeta_\kappa
  +\hat{\mathcal R}_\kappa'\ab(u_{\kappa,*}\ab(0)).
  \label{eq:data-rich-coefficient-stationarity}
\end{align}
Thus the scale of the coefficient displacement is determined by the
stationarity equation involving the regularization derivative and the complete data derivative
\(\sqrt{\hat Q_\kappa}\,\zeta_\kappa+
\hat{\mathcal R}_\kappa'(u)\).

First, we evaluate the scale of \(\sqrt{\hat Q_\kappa}\,\zeta_\kappa+
\hat{\mathcal R}_\kappa'(u)\).
Equations~\eqref{eq:appendix-input-linear-field} and
\eqref{eq:appendix-input-random-function} express the two parts of this data
derivative through the same sample-side one-body solutions.  Adding them
cancels their values at zero and gives
\begin{align}
  \sqrt{\hat Q_\kappa}\,\zeta_\kappa
  +\hat{\mathcal R}_\kappa'\ab(u)
  &=
  \sum_{l\geq1}
  \sqrt{\hat s_{\kappa,l}}\,\hat g_l
  \ell'\ab(
    v_{\kappa,l,*}\ab(
      \sqrt{\hat s_{\kappa,l}}\,\hat g_lu
    )
  ).
  \label{eq:data-rich-full-data-derivative}
\end{align}
The Poisson points at ratio \(\kappa\) satisfy
\begin{align}
  \#\ab\{l:\hat s_{\kappa,l}>t\}
  &\overset{\mathrm d}{=}
  \operatorname{Poisson}\ab(
    \frac{2\kappa d_\alpha}{\alpha}t^{-\alpha/2}
  \ab).
  \label{eq:data-rich-poisson-intensity}
\end{align}
If \(\{\hat s_{1,l}\}_{l\geq1}\) denotes the same process at ratio one, then
\begin{align}
  \#\ab\{l:\kappa^{2/\alpha}\hat s_{1,l}>t\}
  &=
  \#\ab\{l:\hat s_{1,l}>\kappa^{-2/\alpha}t\}
  \\
  &\overset{\mathrm d}{=}
  \operatorname{Poisson}\ab(
    \frac{2d_\alpha}{\alpha}
    \ab(\kappa^{-2/\alpha}t)^{-\alpha/2}
  \ab)
  \\
  &=
  \operatorname{Poisson}\ab(
    \frac{2\kappa d_\alpha}{\alpha}t^{-\alpha/2}
  \ab).
\end{align}
Hence the two point processes and their square roots obey
\begin{align}
  \ab\{\hat s_{\kappa,l}\}_{l\geq1}
  &\overset{\mathrm d}{=}
  \ab\{\kappa^{2/\alpha}\hat s_{1,l}\}_{l\geq1},
  \\
  \ab\{\sqrt{\hat s_{\kappa,l}}\}_{l\geq1}
  &\overset{\mathrm d}{=}
  \ab\{\kappa^{1/\alpha}\sqrt{\hat s_{1,l}}\}_{l\geq1}.
  \label{eq:data-rich-poisson-scaling}
\end{align}
Both occurrences of \(\sqrt{\hat s_{\kappa,l}}\) are normalized by evaluating
the derivative at \(\kappa^{-1/\alpha}u\) and dividing its value by
\(\kappa^{1/\alpha}\).  Define this normalized derivative by
\begin{align}
  \Phi_\kappa\ab(u)
  &:={}
  \kappa^{-1/\alpha}
  \ab[
    \sqrt{\hat Q_\kappa}\,\zeta_\kappa
    +\hat{\mathcal R}_\kappa'\ab(\kappa^{-1/\alpha}u)
  \ab]
  \\
  &\overset{\mathrm d}{=}
  \sum_{l\geq1}
  \sqrt{\hat s_{1,l}}\,\hat g_l
  \ell'\ab(
    v_{\kappa,l,*}\ab(
      \sqrt{\hat s_{1,l}}\,\hat g_lu
    )
  ).
  \label{eq:data-rich-rescaled-data-derivative}
\end{align}
The right side contains no explicit powers of \(\kappa\).
It depends on \(\kappa\) only through the sample-side solutions
\(v_{\kappa,l,*}\), whose distribution is determined by the fixed point.
To show that \(\Phi_\kappa\) remains \(O_{\mathbb P}(1)\) on bounded intervals
under Assumption~\ref{asm:loss-regularity}, we consider its two loss classes.

\noindent\textit{(i) Globally Lipschitz loss.}
If \(L_\ell\) is a Lipschitz constant of
\(\ell\), then
\begin{align}
  \abs{
    \ell'\ab(
      v_{\kappa,l,*}\ab(
        \sqrt{\hat s_{1,l}}\,\hat g_lu
      )
    )
  }
  &\leq L_\ell.
  \label{eq:data-rich-bounded-score}
\end{align}

\noindent\textit{(ii) Squared loss.}
Here \(\ell'(v)=v\), and applying monotonicity of \(\partial(-\mathcal R_l)\) to Equation~\eqref{eq:replica-scalar-optimality} at \(h\) and \(0\) gives
\begin{align}
  \abs{v_{\kappa,l,*}(h)-v_{\kappa,l,*}(0)}^2
  &\leq h\ab[v_{\kappa,l,*}(h)-v_{\kappa,l,*}(0)\ab].
\end{align}
Therefore,
\begin{align}
  \abs{
    v_{\kappa,l,*}\ab(
      \sqrt{\hat s_{1,l}}\,\hat g_lu
    )
    -v_{\kappa,l,*}(0)
  }\leq
  \sqrt{\hat s_{1,l}}\abs{\hat g_lu}.
  \label{eq:data-rich-squared-score-increment}
\end{align}
Under the coupling in Equation~\eqref{eq:data-rich-poisson-scaling}, the
finite sample-side scale gives
\begin{align}
  \sum_{l\geq1}
  \hat s_{1,l}v_{\kappa,l,*}(0)^2
  &=
  \kappa^{-2/\alpha}\hat Q_\kappa
  <\infty.
\end{align}
Since \(\sum_l\hat s_{1,l}\hat g_l^2<\infty\) almost surely,
Equation~\eqref{eq:data-rich-squared-score-increment} gives, under this
coupling, for every \(M<\infty\),
\begin{align}
  \sup_{\abs u\leq M}\abs{\Phi_\kappa(u)-\Phi_\kappa(0)}
  &\leq M\sum_{l\geq1}\hat s_{1,l}\hat g_l^2<\infty
  \quad\text{almost surely}.
\end{align}
Thus the variation of \(\Phi_\kappa\) on any bounded interval is bounded
by a random constant independent of \(\kappa\).

For both loss classes in Assumption~\ref{asm:loss-regularity}, the data-rich
RS fixed-point limit gives local-uniform convergence of $\Phi_\kappa$ to a
continuous random monotone function $\Phi$ with an almost surely unique
zero $U\in\R$.
Consequently, for every $M<\infty$,
\begin{align}
  \sup_{\abs u\leq M}\abs{\Phi_\kappa(u)}
  &=O_{\mathbb P}(1).
\end{align}
The limiting zero satisfies
\begin{align}
  0
  <\E\abs U^\alpha
  <\infty.
  \label{eq:data-rich-limiting-root-moment}
\end{align}

\begin{proof}[Derivation of Result~\ref{res:data-rich-tail-reliability}]
  Substituting the definition of \(\Phi_\kappa\) into
  Equation~\eqref{eq:data-rich-coefficient-stationarity} and dividing by
  \(\kappa^{1/\alpha}\) gives
  \begin{align}
    0
    &\in
    \frac{\lambda}{\kappa^{1/\alpha}}
    \partial\gamma\ab(w^*+u_{\kappa,*}(0))
    +\Phi_\kappa\ab(\kappa^{1/\alpha}u_{\kappa,*}(0)).
    \label{eq:data-rich-rescaled-root}
  \end{align}
  The finite convex regularizer has a
  locally bounded subdifferential.  Thus, for every $M<\infty$,
  \begin{align}
    \frac{\lambda}{\kappa^{1/\alpha}}
    \sup_{\abs u\leq M}
    \sup_{s\in\partial\gamma(w^*+\kappa^{-1/\alpha}u)}
    \abs s
    &\longrightarrow0.
    \label{eq:data-rich-regularizer-vanishes}
  \end{align}
  The local-uniform RS limit of $\Phi_\kappa$ makes the data term in
  Equation~\eqref{eq:data-rich-rescaled-root} remain of order one.
  Each $\Phi_\kappa$ is nondecreasing because it is the rescaled subgradient
  of the convex data contribution; its local-uniform limit $\Phi$ is
  therefore also nondecreasing.
  Since $\Phi$ is continuous and monotone with unique zero $U$, for every $\varepsilon>0$,
  \begin{align}
    \Phi\ab(U-\varepsilon)
    &<0<
    \Phi\ab(U+\varepsilon)
    \quad\text{almost surely}.
    \label{eq:data-rich-root-bracket}
  \end{align}
  The locally uniform $o(1)$ regularizer term preserves these signs in Equation~\eqref{eq:data-rich-rescaled-root}.  Monotonicity therefore implies that its zero lies between $U-\varepsilon$ and $U+\varepsilon$ with probability tending to one.
  Hence
  \begin{align}
    \kappa^{1/\alpha}u_{\kappa,*}\ab(0)
    &\xrightarrow{\mathrm d}
    U,
    \label{eq:data-rich-limiting-root}
  \end{align}
  where $U$ is the solution of $\Phi\ab(U)=0$.
  Taking the corresponding RS order-parameter limit gives convergence of the $\alpha$-moments:
  \begin{align}
    \kappa q_\alpha\ab(\kappa)
    &=
    \E\abs{\kappa^{1/\alpha}u_{\kappa,*}(0)}^\alpha
    \longrightarrow
    \E\abs U^\alpha
    =:c_\alpha.
    \label{eq:data-rich-moment-chain}
  \end{align}
  Hence
  \begin{align}
    q_\alpha\ab(\kappa)
    &=
    \frac{c_\alpha}{\kappa}
    +o\ab(\kappa^{-1}).
    \label{eq:data-rich-coefficient-scale-proof}
  \end{align}

  Proposition~\ref{prop:test-prediction-generalization} now gives
  \begin{align}
    \kappa^{1/\alpha}q_\alpha\ab(\kappa)^{1/\alpha}S_\alpha
    &=
    \ab(\kappa q_\alpha\ab(\kappa)\ab)^{1/\alpha}S_\alpha
    \xrightarrow{\mathrm d}
    c_\alpha^{1/\alpha}S_\alpha.
  \end{align}
  Since $\abs{S_\alpha}$ has a finite positive median,
  \begin{align}
    \operatorname{median}\ab(
      q_\alpha\ab(\kappa)^{1/\alpha}\abs{S_\alpha}
    \ab)
    &=
    q_\alpha\ab(\kappa)^{1/\alpha}
    \operatorname{median}\ab(\abs{S_\alpha})
    =
    \Theta\ab(\kappa^{-1/\alpha}).
  \end{align}
  The Gaussian fixed-point equations give $q_2^{\mathrm G}(\kappa)=c_2/\kappa+o(\kappa^{-1})$.
  Since $\abs{S_2}$ also has a finite positive median,
  \begin{align}
    \operatorname{median}\ab(
      q_2^{\mathrm G}\ab(\kappa)^{1/2}\abs{S_2}
    \ab)
    &=
    \Theta\ab(\kappa^{-1/2}).
  \end{align}
  This proves the median statements.

  For every fixed $t>0$, Equation~\eqref{eq:data-rich-coefficient-scale-proof} gives
  \begin{align}
    tq_\alpha\ab(\kappa)^{-1/\alpha}
    &=
    t\ab(\frac{\kappa}{c_\alpha})^{1/\alpha}
    \ab(1+o\ab(1))
    \longrightarrow\infty.
    \label{eq:data-rich-diverging-stable-threshold}
  \end{align}
  The stable-tail relation \(\Pr(\abs{S_\alpha}>x)\sim Cx^{-\alpha}\) can therefore be applied at this threshold:
  \begin{align}
    \Pr\ab(
      q_\alpha\ab(\kappa)^{1/\alpha}\abs{S_\alpha}>t
    )
    &=
    \Pr\ab(
      \abs{S_\alpha}>
      tq_\alpha\ab(\kappa)^{-1/\alpha}
    )
    \\
    &\sim
    Cq_\alpha\ab(\kappa)t^{-\alpha}
    \\
    &\sim
    \frac{Cc_\alpha}{\kappa t^\alpha}.
  \end{align}
  This is Equation~\eqref{eq:data-rich-fixed-threshold-failure}.

  For Gaussian features, write
  \(S_2=\sqrt{2}Z\) with \(Z\sim\mathcal N(0,1)\).
  The standard Gaussian tail asymptotic gives
  \begin{align}
    \Pr\ab(
      q_2^{\mathrm G}\ab(\kappa)^{1/2}\abs{S_2}>t
    )
    &\sim
    \frac{2\sqrt{q_2^{\mathrm G}\ab(\kappa)}}{t\sqrt\pi}
    \exp\ab(-\frac{t^2}{4q_2^{\mathrm G}\ab(\kappa)}).
  \end{align}
  Therefore
  \begin{align}
    \frac1\kappa
    \log\Pr\ab(
      q_2^{\mathrm G}\ab(\kappa)^{1/2}\abs{S_2}>t
    )
    &=
    -\frac{t^2}{4\kappa q_2^{\mathrm G}\ab(\kappa)}
    +O\ab(\frac{\log\kappa}{\kappa})
    \longrightarrow
    -\frac{t^2}{4c_2},
  \end{align}
  which proves the Gaussian statement in Result~\ref{res:data-rich-tail-reliability}.
\end{proof}

%% file: section/bayes_benchmark_proof.tex
\section{Derivation of Result~\ref{res:bayes-optimal-error}}
\label{app:bayes-benchmark}

In this section, we prove Result~\ref{res:bayes-optimal-error}.
We specialize
the two fixed-point updates in Appendix~\ref{app:self-consistent-equations} to
the matched quadratic model and then expand the resulting curvature equations
as \(\kappa\to\infty\).

\subsection{Quadratic specialization of the fixed point}

Equations~\eqref{eq:bayes-posterior-density}--\eqref{eq:bayes-ridge-estimator}
show that the Bayes estimator is the regularized ERM with squared loss,
quadratic regularization, and \(\lambda=\sigma^2/\tau^2\).  We therefore
specialize the two updates in Appendix~\ref{app:self-consistent-equations} to
\begin{align}
  \ell(r)
  &=\frac{r^2}{2},
  \quad
  \gamma(w)=\frac{w^2}{2},
  \quad
  \lambda=\frac{\sigma^2}{\tau^2}.
  \label{eq:bayes-rs-specialization}
\end{align}
Both fixed-point updates preserve the family of centered quadratic functions.
We may therefore write the Bayes fixed point as
\begin{align}
  \mathcal R(h)
  &=
  \frac{\mathcal R''}{2}h^2,
  \quad
  \hat{\mathcal R}(h)
  =
  \frac{\hat{\mathcal R}''}{2}h^2,
  \quad
  \mathcal R''\leq0,
  \quad
  \hat{\mathcal R}''\geq0.
  \label{eq:bayes-quadratic-responses}
\end{align}
Substituting these forms into the two one-body problems and saddle equations
gives their complete quadratic specialization.  On the coefficient side,
Equation~\eqref{eq:appendix-input-optimizer-copy} becomes
\begin{align}
  u_{j,*}(h)
  &=
  \argmin_{u\in\R}
  \ab\{
    \frac{\lambda+\hat{\mathcal R}_j''}{2}u^2
    +\ab(
      \lambda w_j^*+\sqrt{\hat Q_j}\,\zeta_j-h
    \ab)u
  \ab\}
  \\
  &=
  \frac{
    h-\lambda w_j^*-\sqrt{\hat Q_j}\,\zeta_j
  }{
    \lambda+\hat{\mathcal R}_j''
  }.
  \label{eq:bayes-quadratic-coefficient-solution}
\end{align}
In particular,
\begin{align}
  u_{j,*}\ab(-\sqrt{s_j}\,g_jh)-u_{j,*}(0)
  &=
  -\frac{\sqrt{s_j}\,g_jh}{\lambda+\hat{\mathcal R}_j''}.
\end{align}
Equations~\eqref{eq:appendix-output-random-scale}--\eqref{eq:appendix-output-remainder}
therefore reduce to
\begin{align}
  Q
  &=
  \sum_{j\geq1}s_j u_{j,*}(0)^2
  \\
  \sqrt Q\,\xi
  &=
  \sum_{j\geq1}\sqrt{s_j}\,g_j u_{j,*}(0)
  \label{eq:bayes-updated-sample-field}
  \\
  \mathcal R''
  &=
  -\sum_{j\geq1}
  \frac{s_jg_j^2}{\lambda+\hat{\mathcal R}_j''}
  \leq0.
  \label{eq:bayes-coefficient-response-curvature}
\end{align}

On the sample side, for \(\mathcal R_l''<0\),
\((-\mathcal R_l)^*(h)=-h^2/(2\mathcal R_l'')\).  Hence
Equation~\eqref{eq:appendix-output-message} becomes
\begin{align}
  v_{l,*}(h)
  &=
  \argmin_{v\in\R}
  \ab\{
    \frac{v^2}{2}
    -
    \frac{
      \ab(
        \sqrt{Q_l}\,\xi_l-\epsilon_l+h-v
      \ab)^2
    }{
      2\mathcal R_l''
    }
  \ab\}
  \\
  &=
  \frac{
    \sqrt{Q_l}\,\xi_l-\epsilon_l+h
  }{
    1-\mathcal R_l''
  }.
  \label{eq:bayes-quadratic-sample-solution}
\end{align}
This solution satisfies
\begin{align}
  v_{l,*}\ab(\sqrt{\hat s_l}\,\hat g_lh)-v_{l,*}(0)
  &=
  \frac{\sqrt{\hat s_l}\,\hat g_lh}{1-\mathcal R_l''}.
\end{align}
Since \(\ell'(v)=v\),
Equations~\eqref{eq:appendix-input-random-scale}--\eqref{eq:appendix-input-random-function}
reduce to
\begin{align}
  \hat Q
  &=
  \sum_{l\geq1}\hat s_l v_{l,*}(0)^2
  \\
  \sqrt{\hat Q}\,\zeta
  &=
  \sum_{l\geq1}\sqrt{\hat s_l}\,\hat g_l v_{l,*}(0)
  \\
  \hat{\mathcal R}''
  &=
  \sum_{l\geq1}
  \frac{\hat s_l\hat g_l^2}{1-\mathcal R_l''}
  \geq0.
  \label{eq:bayes-sample-response-curvature}
\end{align}
The one-body solutions still contain the coefficient-side field
\(\lambda w^*+\sqrt{\hat Q}\,\zeta\) and the sample-side field
\(\sqrt Q\,\xi-\epsilon\).  Next, we substitute the one-body solutions into
the two linear-field updates.  Equations~\eqref{eq:bayes-quadratic-coefficient-solution}
and~\eqref{eq:bayes-updated-sample-field} give
\begin{align}
  \sqrt Q\,\xi-\epsilon
  &={}-\epsilon
  -\sum_{j\geq1}
  \frac{\sqrt{s_j}\,g_j}{\lambda+\hat{\mathcal R}_j''}
  \ab(
    \lambda w_j^*+\sqrt{\hat Q_j}\,\zeta_j
  \ab),
  \label{eq:bayes-sample-field-substitution}
\end{align}
while Equation~\eqref{eq:bayes-quadratic-sample-solution} and the sample-side
linear-field update give
\begin{align}
  \lambda w^*+\sqrt{\hat Q}\,\zeta
  &=
  \lambda w^*
  +\sum_{l\geq1}
  \frac{\sqrt{\hat s_l}\,\hat g_l}{1-\mathcal R_l''}
  \ab(
    \sqrt{Q_l}\,\xi_l-\epsilon_l
  \ab).
  \label{eq:bayes-coefficient-field-substitution}
\end{align}
The sources \(w^*\) and \(\epsilon\) are independent centered Gaussian
variables, and the two updates are linear.  Hence the two updates preserve
the family of conditionally centered Gaussian fields.  At the fixed point,
their conditional variances satisfy
\begin{align}
  &\operatorname{Var}\ab(
    \sqrt Q\,\xi-\epsilon
    \mathrel{|}
    \{s_j,g_j,\hat{\mathcal R}_j''\}_{j\geq1}
  \ab)=
  \sigma^2
  +
  \sum_{j\geq1}
  \frac{s_jg_j^2}{\ab(\lambda+\hat{\mathcal R}_j''\ab)^2}
  \operatorname{Var}\ab(
    \left.
    \lambda w_j^*+\sqrt{\hat Q_j}\,\zeta_j
    \right|
    \hat{\mathcal R}_j''
  \ab)
  \label{eq:bayes-sample-field-variance-recursion}
  \\
  &\operatorname{Var}\ab(
    \lambda w^*+\sqrt{\hat Q}\,\zeta
    \mathrel{|}
    \{\hat s_l,\hat g_l,\mathcal R_l''\}_{l\geq1}
  \ab)=
  \lambda^2\tau^2
  +
  \sum_{l\geq1}
  \frac{\hat s_l\hat g_l^2}{\ab(1-\mathcal R_l''\ab)^2}
  \operatorname{Var}\ab(
    \left.
    \sqrt{Q_l}\,\xi_l-\epsilon_l
    \right|
    \mathcal R_l''
  \ab).
  \label{eq:bayes-coefficient-field-variance-recursion}
\end{align}
The two curvature equations determine these variance recursions without
additional quantities.  Substituting the sample-side variance
\(\sigma^2(1-\mathcal R'')\) or the coefficient-side variance
\(\sigma^2(\lambda+\hat{\mathcal R}'')\) cancels one power of the corresponding
denominator, and the remaining sum equals the corresponding curvature.  Explicitly,
\begin{align}
  \sigma^2
  +\sum_{j\geq1}
  \frac{s_jg_j^2}{\ab(\lambda+\hat{\mathcal R}_j''\ab)^2}
  \sigma^2\ab(\lambda+\hat{\mathcal R}_j''\ab)
  &=
  \sigma^2\ab(
    1+\sum_{j\geq1}
    \frac{s_jg_j^2}{\lambda+\hat{\mathcal R}_j''}
  \ab)
  \\
  &=
  \sigma^2\ab(1-\mathcal R''),
  \label{eq:bayes-sample-field-variance}
  \\
  \lambda^2\tau^2
  +\sum_{l\geq1}
  \frac{\hat s_l\hat g_l^2}{\ab(1-\mathcal R_l''\ab)^2}
  \sigma^2\ab(1-\mathcal R_l''\ab)
  &=
  \sigma^2\lambda
  +\sigma^2\sum_{l\geq1}
  \frac{\hat s_l\hat g_l^2}{1-\mathcal R_l''}
  \\
  &=
  \sigma^2\ab(\lambda+\hat{\mathcal R}''),
  \label{eq:bayes-coefficient-field-variance}
\end{align}
where we used
Equations~\eqref{eq:bayes-coefficient-response-curvature} and
\eqref{eq:bayes-sample-response-curvature}, together with
\(\lambda^2\tau^2=\sigma^2\lambda\), which follows from
\(\lambda=\sigma^2/\tau^2\).  Thus the conditional distributions reduce to
\begin{align}
  \sqrt Q\,\xi-\epsilon
  \mathrel{|}\mathcal R''
  &\sim
  \mathcal N\ab(
    0,
    \sigma^2\ab(1-\mathcal R''\ab)
  ),
  \label{eq:bayes-preactivation-conditional-law}
  \\
  \lambda w^*+\sqrt{\hat Q}\,\zeta
  \mathrel{|}\hat{\mathcal R}''
  &\sim
  \mathcal N\ab(
    0,
    \sigma^2\ab(\lambda+\hat{\mathcal R}''\ab)
  ).
  \label{eq:bayes-linear-field-closure}
\end{align}
Substituting the second relation into
Equation~\eqref{eq:bayes-quadratic-coefficient-solution} gives
\begin{align}
  u_*(0)
  \mathrel{|}\hat{\mathcal R}''
  &\sim
  \mathcal N\ab(
    0,
    \frac{\sigma^2}{\lambda+\hat{\mathcal R}''}
  ).
  \label{eq:bayes-coefficient-conditional-law}
\end{align}

The coefficient distribution
\eqref{eq:appendix-limiting-coefficient-law},
Equation~\eqref{eq:bayes-coefficient-conditional-law}, and
Equation~\eqref{eq:bayes-coefficient-error-definition} now reduce the
coefficient error to a negative moment of \(\lambda+\hat{\mathcal R}''\).
Indeed,
\begin{align}
  \E\ab[
    \left.\abs{u_*(0)}^\alpha
    \right|\hat{\mathcal R}''
  \ab]
  &=
  \frac{
    2^{\alpha/2}\Gamma\ab((\alpha+1)/2)
  }{
    \sqrt\pi
  }
  \sigma^\alpha
  \ab(\lambda+\hat{\mathcal R}''\ab)^{-\alpha/2}.
\end{align}
Averaging over \(\hat{\mathcal R}''\) gives
\begin{align}
  q_{\alpha,\mathrm B}\ab(\kappa)
  &=
  \frac{
    2^{\alpha/2}\Gamma\ab((\alpha+1)/2)
  }{
    \sqrt\pi
  }
  \sigma^\alpha
  \E\ab(\lambda+\hat{\mathcal R}''\ab)^{-\alpha/2}.
  \label{eq:bayes-coefficient-error-curvature}
\end{align}

\subsection{Data-rich expansion}

It remains to evaluate the negative curvature moment in
Equation~\eqref{eq:bayes-coefficient-error-curvature}.  The marked-Poisson
representation in Equation~\eqref{eq:replica-levy-khintchine} implies that
the signed points \(\sqrt{s_j}g_j\) have total tail constant
\(C_\alpha^{\mathrm S}\).  Define the squared points
\begin{align}
  Y_j:=s_jg_j^2,
  \quad
  \hat Y_l:=\hat s_l\hat g_l^2.
\end{align}
The mapping theorem shows that \(\{Y_j\}_{j\geq1}\) and
\(\{\hat Y_l\}_{l\geq1}\) are positive Poisson processes with respective
L\'evy measures
\begin{align}
  \nu(\mathrm dy)
  =
  \frac{\alpha}{2}C_\alpha^{\mathrm S}
  y^{-1-\alpha/2}\,\mathrm dy,
  \quad
  \hat\nu(\mathrm dy)
  =\kappa\nu(\mathrm dy).
\end{align}
Equation~\eqref{eq:bayes-sample-response-curvature} and the exponential
formula for a marked Poisson process first give
\begin{align}
  \E\exp\ab(-t\hat{\mathcal R}'')
  &=
  \E\exp\ab(
    -t\sum_{l\geq1}
    \frac{\hat Y_l}{1-\mathcal R_l''}
  \ab)
  \\
  &=
  \exp\Bigg\{
    -\kappa\frac{\alpha}{2}C_\alpha^{\mathrm S}
    \int_0^\infty
    \E\ab[
      1-\exp\ab(-\frac{ty}{1-\mathcal R''})
    \ab]
    y^{-1-\alpha/2}\,\mathrm dy
  \Bigg\}.
\end{align}
The integral in this exponent comes from the L\'evy measure.  For fixed
\(r\leq0\) and \(t>0\), the change of variables
\(u=ty/(1-r)\) and integration by parts yield
\begin{align}
  \frac{\alpha}{2}
  \int_0^\infty
  \ab(1-e^{-ty/(1-r)})y^{-1-\alpha/2}\,\mathrm dy
  &=
  \ab(\frac{t}{1-r})^{\alpha/2}\frac{\alpha}{2}
  \int_0^\infty
  \ab(1-e^{-u})u^{-1-\alpha/2}\,\mathrm du
  \\
  &=
  \ab(\frac{t}{1-r})^{\alpha/2}
  \int_0^\infty e^{-u}u^{-\alpha/2}\,\mathrm du
  \\
  &=
  \Gamma\ab(1-\alpha/2)
  \ab(\frac{t}{1-r})^{\alpha/2}.
\end{align}
Substituting this identity at \(r=\mathcal R''\) gives
\begin{align}
  \E\exp\ab(-t\hat{\mathcal R}'')
  &=
  \exp\Bigg\{
    -\kappa C_\alpha^{\mathrm S}\Gamma\ab(1-\alpha/2)
    \E\ab(\frac{t}{1-\mathcal R''})^{\alpha/2}
  \Bigg\}
  \\
  &=
  \exp\Bigg\{
    -\kappa C_\alpha^{\mathrm S}\Gamma\ab(1-\alpha/2)
    \E\ab(1-\mathcal R''\ab)^{-\alpha/2}
    t^{\alpha/2}
  \Bigg\},
  \label{eq:bayes-hatted-curvature-laplace}
\end{align}
Similarly, Equation~\eqref{eq:bayes-coefficient-response-curvature} gives
\begin{align}
  \E\exp\ab(t\mathcal R'')
  &=
  \E\exp\ab(
    -t\sum_{j\geq1}
    \frac{Y_j}{\lambda+\hat{\mathcal R}_j''}
  \ab)
  \\
  &=
  \exp\Bigg\{
    -\frac{\alpha}{2}C_\alpha^{\mathrm S}
    \int_0^\infty
    \E\ab[
      1-\exp\ab(-\frac{ty}{\lambda+\hat{\mathcal R}''})
    \ab]
    y^{-1-\alpha/2}\,\mathrm dy
  \Bigg\}
  \\
  &=
  \exp\Bigg\{
    -C_\alpha^{\mathrm S}\Gamma\ab(1-\alpha/2)
    \E\ab(\lambda+\hat{\mathcal R}''\ab)^{-\alpha/2}
    t^{\alpha/2}
  \Bigg\}.
  \label{eq:bayes-unhatted-curvature-laplace}
\end{align}
We then substitute these Laplace transforms into
\begin{align}
  z^{-\alpha/2}
  &=
  \frac1{\Gamma(\alpha/2)}
  \int_0^\infty
  t^{\alpha/2-1}e^{-tz}
  \,\mathrm dt,
  \qquad z>0.
  \label{eq:bayes-negative-power-mellin}
\end{align}
This gives the two closed moment equations
\begin{align}
  \E\ab(\lambda+\hat{\mathcal R}''\ab)^{-\alpha/2}
  &=
  \frac1{\Gamma(\alpha/2)}
  \int_0^\infty
  t^{\alpha/2-1}
  \exp\Bigg\{
    -\lambda t
    -\kappa C_\alpha^{\mathrm S}\Gamma\ab(1-\alpha/2)
    \E\ab(1-\mathcal R''\ab)^{-\alpha/2}
    t^{\alpha/2}
  \Bigg\}
  \,\mathrm dt,
  \label{eq:bayes-hatted-curvature-moment}
  \\
  \E\ab(1-\mathcal R''\ab)^{-\alpha/2}
  &=
  \frac1{\Gamma(\alpha/2)}
  \int_0^\infty
  t^{\alpha/2-1}
  \exp\Bigg\{
    -t
    -C_\alpha^{\mathrm S}\Gamma\ab(1-\alpha/2)
    \E\ab(\lambda+\hat{\mathcal R}''\ab)^{-\alpha/2}
    t^{\alpha/2}
  \Bigg\}
  \,\mathrm dt.
  \label{eq:bayes-unhatted-curvature-moment}
\end{align}
Using $\lambda=\sigma^2/\tau^2$, making the changes of variables
$u=t^{\alpha/2}$ and $v=t^{\alpha/2}$ in the two respective integrals,
substituting
Equation~\eqref{eq:bayes-unhatted-curvature-moment} into
Equation~\eqref{eq:bayes-hatted-curvature-moment}, and using
Equation~\eqref{eq:bayes-coefficient-error-curvature} gives
Equation~\eqref{eq:bayes-optimal-alpha-error-fixed-point}.

Since \(\hat{\mathcal R}''\geq0\),
\begin{align}
  \E\ab(\lambda+\hat{\mathcal R}''\ab)^{-\alpha/2}
  &\leq
  \lambda^{-\alpha/2}.
\end{align}
Equation~\eqref{eq:bayes-unhatted-curvature-moment} therefore gives the
\(\kappa\)-independent lower bound
\begin{align}
  \E\ab(1-\mathcal R''\ab)^{-\alpha/2}
  &\geq
  \frac1{\Gamma(\alpha/2)}
  \int_0^\infty
  t^{\alpha/2-1}
  \exp\Bigg\{
    -t
    -C_\alpha^{\mathrm S}\Gamma\ab(1-\alpha/2)
    \lambda^{-\alpha/2}t^{\alpha/2}
  \Bigg\}
  \,\mathrm dt
  =:b_\lambda>0.
\end{align}
Substituting this lower bound into
Equation~\eqref{eq:bayes-hatted-curvature-moment} yields
\begin{align}
  \E\ab(\lambda+\hat{\mathcal R}''\ab)^{-\alpha/2}
  &\leq
  \frac1{\Gamma(\alpha/2)}
  \int_0^\infty
  t^{\alpha/2-1}
  \exp\Bigg\{
    -\kappa C_\alpha^{\mathrm S}\Gamma\ab(1-\alpha/2)
    b_\lambda t^{\alpha/2}
  \Bigg\}
  \,\mathrm dt
  \\
  &=
  \frac1{
    \kappa C_\alpha^{\mathrm S}
    \Gamma\ab(1-\alpha/2)
    \Gamma\ab(1+\alpha/2)b_\lambda
  }
  \longrightarrow0.
\end{align}
Dominated convergence in Equation~\eqref{eq:bayes-unhatted-curvature-moment}
then gives
\begin{align}
  \E\ab(1-\mathcal R''\ab)^{-\alpha/2}
  &\longrightarrow1.
  \label{eq:bayes-unhatted-curvature-limit}
\end{align}
Substituting this limit into
Equation~\eqref{eq:bayes-hatted-curvature-moment} and setting
\(t=\kappa^{-2/\alpha}s\) gives
\begin{align}
  \E\ab(\lambda+\hat{\mathcal R}''\ab)^{-\alpha/2}
  &=
  \frac1{\kappa\Gamma(\alpha/2)}
  \int_0^\infty
  s^{\alpha/2-1}
  \notag\\
  &\quad{}\times
  \exp\Bigg\{
    -\lambda\kappa^{-2/\alpha}s
    -C_\alpha^{\mathrm S}\Gamma\ab(1-\alpha/2)
    \E\ab(1-\mathcal R''\ab)^{-\alpha/2}s^{\alpha/2}
  \Bigg\}
  \,\mathrm ds
  \\
  &=
  \frac1{\kappa\Gamma(\alpha/2)}
  \int_0^\infty
  s^{\alpha/2-1}
  \exp\Bigg\{
    -C_\alpha^{\mathrm S}\Gamma\ab(1-\alpha/2)s^{\alpha/2}
  \Bigg\}
  \,\mathrm ds
  +o\ab(\kappa^{-1})
  \\
  &=
  \frac1{
    \kappa C_\alpha^{\mathrm S}
    \Gamma\ab(1-\alpha/2)
    \Gamma\ab(1+\alpha/2)
  }
  +o\ab(\kappa^{-1}).
  \label{eq:bayes-curvature-moment-asymptotic}
\end{align}
Substituting Equation~\eqref{eq:bayes-curvature-moment-asymptotic} into
Equation~\eqref{eq:bayes-coefficient-error-curvature} gives
\begin{align}
  q_{\alpha,\mathrm B}\ab(\kappa)
  &=
  \frac{
    2^{\alpha/2}\Gamma\ab((\alpha+1)/2)
  }{
    \sqrt\pi
  }
  \sigma^\alpha
  \ab[
    \frac1{
      \kappa C_\alpha^{\mathrm S}
      \Gamma\ab(1-\alpha/2)
      \Gamma\ab(1+\alpha/2)
    }
    +o\ab(\kappa^{-1})
  \ab]
  \\
  &=
  \frac{
    2^{\alpha/2}\Gamma\ab((\alpha+1)/2)
  }{
    \alpha\sqrt\pi\,\Gamma\ab(\alpha)
  }
  \frac{\sigma^\alpha}{\kappa}
  +o\ab(\kappa^{-1})
  \\
  &=
  \frac{2^{-\alpha/2}}{\Gamma\ab(1+\alpha/2)}
  \frac{\sigma^\alpha}{\kappa}
  +o\ab(\kappa^{-1}).
\end{align}
The last two equalities use the definition of \(C_\alpha^{\mathrm S}\) and
the Gamma duplication formula.

\begin{comment}
Since $S_\alpha\sim S\alpha S(1)$ has total tail constant
$C_\alpha^{\mathrm S}$,
\begin{align}
  \Pr\ab(
    q_{\alpha,\mathrm B}\ab(\kappa)^{1/\alpha}
    \abs{S_\alpha}>t
  )
  &=
  C_\alpha^{\mathrm S}
  \frac{q_{\alpha,\mathrm B}\ab(\kappa)}{t^\alpha}
  +o\ab(q_{\alpha,\mathrm B}\ab(\kappa))
  \\
  &=
  \frac{C_\alpha^{\mathrm S}2^{-\alpha/2}}{
    \Gamma\ab(1+\alpha/2)
  }
  \frac{\sigma^\alpha}{\kappa t^\alpha}
  +o\ab(\kappa^{-1})
  \\
  &=
  \frac{
    2^{1+\alpha/2}
    \Gamma\ab((\alpha+1)/2)
    \sin\ab(\pi\alpha/2)
  }{
    \alpha\pi^{3/2}
  }
  \frac{\sigma^\alpha}{\kappa t^\alpha}
  +o\ab(\kappa^{-1}).
\end{align}
The last equality uses the definition of $C_\alpha^{\mathrm S}$ and the
Gamma duplication formula.  Requiring the leading term to be at most
\(\delta\) gives
\begin{align}
  N_{\mathrm{req}}=\kappa D
  &\gtrsim
  \frac{
    2^{1+\alpha/2}
    \Gamma\ab((\alpha+1)/2)
    \sin\ab(\pi\alpha/2)
  }{
    \alpha\pi^{3/2}
  }
  \frac{D\sigma^\alpha}{\delta t^\alpha}.
\end{align}
\end{comment}

\section{Derivation of Result~\ref{res:iid-elliptical-bayes-comparison}}
\label{app:bayes-elliptical-comparison}

In this section, we derive Result~\ref{res:iid-elliptical-bayes-comparison} by obtaining the elliptical Bayes fixed point and comparing its coefficient error with that of the entry-wise independent model.

\subsection{Elliptical Bayes fixed point}

Since $x_{\mu i}\sim S\alpha S(1)$,
$\Pr(x_{\mu i}^2>y)\sim C_\alpha^{\mathrm S}y^{-\alpha/2}$.
The normalization in Equation~\eqref{eq:elliptical-radial-normalization}
therefore gives $R_\mu^2$ the positive $\alpha/2$-stable distribution
\begin{align*}
  \E e^{-tR_\mu^2}
  &=\lim_{D\to\infty}
    \ab[\E e^{-tD^{-2/\alpha}x_{\mu 1}^2}\ab]^D
  \\
  &=\exp\ab(
    -C_\alpha^{\mathrm S}\Gamma\ab(1-\alpha/2)t^{\alpha/2}
  )
  =\exp\ab(
    -\frac{2^\alpha\Gamma\ab((\alpha+1)/2)}{\sqrt\pi}
    t^{\alpha/2}
  ).
\end{align*}
The stable-limit exponent follows from the squared-coordinate tail; its
simplification uses the Gamma reflection and duplication formulas.

Let $X^{\mathrm{ell}}$ have rows $(\vb x_\mu^{\mathrm{ell}})^\top$.
Write $\chi$ for the
proportional-limit normalized resolvent
\begin{align}
  \chi
  =\lim_{D\to\infty}\frac1D\operatorname{Tr}
    \ab(
      D^{-2/\alpha}(X^{\mathrm{ell}})^\top X^{\mathrm{ell}}
      +\frac{\sigma^2}{\tau^2}I_D
    )^{-1}.
\end{align}
In our notation, the Gaussian-prior/noise Bayes equations of
\citet{Adomaityte_2024} give the fixed point
\begin{align}
  1
  &=\frac{\sigma^2}{\tau^2}\chi
    +\kappa\left[1-\E\frac{1}{1+\chi R_\mu^2}\right].
\end{align}
Rotational invariance makes the diagonal posterior variances converge to
$\sigma^2\chi$.  Thus, for $Z\sim\mathcal N(0,1)$,
\begin{align}
  q_{\alpha,\mathrm B}^{\mathrm{ell}}
  &=\E|Z|^\alpha\sigma^\alpha\chi^{\alpha/2}
   =\frac{2^{\alpha/2}\Gamma\ab((\alpha+1)/2)}{\sqrt\pi}
    \sigma^\alpha\chi^{\alpha/2}.
\end{align}
Consequently, the prior term becomes
\begin{align}
  \frac{\sigma^2}{\tau^2}\chi
  &=\frac1{\tau^2}
    \left[
      \frac{\sqrt\pi\,q_{\alpha,\mathrm B}^{\mathrm{ell}}}
           {2^{\alpha/2}\Gamma\ab((\alpha+1)/2)}
    \right]^{2/\alpha},
\end{align}
while the Laplace representation of $(1+\chi R_\mu^2)^{-1}$ gives
\begin{align}
  \E\frac1{1+\chi R_\mu^2}
  &=\int_0^\infty
    \exp\ab(
      -s
      -\frac{2^\alpha\Gamma\ab((\alpha+1)/2)}{\sqrt\pi}
       \chi^{\alpha/2}s^{\alpha/2}
    )\,\mathrm ds
  \\
  &=\sigma^2\int_0^\infty
    \exp\ab(
      -\sigma^2u
      -2^{\alpha/2}q_{\alpha,\mathrm B}^{\mathrm{ell}}u^{\alpha/2}
    )\,\mathrm du.
\end{align}
Substituting these two expressions into the known fixed point yields
Equation~\eqref{eq:elliptical-bayes-alpha-error}:
\begin{align}
  \label{eq:first}
  1
  &=
  \frac1{\tau^2}
  \left[
    \frac{
      \sqrt\pi\,
      q_{\alpha,\mathrm B}^{\mathrm{ell}}\ab(\kappa)
    }{
      2^{\alpha/2}\Gamma\ab((\alpha+1)/2)
    }
  \right]^{2/\alpha}
  +
  \kappa
  \left[
    1-\sigma^2
    \int_0^\infty
    \exp\ab(
      -\sigma^2u
      -2^{\alpha/2}
      q_{\alpha,\mathrm B}^{\mathrm{ell}}\ab(\kappa)
      u^{\alpha/2}
    )
    \,\mathrm du
  \right],
\end{align}

\subsection{Comparison with elliptical features}

Finally, we derive Result~\ref{res:iid-elliptical-bayes-comparison}.
Equation~\eqref{eq:first} gives the elliptical coefficient error,
and the coordinate-wise independent coefficient error satisfies
\begin{align}
  \label{eq:second}
  q_{\alpha,\mathrm B}&\ab(\kappa)
  =
  \frac{
    2^{\alpha/2}\Gamma\ab((\alpha+1)/2)\sigma^\alpha
  }{
    \sqrt\pi\,\Gamma\ab(1+\alpha/2)
  }
  \int_0^\infty
  \exp\ab(-\frac{\sigma^2}{\tau^2}u^{2/\alpha})
  \notag\\
  &\qquad\times
  \exp\Bigg\{
    -\frac{
      \kappa 2^\alpha\Gamma\ab((\alpha+1)/2)
    }{
      \sqrt\pi\,\Gamma\ab(1+\alpha/2)
    }
    u
    \int_0^\infty
    \exp\ab(
      -v^{2/\alpha}
      -\frac{
        2^{\alpha/2}q_{\alpha,\mathrm B}\ab(\kappa)
      }{
        \sigma^\alpha
      }v
    )
    \,\mathrm dv
  \Bigg\}
  \,\mathrm du.
\end{align}
Our strategy is to show that the right-hand side of
Equation~\eqref{eq:first} is strictly increasing in the coefficient error
$q$.  We then use Equation~\eqref{eq:second} to show that this right-hand
side, evaluated at $q=q_{\alpha,\mathrm B}\ab(\kappa)$, is strictly less
than $1$.  Since it equals $1$ at
$q=q_{\alpha,\mathrm B}^{\mathrm{ell}}\ab(\kappa)$, the desired ordering
$q_{\alpha,\mathrm B}\ab(\kappa)
<q_{\alpha,\mathrm B}^{\mathrm{ell}}\ab(\kappa)$ follows.

First, we rewrite the right-hand side of Equation~\eqref{eq:first} and show
that it is strictly increasing in $q$.
For any \(q>0\), the second term in Equation~\eqref{eq:first} is
\begin{align}
  &\kappa\left[
    1-\sigma^2
    \int_0^\infty
    \exp\ab(
      -\sigma^2u-2^{\alpha/2}qu^{\alpha/2}
    )
    \,\mathrm du
  \right]\\
  =&
  \kappa\left[
    1-
    \int_0^\infty
    \exp\ab(
      -x-\frac{2^{\alpha/2}q}{\sigma^\alpha}x^{\alpha/2}
    )
    \,\mathrm dx
  \right]\\
  =&
  -\kappa
  \int_0^\infty
  e^{-x}\,
  \frac{\mathrm d}{\mathrm dx}
  \exp\ab(
    -\frac{2^{\alpha/2}q}{\sigma^\alpha}x^{\alpha/2}
  )
  \,\mathrm dx\\
  =&
  \frac{\alpha 2^{\alpha/2}\kappa q}{2\sigma^\alpha}
  \int_0^\infty
  x^{\alpha/2-1}
  \exp\ab(
    -x-\frac{2^{\alpha/2}q}{\sigma^\alpha}x^{\alpha/2}
  )
  \,\mathrm dx\\
  =&
  \frac{2^{\alpha/2}\kappa q}{\sigma^\alpha}
  \int_0^\infty
  \exp\ab(
    -v^{2/\alpha}
    -\frac{2^{\alpha/2}q}{\sigma^\alpha}v
  )
  \,\mathrm dv.
  \label{eq:elliptical-comparison-data-term}
\end{align}
Consequently, Equation~\eqref{eq:first} is equivalently
\begin{align}
  1
  &=
  \frac1{\tau^2}
  \left[
    \frac{
      \sqrt\pi\,
      q_{\alpha,\mathrm B}^{\mathrm{ell}}
    }{
      2^{\alpha/2}\Gamma\ab((\alpha+1)/2)
    }
  \right]^{2/\alpha}
  +
  \frac{
    2^{\alpha/2}\kappa q_{\alpha,\mathrm B}^{\mathrm{ell}}
  }{
    \sigma^\alpha
  }
  \int_0^\infty
  \exp\ab(
    -v^{2/\alpha}
    -\frac{
      2^{\alpha/2}q_{\alpha,\mathrm B}^{\mathrm{ell}}
    }{
      \sigma^\alpha
    }v
  )
  \,\mathrm dv.
  \label{eq:elliptical-comparison-fixed-point}
\end{align}
The first term on the right-hand side of
Equation~\eqref{eq:elliptical-comparison-fixed-point} is strictly increasing
in the coefficient error, and the second term satisfies
\begin{align}
  &\frac{\mathrm d}{\mathrm dq}
  \kappa\left[
    1-
    \int_0^\infty
    \exp\ab(
      -x-\frac{2^{\alpha/2}q}{\sigma^\alpha}x^{\alpha/2}
    )
    \,\mathrm dx
  \right]
  \\
  =&
  \frac{2^{\alpha/2}\kappa}{\sigma^\alpha}
  \int_0^\infty
  x^{\alpha/2}
  \exp\ab(
    -x-\frac{2^{\alpha/2}q}{\sigma^\alpha}x^{\alpha/2}
  )
  \,\mathrm dx
  >0.
\end{align}
Thus the entire right-hand side is strictly increasing in $q$.

We next put Equation~\eqref{eq:second} in the same form.
For \(q_{\alpha,\mathrm B}=q_{\alpha,\mathrm B}(\kappa)\), let \(T\) be the
positive \(\alpha/2\)-stable random variable with
\begin{align}
  \E e^{-xT}
  &=
  \exp\Bigg\{
    -\frac{
      \kappa 2^\alpha\Gamma\ab((\alpha+1)/2)
    }{
      \sqrt\pi\,\Gamma\ab(1+\alpha/2)
    }
    x^{\alpha/2}
    \int_0^\infty
    \exp\ab(
      -v^{2/\alpha}
      -\frac{
        2^{\alpha/2}q_{\alpha,\mathrm B}
      }{
        \sigma^\alpha
      }v
    )
    \,\mathrm dv
  \Bigg\}.
  \label{eq:iid-comparison-stable-variable}
\end{align}
Using this notation gives
\begin{align}
  &\frac{
    \sqrt\pi\,q_{\alpha,\mathrm B}
  }{
    2^{\alpha/2}\Gamma\ab((\alpha+1)/2)\sigma^\alpha
  }\\
  =&
  \frac1{\Gamma\ab(1+\alpha/2)}
  \int_0^\infty
  \exp\ab(-\frac{\sigma^2}{\tau^2}u^{2/\alpha})
  \E e^{-u^{2/\alpha}T}
  \,\mathrm du \\
  =&
  \frac1{\Gamma\ab(\alpha/2)}
  \int_0^\infty
  x^{\alpha/2-1}
  \exp\ab(-\frac{\sigma^2}{\tau^2}x)
  \E e^{-xT}
  \,\mathrm dx
  \\
  = &
  \E\ab(
    \frac{\sigma^2}{\tau^2}+T
  \ab)^{-\alpha/2}.
  \label{eq:iid-comparison-negative-moment}
\end{align}
From the same Laplace transform, we have the following identity: 
\begin{align}
  &\E\frac{T}{\sigma^2/\tau^2+T}\\
  =&
  \int_0^\infty
  \exp\ab(-\frac{\sigma^2}{\tau^2}x)
  \E\ab(Te^{-xT})
  \,\mathrm dx \\
  =&
  -\int_0^\infty
  \exp\ab(-\frac{\sigma^2}{\tau^2}x)
  \frac{\mathrm d}{\mathrm dx}\E e^{-xT}
  \,\mathrm dx
  \\
  =&
  \frac{
    \alpha\kappa 2^\alpha\Gamma\ab((\alpha+1)/2)
  }{
    2\sqrt\pi\,\Gamma\ab(1+\alpha/2)
  }
  \int_0^\infty
  \exp\ab(
    -v^{2/\alpha}
    -\frac{
      2^{\alpha/2}q_{\alpha,\mathrm B}
    }{
      \sigma^\alpha
    }v
  )
  \,\mathrm dv
  \int_0^\infty
  x^{\alpha/2-1}
  \exp\ab(-\frac{\sigma^2}{\tau^2}x)
  \E e^{-xT}
  \,\mathrm dx
  \\
  =&
  \frac{
    \kappa 2^\alpha\Gamma\ab((\alpha+1)/2)
  }{
    \sqrt\pi\,\Gamma\ab(1+\alpha/2)
  }
  \int_0^\infty
  \exp\ab(
    -v^{2/\alpha}
    -\frac{
      2^{\alpha/2}q_{\alpha,\mathrm B}
    }{
      \sigma^\alpha
    }v
  )
  \,\mathrm dv
  \int_0^\infty
  \exp\ab(-\frac{\sigma^2}{\tau^2}u^{2/\alpha})
  \E e^{-u^{2/\alpha}T}
  \,\mathrm du \\
  =&
  \frac{2^{\alpha/2}\kappa q_{\alpha,\mathrm B}}{\sigma^\alpha}
  \int_0^\infty
  \exp\ab(
    -v^{2/\alpha}
    -\frac{
      2^{\alpha/2}q_{\alpha,\mathrm B}
    }{
      \sigma^\alpha
    }v
  )
  \,\mathrm dv.
  \label{eq:iid-comparison-data-moment}
\end{align}
The last equality uses Equation~\eqref{eq:second}.  Therefore
\begin{align}
  1
  &=
  \E\frac{\sigma^2/\tau^2+T}{\sigma^2/\tau^2+T}
  \\
  &=
  \frac{\sigma^2}{\tau^2}
  \E\ab(
    \frac{\sigma^2}{\tau^2}+T
  \ab)^{-1}
  +
  \E\frac{T}{\sigma^2/\tau^2+T}
  \\
  &=
  \frac{\sigma^2}{\tau^2}
  \E\ab(
    \frac{\sigma^2}{\tau^2}+T
  \ab)^{-1}+
  \frac{2^{\alpha/2}\kappa q_{\alpha,\mathrm B}}{\sigma^\alpha}
  \int_0^\infty
  \exp\ab(
    -v^{2/\alpha}
    -\frac{
      2^{\alpha/2}q_{\alpha,\mathrm B}
    }{
      \sigma^\alpha
    }v
  )
  \,\mathrm dv.
  \label{eq:iid-comparison-fixed-point}
\end{align}

We now use Equation~\eqref{eq:iid-comparison-fixed-point} and the strict
monotonicity of \(L^p\) norms for \(0<\alpha/2<1\):
\begin{align}
  1
  &=
  \frac{\sigma^2}{\tau^2}
  \E\ab(
    \frac{\sigma^2}{\tau^2}+T
  \ab)^{-1}
  +
  \frac{2^{\alpha/2}\kappa q_{\alpha,\mathrm B}}{\sigma^\alpha}
  \int_0^\infty
  \exp\ab(
    -v^{2/\alpha}
    -\frac{
      2^{\alpha/2}q_{\alpha,\mathrm B}
    }{
      \sigma^\alpha
    }v
  )
  \,\mathrm dv
  \\
  &>
  \frac{\sigma^2}{\tau^2}
  \left[
    \E\ab(
      \frac{\sigma^2}{\tau^2}+T
    \ab)^{-\alpha/2}
  \right]^{2/\alpha}
  +
  \frac{2^{\alpha/2}\kappa q_{\alpha,\mathrm B}}{\sigma^\alpha}
  \int_0^\infty
  \exp\ab(
    -v^{2/\alpha}
    -\frac{
      2^{\alpha/2}q_{\alpha,\mathrm B}
    }{
      \sigma^\alpha
    }v
  )
  \,\mathrm dv
  \\
  &=
  \frac1{\tau^2}
  \left[
    \frac{
      \sqrt\pi\,q_{\alpha,\mathrm B}
    }{
      2^{\alpha/2}\Gamma\ab((\alpha+1)/2)
    }
  \right]^{2/\alpha}
  +
  \frac{2^{\alpha/2}\kappa q_{\alpha,\mathrm B}}{\sigma^\alpha}
  \int_0^\infty
  \exp\ab(
    -v^{2/\alpha}
    -\frac{
      2^{\alpha/2}q_{\alpha,\mathrm B}
    }{
      \sigma^\alpha
    }v
  )
  \,\mathrm dv.
  \label{eq:iid-elliptical-fixed-point-difference}
\end{align}
The strict inequality follows from the strict monotonicity of \(L^p\) norms
for \(0<\alpha/2<1\).  It is strict
because the Laplace exponent in
Equation~\eqref{eq:iid-comparison-stable-variable} is positive for
\(\kappa>0\), so \(T\) is nonconstant.  The last equality uses
Equation~\eqref{eq:iid-comparison-negative-moment}.

Thus the elliptical right-hand side at \(q_{\alpha,\mathrm B}\) is strictly
less than one.  By the strict monotonicity established above, it equals one
only at a larger coefficient error:
\begin{align}
  q_{\alpha,\mathrm B}^{\mathrm{ell}}\ab(\kappa)
  &>
  q_{\alpha,\mathrm B}\ab(\kappa),
\end{align}
which proves Equation~\eqref{eq:iid-elliptical-bayes-ordering}.

It remains to compare the leading data-rich terms of $q_{\alpha,\mathrm B}^{\mathrm{ell}}(\kappa)$.
Equation~\eqref{eq:elliptical-comparison-fixed-point} first implies
\(q_{\alpha,\mathrm B}^{\mathrm{ell}}(\kappa)\to0\).  Hence
\begin{align}
  \int_0^\infty
  \exp\ab(
    -v^{2/\alpha}
    -\frac{
      2^{\alpha/2}q_{\alpha,\mathrm B}^{\mathrm{ell}}\ab(\kappa)
    }{
      \sigma^\alpha
    }v
  )
  \,\mathrm dv
  &=
  \int_0^\infty e^{-v^{2/\alpha}}\,\mathrm dv+o(1)
  \\
  &=
  \frac{\alpha}{2}
  \int_0^\infty e^{-x}x^{\alpha/2-1}\,\mathrm dx+o(1)
  \\
  &=
  \Gamma\ab(1+\alpha/2)+o(1).
\end{align}
Substitution into the same fixed point gives
\begin{align}
  1
  &=
  \frac1{\tau^2}
  \left[
    \frac{
      \sqrt\pi\,q_{\alpha,\mathrm B}^{\mathrm{ell}}\ab(\kappa)
    }{
      2^{\alpha/2}\Gamma\ab((\alpha+1)/2)
    }
  \right]^{2/\alpha}
  +
  \frac{
    2^{\alpha/2}\kappa
    q_{\alpha,\mathrm B}^{\mathrm{ell}}\ab(\kappa)
  }{
    \sigma^\alpha
  }
  \left[
    \Gamma\ab(1+\alpha/2)+o(1)
  \right]
  \\
  &=
  o(1)
  +
  \frac{
    2^{\alpha/2}\kappa
    q_{\alpha,\mathrm B}^{\mathrm{ell}}\ab(\kappa)
  }{
    \sigma^\alpha
  }
  \left[
    \Gamma\ab(1+\alpha/2)+o(1)
  \right],
\end{align}
which gives
\begin{align}
  q_{\alpha,\mathrm B}^{\mathrm{ell}}\ab(\kappa)
  &=
  \frac{2^{-\alpha/2}}{\Gamma\ab(1+\alpha/2)}
  \frac{\sigma^\alpha}{\kappa}
  +o\ab(\kappa^{-1}).
\end{align}
Together with Equation~\eqref{eq:bayes-optimal-alpha-error}, this proves
Equation~\eqref{eq:iid-elliptical-bayes-leading-order}.